\documentclass[12pt]{article}

\usepackage{amsmath,amsthm,amssymb,amsfonts,authblk,array,apptools,arydshln,algorithm,algorithmic}
\usepackage{bm,booktabs,chngcntr,epsf,enumerate,graphicx}
\usepackage{multirow,natbib,psfrag,paralist,titlesec,subfig,setspace,url}
\usepackage[labelsep=period]{caption}

\graphicspath{{figures/}}
 
\theoremstyle{plain}
\newtheorem{assumption}{Assumption}
\newtheorem{theorem}{Theorem}
\newtheorem{lemma}{lemma}

\theoremstyle{remark}
\newtheorem{experiment}{Experiment}

\newcommand{\argmin}{\operatornamewithlimits{argmin}}

\newcommand{\blind}{0}

\begin{document}

	\def\spacingset#1{\renewcommand{\baselinestretch}%
		{#1}\small\normalsize} \spacingset{1}

	%%%%%%%%%%%%%%%%%%%%%%%%%%%%%%%%%%%%%%%%%%%%%%%%%%%%%%%%%%%%%%%%%%%%%%%%%%%%%%
	
	\if0\blind
	{
		\title{\bf Shape constrained kernel-weighted least squares: Estimating production functions for Chilean manufacturing industries
		\thanks{We thank two anonymous reviewers and the Associate Editor for providing useful suggestions that helped improve this manuscript. We also thank Chris Parmeter, Jeff Racine and Qi Li for their helpful comments.}
		}
		
		\author[1]{Daisuke Yagi}
        \author[2]{Yining Chen}
        \author[1,3]{Andrew L. Johnson}
        \author[4]{Timo Kuosmanen}
        \affil[1]{Texas A\&M University}
        \affil[2]{London School of Economics and Political Science}        \affil[3]{Osaka University}
        \affil[4]{Aalto University}
        \renewcommand\Authands{ and }
		
%		\author{Daisuke Yagi and Andrew L. Johnson\\
%			Department of Industrial and Systems Engineering, Texas A\&M University,\\
%			College Station, TX, USA\\
%			and \\
%			Timo Kuosmanen\\
%			School of Business, Aalto University, Helsinki, Finland}
        \maketitle
	} \fi
	
	\if1\blind
	{
		\bigskip
		\bigskip
		\bigskip
		\begin{center}
			{\LARGE\bf Shape constrained kernel-weighted least squares: Estimating production functions for Chilean manufacturing industries}
		\end{center}
		\medskip
	} \fi

	\bigskip
	\begin{abstract}
		In this paper we examine a novel way of imposing shape constraints on a local polynomial kernel estimator. The proposed approach is referred to as Shape Constrained Kernel-weighted Least Squares (SCKLS). We prove uniform consistency of the SCKLS estimator with monotonicity and convexity/concavity constraints and establish its convergence rate. %We also show that SCKLS is a generalization of Convex Nonparametric Least Squares (CNLS). 
		In addition, we propose a test to validate whether shape constraints are correctly specified. 
		The competitiveness of SCKLS is shown in a comprehensive simulation study. Finally, we analyze Chilean manufacturing data using the SCKLS estimator and quantify production in the plastics and wood industries. The results show that exporting firms have significantly higher productivity.
	\end{abstract}
	
	\noindent%
	%JEL-Classification: C14, D24\\\\
	\vfill

	\newpage
	\spacingset{1.45} % DON'T change the spacing!

	\section{Introduction}
	\label{sec:intro}
	
	Nonparametric regression methods, such as the local linear (LL) estimator, avoid functional form misspecification. %but may suffer from the curse of dimensionality. 
	To model production with a production or a cost function, the flexible nature of nonparametric methods can cause difficulties in interpreting the results. Fortunately, microeconomic theory provides additional structure in the form of shape constraints. Recently several nonparametric shape constrained estimators have been proposed that combine the advantage of avoiding parametric functional specification with improved small sample performance relative to unconstrained nonparametric estimators. Nevertheless, the existing methods have limitations regarding either estimation performance or computational feasibility. %\footnote{These issues are discussed in Section 2 and Appendix A.}.
	In this paper, we propose a new estimator that imposes shape restrictions on local kernel weighting methods. By combining local averaging with shape constrained estimation, we improve finite sample performance by avoiding overfitting.
	
	Work on shape-constrained regression first started in the 1950s with \cite{hildreth1954point}, who studied the univariate regressor case with a least squares objective subject to monotonicity and concavity/convexity constraints. See also \cite{brunk1955maximum} and \cite{grenander1956theory} for alternative shape constrained estimators. Under the concavity/convexity constraint, properties such as consistency, rate of convergence, and asymptotic distribution have been shown by \cite{hanson1976consistency}, \cite{mammen1991nonparametric}, and \cite{groeneboom2001estimation}, respectively. In the multivariate case, \cite{kuosmanen2008representation} developed the characterization of the least squares estimator subject to concavity/convexity and monotonicity constraints, which we will refer to as Convex Nonparametric Least Squares (CNLS) throughout this paper.
	Furthermore, consistency of the least squares estimator was shown independently by \cite{seijo2011nonparametric} and \cite{lim2012consistency}.  

	Regarding the nonparametric estimation implemented using kernel based methods, \cite{birke2007estimating}, \cite{carroll2011testing}, and \cite{hall2001nonparametric} investigated the univariate case and proposed smooth estimators that can impose derivative-based constraints including monotonicity and concavity/convexity. \cite{du2013nonparametric} proposed Constrained Weighted Bootstrap (CWB) by generalizing Hall and Huang's method to the multivariate regression setting. \cite{beresteanu2007nonparametric} developed a similar type of estimator but for use with spline based estimators. %However, CWB still faces significant computational difficulties with global concave/convex constraint. 
	Finally, we mention the work of \cite{li2016nonparametric}, which extended Hall and Huang's method to use the $k$-nearest neighbor approach subject to the monotonicity constraint.
	
	In this paper, \emph{Shape Constrained Kernel-weighted Least Squares} (SCKLS) estimator is described, which optimizes a local polynomial kernel criterion while estimating a multivariate regression function with shape constraints. Under the monotonicity and convex/concavity constraints, we prove uniform consistency and establish the convergence rate of the SCKLS estimator. %Furthermore, we show that the CNLS estimator can be viewed as a special case of the SCKLS estimator when the bandwidth of the kernel approaches zero.
	\cite{kuosmanen2008representation}, \cite{seijo2011nonparametric} and \cite{lim2012consistency} emphasize the potential advantage that CNLS does not require the selection of tuning parameters. Our proposed SCKLS estimator sheds further light on this issue: in the SCKLS framework, CNLS can be seen as the zero bandwidth estimator; we argue that, compared to unrestricted kernel methods, the SCKLS estimator is relatively robust to the bandwidth selected and is able to alleviate well-known issues such as boundary inconsistency faced by the CNLS estimator.%, thus improving its finite-sample performance further.

	Note that with $n$ observations, CNLS imposes $O(n^2)$ concavity/convexity constraints, which can lead to computational difficulties. The number of constraints and the number of variables in the SCKLS estimator do not depend on the number of observations, but rather the number of evaluation points which is arbitrarily defined by the modeler, thereby bring the computational complexity of the estimator largely under control of the modeler. In this paper, we implement an iterative algorithm that reduces the number of constraints by building on the ideas in \cite{lee2013more} to further improve the computational performance. We then validate the performance of the SCKLS estimator via Monte Carlo simulations. For a variety of parameter settings, we find performance of SCKLS to be better or at least competitive with CNLS, CWB, and the local linear estimators. We provide the first simulation study of CWB with global concavity constraints. % We also validate the usefulness of the iterative algorithm to reduce computational time through the simulations. 
	We also investigate the use of variable bandwidth methods that are a function of the data density  \footnote{A variable bandwidth method allows the bandwidth associated with a particular regressor to vary with the density of the data.} 
	and propose variants of a uniform grid as practical ways to further improve the performance of SCKLS.
	%find that these methods lead to relatively flat (very little curvature) functions in case of non-uniform input. Thus, we propose an alternative to estimate function on an adjusted grid.

	Crucially, we also investigate the behavior of SCKLS when the shape constraints are misspecified and propose a hypothesis test to validate the shape constraints imposed. % using the bootstrap method similar to \cite{du2013nonparametric}. 
	Having a test that validates the shape constraints is critical because otherwise our estimation procedure would lead to inconsistent estimates. %We also propose the modified bootstrap test statistics which make our bootstrap procedure significantly efficient. 
	%Furthermore, motivated by \cite{sen2016testing}, we show how to use the SCKLS estimator to test affinity. Again, the usefulness and effectiveness of these two tests are demonstrated through both theory and Monte Carlo simulations.

	Finally, we apply the SCKLS estimator empirically on Chilean manufacturing data from the Chilean Annual Industrial Survey. The estimation results provide a concise description of the supply-side of the Chilean plastic and wood industries as we report marginal productivity, marginal rate of substitution and most productive scale size. We also investigate the impact of exporting on productivity by including additional predictors of output in a semi-parametric model. We find that exporting correlates with higher productivity, thus supporting international trade theories that high productivity firms are more likely to compete in international markets.
		
	Our focus on production functions guides our selection of the polynomial function used in estimation, the data generation processes (DGP) in the Monte Carlo simulations. For the application analyzing the Chilean manufacturing data, we are interested in monotonic and concave shape constraints and use a local linear kernel function. These assumptions are motivated by standard economic theory for production functions \citep{varian1984nonparametric}. However, the methods proposed in the paper are general and applicable for other applications with higher order polynomial functions or alternative shape restrictions, as discussed in Appendix \ref{App:AppendixA}.
	
	The remainder of this paper is as follows. Section \ref{sec:2.ModelEst} describes the model framework and presents our estimator, SCKLS. Section \ref{sec:4.property} contains the statistical properties of the estimator, and Section~\ref{sec:5.test} discusses the behavior of SCKLS under misspecification, as well as a test for concavity and monotonicity. %, and how SCKLS can be used to test whether a function is affine. 
	Monte Carlo simulation results %for a set of three estimator (SCKLS, CNLS and LL) 
	under several different experimental settings are shown in Section \ref{sec:5.simulation}. Section \ref{sec:6.application} applies the SCKLS estimator to estimate a production function for both the Chilean plastics and wood industries. Section \ref{sec:7.conc} concludes and suggests future research directions. Appendix \ref{App:AppendixA} provides extensions to SCKLS and a comparison to CNLS and CWB. Appendix \ref{App:AppendixC} contains all the technical proofs and Appendix \ref{sec:5.2test} describes a test for affinity. Appendix \ref{App:AppendixD} states the details of the iterative algorithm for SCKLS, and Appendix \ref{App:CompResults} presents a more extensive set of simulation results. Appendix \ref{App:AppendixI} describes the details of the partially linear model, and Appendix \ref{App:application} gives further details about the application to the Chilean manufacturing data.

	\section{Model Framework and Methodology}
	\label{sec:2.ModelEst}
	
%	The following section introduces the model framework and describes two existing nonparametric estimation methods with concavity/convexity and monotonicity constraints. 

	\subsection{Model}
	\label{sec:2.1.Model}
	
	Suppose we observe $n$ pairs of input and output data, $\{\bm{X}_j,y_j\}_{j=1}^n$, where for every $j=1,\ldots,n$, $\bm{X}_j=(X_{j1},\ldots,X_{jd} )' \in \mathbb{R}^d$ is a $d$-dimensional input vector, and $y_j \in \mathbb{R}$ is an output. 
	Consider the following regression model
	\[ y_j=g_0(\bm{X}_j)+\epsilon_j, \quad \mbox{ for } j = 1,\ldots,n,\]
	where $\epsilon_j$ is a random variable satisfying $E(\epsilon_j |\bm{X}_j)=0$. Assume that the regression function $g_0: \mathbb{R}^d \rightarrow \mathbb{R}$ belongs to a class of functions, $G$, that satisfies certain shape restrictions. Here our estimator can impose any shape restriction that can be modeled as a lower or upper bound on a derivative. Examples are supermodularity, convexity, monotonicity, and quasi-convexity. %The most general setting is discussed in Appendix~\ref{App:AppendixA}. 
	For purposes of concreteness, and in view of the application to production functions, we focus on imposing monotonicity and global convexity/concavity, specifically, $g_0$ is concave if:
	\[\lambda g_0(\bm{x_1})+(1-\lambda)g_0(\bm{x_2})\leq g_0(\lambda \bm{x_1}+(1-\lambda)\bm{x_2}),\hspace{1cm}\forall\bm{x_1},\bm{x_2}\in \mathbb{R}^d \mbox{ and }\forall\lambda \in [0,1]\]
	Furthermore, saying $g_0$ is monotonically increasing means that
	\[\mbox{if } \bm{x_1}\leq \bm{x_2}\mbox{, then } g_0(\bm{x_1})\leq g_0(\bm{x_2}), \]
	where the inequality of $\bm{x_1}\leq \bm{x_2}$ means that every component of $\bm{x_2}$ is greater than or equal to the corresponding component of $\bm{x_1}$. Here we denote $G_2$ as the set of functions satisfying these constraints.

	\subsection{Shape Constrained Kernel-weighted Least Squares (SCKLS) with Local Linear}
	\label{sec:3.SCKLS}
	%This section explains the proposed estimation method based on the Local Linear and Local Polynomial estimators.
	Given observations $\{\bm{X}_j,y_j\}_{j=1}^n$, we state the (multivariate) local linear kernel estimator developed by \cite{stone1977consistent} and \cite{cleveland1979robust} as
	\begin{equation}
	\begin{aligned}
	\label{eq:4.LL}
	\min_{a,\bm{b}}\sum_{j=1}^{n}(y_j-a-(\bm{X}_j-\bm{x})'\bm{b})^2K\left(\frac{\bm{X}_j-\bm{x}}{\bm{h}}\right),
	\end{aligned}
	\end{equation}	
	where $a$ is a functional estimate, and $\bm{b}$ is an estimate of the slope of the function at $\bm{x}$ with $\bm{x}$ being an arbitrary point in the input space, $K\left(\frac{\bm{X}_j-\bm{x}}{\bm{h}}\right)$ denotes a product kernel, and $\bm{h}$ is a vector of bandwidths (see \cite{racine2004nonparametric} for more detail). We note that the objective function uses kernel weights, so more weight is given to the observations that are closer to the  point $\bm{x}$.
	
	We introduce a set of $m$ points, $\bm{x}_1,\ldots,\bm{x}_m$, for evaluating constraints, which we call evaluation points, and impose shape constraints on the local linear kernel estimator. In the spirit of local linear kernel estimator, we define Shape Constrained Kernel-weighted Least Squares (SCKLS) estimator, for the case of monotonicity and concavity, to be the function $\hat{g}_n:\mathbb{R}^d\rightarrow\mathbb{R}$ such that
	\begin{equation}
	\label{eq:7.Extrapolation}
	\hat{g}_n(\bm{x};\bm{\hat{a}},\bm{\hat{b}})=\min_{i\in\{1,\ldots,m\}}\Big\{\hat{a}_i+(\bm{x}-\bm{x}_i)'\bm{\hat{b}}_i\Big\}
	\end{equation}
	for any $\bm{x}\in\mathbb{R}^d$, where $\hat{\bm{a}}=(\hat{a}_1,\ldots,\hat{a}_m)'$ and $\hat{\bm{b}}=(\bm{\hat{b}}_1',\ldots,\bm{\hat{b}}_m')'$ are the solutions to the following optimization problem
	\begin{equation}
	\begin{aligned}
	\label{eq:6.SCKLS}
	& \min_{\bm{a},\bm{b}}
	& & \sum_{i=1}^{m}\sum_{j=1}^{n}(y_j-a_i-(\bm{X}_j-\bm{x}_i)'\bm{b}_i)^2K\left(\frac{\bm{X}_j-\bm{x}_i}{\bm{h}}\right)\\
	& \mbox{subject to}
	& & a_i-a_l\geq \bm{b}_i'(\bm{x}_i-\bm{x}_l), \; & i,l=1,\ldots,m\\
	&
	& & \bm{b}_i\geq 0, \; & i=1,\ldots,m.
	\end{aligned}
	\end{equation}	
	The first set of constraints in (\ref{eq:6.SCKLS}) imposes concavity and the second set of constraints imposes non-negativity of $\bm{b}_i$ at each evaluation point $\bm{x}_i$. For more details see \cite{kuosmanen2008representation}. Note that (\ref{eq:7.Extrapolation})  implies the functional estimate is constructed by taking the minimum of linear interpolations between the evaluation points. This makes SCKLS a globally shape constrained function although it is a non-smooth piece-wise linear function.
	
% We discuss the impacts of these extensions in section~\ref{sec:5.simulation} where the simulation results are presented.

	The SCKLS estimator requires the user to specify the number and the locations of the evaluation points. A standard method for determining the location of evaluation points, $\{\bm{x}_i\}_{i=1}^m$, is to construct a uniform grid, where each dimension is divided using equal spacing. However, we can address the skewness of input variable distributions common in manufacturing survey data by using a non-uniform grid method, specifically percentile gridding, to specify evaluation points. 
	
	Alternatively, we can deal with the input skewness by applying the $k$-nearest neighbor ($k$-NN) approach, \cite{li2016nonparametric}. The $k$-NN approach uses a smaller bandwidth in dense data regions and a larger bandwidth when the data is sparse. The analysis in Section \ref{sec:6.application} uses both a percentile grid and $k$-NN approach to define the kernel function. For details of these extensions, see Appendix \ref{App:AppendixA}.
	
	As the density of the evaluation points increases, the estimated function potentially has more hyperplane components and is more flexible; however, the computation time typically increases. If a smooth functional estimate is preferred, see \cite{nesterov2005smooth} and \cite{mazumder2015Computational}, where methods for smoothing are provided. In practice, we propose to select the bandwidth vector $\bm{h}$ via the leave-one-out cross-validation based on the unconstrained estimator. See Section~\ref{sec:5.simulation} for the details.
	
	%\footnote{%A piece-wise linear approximation is sometimes believed to be a poor estimator for a smooth underlying function. 
	
	Appendix \ref{App:AppendixA} proposes several alternative implementations of the SCKLS estimator: (1) SCKLS with Local Polynomial approximation, (2) a $k$-nearest neighbor ($k$-NN) approach and (3) non-uniform grid method.

	\section{Theoretical Properties of SCKLS}
	\label{sec:4.property}
	For mathematical concreteness, we next consider the statistical properties of SCKLS under monotonicity and concavity constraints. 
	Recall that $G_2$ is the class of functions which are monotonically increasing and globally concave, and $g_0$ is the truth to be estimated from $n$ pairs of observations. We make the following assumptions: 
	\begin{assumption}
		\leavevmode
		\label{ass:1}
		\begin{enumerate}[(i)]
			\item \label{ass:1.1} $\{\bm{X}_j,y_j\}_{j=1}^\infty$ are a sequence of i.i.d. random variables with  $y_j = g_0(\bm{X}_j)+ \epsilon_j$.
			\item \label{ass:1.2} $g_0\in G_2$ and is twice-differentiable.
			\item \label{ass:1.3} $\bm{X}_j$ follows a distribution with continuous density function $f$ and support $\bm{S}$. Here $\bm{S}$ is a convex, non-degenerate and compact subset of $\mathbb{R}^d$. Moreover, \[\min_{\bm{x} \in \bm{S}} f(\bm{x}) > 0.\]	
			\item \label{ass:1.4} The conditional probability density function of $\epsilon_j$, given $\bm{X}_j$, denoted as $p(e|\bm{x})$, is continuous with respect to both $e$ and $\bm{x}$, with the mean function \[\mu(\cdot) = E(\epsilon_j|\bm{X}_j=\cdot) = 0\] and the variance function \[\sigma^2(\cdot) = \mathrm{Var}(\epsilon_j|\bm{X}_j=\cdot)\] bounded away from 0 and continuous over $\bm{S}$. Moreover, $\sup_{\bm{x} \in \bm{S}}E\Big(\epsilon_j^4\Big|\bm{X}_j = \bm{x}\Big) < \infty$.
	
			\item \label{ass:1.5} $K(\cdot)$ is a non-negative, Lipschitz second order kernel with a compact and convex support. For simplicity, we set the bandwidth associated with each explanatory variable, $h_k$, for $k=1,\ldots,d$, to be $h_1=\cdots=h_d = h$.
			\item \label{ass:1.6} $h =  O(n^{-1/(4+d)})$ as $n \rightarrow \infty$.
		\end{enumerate}
	\end{assumption}

	Here (\ref{ass:1.1}) states that the data are i.i.d.; (\ref{ass:1.2}) says that the constraints we impose on the SCKLS estimator are satisfied by the true function; (\ref{ass:1.3}) makes a further assumption on the distribution of the covariates; (\ref{ass:1.4}) states that the noise can be heteroscedastic in certain ways, but requires the change in the variance to be smooth; (\ref{ass:1.5}) is rather standard in local polynomial estimation to facilitate the theoretical analysis; and (\ref{ass:1.6}) assures the bandwidths become sufficiently small as $n\rightarrow \infty$ so that both the bias and the variance from local averaging go to zero. For details of the consistency of local linear estimator and a discussion of some of these conditions, see \cite{Masry1996multivariatelocal},  \cite{li2007nonparametric} and \cite{fan2016}. 
	
	We consider two scenarios: let the number of evaluation points (denoted by $m$) grow with $n$, or fix the number of evaluation points a priori. For simplicity, we also assume that the evaluation points are drawn independent of $\{\bm{X}_j,y_j\}_{j=1}^n$.
	
	\begin{assumption}
		\leavevmode
		\label{ass:2}
		\begin{enumerate}[(i)]
			\item \label{ass:2.1} The number of evaluation points $m$ $\rightarrow \infty$ as $n \rightarrow \infty$. For simplicity, we assume that the empirical distribution of $\{\bm{x}_1,\ldots,\bm{x}_m\}$ converges to a distribution $Q$ that has support $\bm{S}$ (i.e. as defined in Assumption~\ref{ass:1}(\ref{ass:1.4}))) and a continuous differentiable density function $q:\bm{S} \rightarrow \mathbb{R}$ satisfying $\min_{\bm{x} \in \bm{S}} q(\bm{x}) > 0$.
			\item \label{ass:2.2}  The number of evaluation points $m$ is fixed.  All the evaluation points lie in the interior of $\bm{S}$. Moreover, 
			\[
			\frac{\sup_{\bm{x} \in \bm{S}} \min_{i = 1,\ldots,m} \|\bm{x} - \bm{x}_i\|}{ \min_{i \neq j; i,j \in \{1,\ldots,m\}} \|\bm{x}_j - \bm{x}_i\|} \le \kappa
			\] 
			for some $\kappa \ge 1$ (i.e. $\{\bm{x}_1,\ldots,\bm{x}_m\}$ are reasonably well spread across $\bm{S}$).
		\end{enumerate}
	\end{assumption}
	
	Our main results are summarized below. A short discussion on our proof strategy and the proofs are available in Appendix B. 
	\begin{theorem}
		\label{thm:2.rate}
		Suppose that Assumption \ref{ass:1}(\ref{ass:1.1})-\ref{ass:1}(\ref{ass:1.6}) and Assumption \ref{ass:2}(\ref{ass:2.1}) or  \ref{ass:2}(\ref{ass:2.2})  hold. Then,
		\[
		\frac{1}{m} \sum_{i=1}^m \{\hat{g}_n(\bm{x}_i)- g_0(\bm{x}_i)\}^2 = O(n^{-4/(4+d)} \log n)
		\].
	\end{theorem}

	\begin{theorem}
		\label{thm:1.consistency}
		$\quad$
		 
		\begin{enumerate}[1.]
		\item \textbf{(The case of an increasing $m$)}  Suppose that Assumption \ref{ass:1}(\ref{ass:1.1})-\ref{ass:1}(\ref{ass:1.6}) and Assumption \ref{ass:2}(\ref{ass:2.1}) hold. Let $\bm{C}$ be any fixed closed set that belongs to the interior of $\bm{S}$. Then with probability one, as $n \rightarrow \infty$, the SCKLS estimator satisfies
		\begin{equation*}
		\begin{aligned}
		\sup_{\bm{x} \in \bm{C}} \big|\hat{g}_n(\bm{x}) - g_0(\bm{x}) \big| \rightarrow 0.
		\end{aligned}
		\end{equation*}		
		
		\item \textbf{(The case of a fixed $m$)} Suppose that Assumption \ref{ass:1}(\ref{ass:1.1})-\ref{ass:1}(\ref{ass:1.6}) and Assumption \ref{ass:2}(\ref{ass:2.2}) hold. Then, as $n \rightarrow \infty$, with probability one, the estimates from SCKLS satisfy
		\begin{equation*}
		\begin{aligned}
			\hat{a}_i & \rightarrow g_0(\bm{x}_i) \quad \mbox{ and }\quad		\hat{\bm{b}}_i & \rightarrow \frac{\partial g_0}{\partial \bm{x}}(\bm{x}_i)
		\end{aligned}
		\end{equation*}
		for all $i = 1,\ldots,m$.
		\end{enumerate}
	\end{theorem}

	Note that this convergence rate is nearly optimal (differing only by a factor of $\log n$). However, in the above, we only manage to show that the SCKLS estimator converges at the evaluation points or in the interior of the domain. It is known that shape-constrained estimators tend to suffer from bad boundary behaviors. For instance, the quantity $\sup_{\bm{S}}\big|\hat{g}_n^{CNLS}(\bm{x})-g_0(\bm{x})\big|$ does \emph{not} converge to zero in probability, where $\hat{g}_n^{CNLS}$ is the CNLS estimator. Though for SCKLS, if we let the number of evaluation points, $m$, grow at a rate slower than $n$, we argue that we can both alleviate the boundary inconsistency and improve the computational efficiency. 
	
	%\addtocounter{assumption}{}
	\begin{assumption}
		\leavevmode
		\label{ass:3}		
	The number of evaluation points $m = o(n^{2/(4+d)}/\log n)$ as $n \rightarrow \infty$.
	\end{assumption}

	\begin{theorem}
		\label{thm:3.uconsistency}
		Suppose that Assumption \ref{ass:1}(\ref{ass:1.1})-\ref{ass:1}(\ref{ass:1.6}),  Assumption \ref{ass:2}(\ref{ass:2.1}) and Assumption \ref{ass:3} hold. Then, with probability one, as $n \rightarrow \infty$, the SCKLS estimator satisfies
		\begin{equation*}
		\sup_{\bm{x} \in \bm{S}} \big|\hat{g}_n(\bm{x}) - g_0(\bm{x}) \big| \rightarrow 0.
		\end{equation*}	
	\end{theorem}	
	
%	Since it is known that CNLS tends to overfit the data and lead to inconsistency estimate on the boundary, in view of Theorem~\ref{thm:3.uconsistency}, we have that

	We also note that CNLS can be viewed as a special case of SCKLS when we let the set of evaluation points be $\{\bm{X}_1,\ldots,\bm{X}_n\}$ and the bandwidth vector $\|\bm{h}\|\rightarrow \mathbf{0}$. See Appendix~\ref{App:AppendixA} for the proof of the relationship between CNLS and SCKLS, together with more discussions on the relationship between SCKLS and alternative shape constrained estimators such as CWB.
	
	\section{Shape Misspecification: Theory and Testing}
	\label{sec:5.test}
	\subsection{Misspecification of the shape restrictions}

	So far we have assumed in our estimation procedures that $g_0 \in G_2$, where $G_2$ is the class of functions which are monotonically increasing and globally concave. To  understand the behavior of SCKLS, we are interested in its performance when $g_0 \notin G_2$.
	
	Let $Q$ be a distribution on $\bm{S}$ (as in Assumption~\ref{ass:2}(\ref{ass:2.1})) and define $g^*: \bm{S} \rightarrow \mathbb{R}$ as 
	\[
	g_0^* := \argmin_{g \in G_2} \int_{\bm{S}}\{g(\bm{x})-g_0(\bm{x})\}^2 Q(d\bm{x}).
	\]
	The existence and $Q$-uniqueness of $g_0^*$ follows from the well-known results about the projection onto a cone in the Hilbert space. When $g_0 \in G_2$, it is easy to check that $g_0^* = g_0$. See also \citet{lim2012consistency}.  The following result can be viewed as a generalization of Theorem~\ref{thm:1.consistency}.
	
	\begin{theorem}
		\label{thm:4.missconsistency}
		$\quad$
		 
        Suppose that Assumption \ref{ass:1}(\ref{ass:1.1}), \ref{ass:1}(\ref{ass:1.3})-\ref{ass:1}(\ref{ass:1.6}) and Assumption \ref{ass:2}(\ref{ass:2.1}) hold. Furthermore, suppose that $g_0$ is twice-differentiable. Let $\bm{C}$ be any compact set that belongs to the interior of $\bm{S}$. Then with probability one, as $n \rightarrow \infty$, the SCKLS estimator satisfies
		\begin{equation*}
		\begin{aligned}
		\sup_{\bm{x} \in \bm{C}} \big|\hat{g}_n(\bm{x}) - g_0^*(\bm{x}) \big| \rightarrow 0.
		\end{aligned}
		\end{equation*}		
	\end{theorem}
	 Theorem \ref{thm:4.missconsistency} assures us that the SCKLS estimator converges uniformly on a compact set to the function $g_0^*$ that is closest in $L^2$ distance to the true function $g_0$ for which our estimator is misspecified. Consequently, as long as $g_0$ is not too far away from $G_2$, our estimator can still be used as a reasonable approximation to the truth, especially when the sample size is moderate. See Appendix~\ref{App:CompResults} for a numerical demonstration.

	\subsection{Hypothesis Testing for the Shape}
	\label{sec:5.1test}
	Admittedly, the SCKLS estimator can be inappropriate if the shape constraints are not fulfilled by $g_0$. Thus, we propose a procedure based on the SCKLS estimators for testing
	\[
	H_0:\;\; \{g_0: \bm{S} \rightarrow \mathbb{R}\}\in G_2 \quad \mbox{against} \quad H_1:\;\; \{g_0: \bm{S} \rightarrow \mathbb{R}\}\notin G_2. 
	\]
    
    Denote by
    \begin{align*}
	\tilde{r}^2\Big(\{\bm{X}_j,y_j\}_{j=1}^n, \{\bm{x}_i\}_{i=1}^m\Big) &= 
		 \min_{\bm{a},\bm{b}} \sum_{i=1}^{m}\sum_{j=1}^{n}(y_j-a_i-(\bm{X}_j-\bm{x}_i)'\bm{b}_i)^2K\left(\frac{\bm{X}_j-\bm{x}_i}{\bm{h}}\right);
    \end{align*}
    the value of the objective function that is minimized by the local linear kernel estimator. And denote by
    \begin{align*}
	\hat{r}^2\Big(\{\bm{X}_j,y_j\}_{j=1}^n, \{\bm{x}_i\}_{i=1}^m\Big) &= 
		 \min_{\bm{a},\bm{b}} \sum_{i=1}^{m}\sum_{j=1}^{n}(y_j-a_i-(\bm{X}_j-\bm{x}_i)'\bm{b}_i)^2K\left(\frac{\bm{X}_j-\bm{x}_i}{\bm{h}}\right), \\
		 	& \quad \mbox{subject to} \quad a_i-a_l\geq \bm{b}_i'(\bm{x}_i-\bm{x}_l) \mbox{ and }  \bm{b}_i\geq 0, \; i,l=1,\ldots,m.
    \end{align*}
    Here $\hat{r}^2(\cdot,\cdot)$ is the value of the objective function that is minimized by SCKLS.
    
	We focus on the test statistic  	 	
	\begin{align*}	 	
	T_n := T\Big(\{\bm{X}_j,y_j\}_{j=1}^n, \{\bm{x}_i\}_{i=1}^m\Big) &= \Big[\frac{1}{mnh^d}\Big\{\hat{r}^2\Big(\{\bm{X}_j,y_j\}_{j=1}^n, \{\bm{x}_i\}_{i=1}^m\Big) - \tilde{r}^2\Big(\{\bm{X}_j,y_j\}_{j=1}^n, \{\bm{x}_i\}_{i=1}^m\Big)\Big\}\Big]^{1/2},
    \end{align*}
    which is a re-scaled version of the difference between the values of the same objective function (with the same bandwidth $\bm{h}$), optimized either with or without the shape constraints. Intuitively, the value of this statistic should be small if $g_0 \in G_2$. This statistic can also be viewed as a smoothed and re-scaled version of the goodness-of-fit statistic.
    
    Here we focus on the boundary case  when $g_0$ is constant (i.e. $g_0 = 0$) because it is hardest to evaluate the null hypothesis when $g_0$ is both non-increasing and non-decreasing and both concave and convex, intuitively and theoretically and it allows us to control the size of our test statistic. Since the noise here might be non-homogeneous, we use the wild bootstrap to approximate the distribution of the test statistic under $H_0$. See \citet{wu1986}, \citet{liu1988}, \citet{mammen1993} and \citet{davidson2008} for an overview of the wild bootstrap procedure.

	Our testing procedure has three steps:
	\begin{enumerate}
		\item Estimate the error at each $\bm{X}_j$ by $\tilde{\epsilon_j} = y_j-\tilde{g}_n(\bm{X}_j)$ for $j = 1,\ldots,n$, where $\tilde{g}$ is the unconstrained local linear estimator with kernel and bandwidth satisfying Assumptions \ref{ass:1}(\ref{ass:1.5})--(\ref{ass:1.6}).
		
		\item  The wild bootstrap method is used to construct a critical region for $T_n$. Let $B$ be the number of Monte Carlo iterations. For every $k = 1,\ldots,B$, let $\bm{u}_k = (u_{1k},\ldots,u_{nk})'$ be a random vector with components sampled independently from the Rademacher distribution, i.e. $P(u_{jk} = 1) = P(u_{jk} = -1) = 0.5$. Furthermore, let $y_{jk} = u_{jk}\ \tilde{\epsilon}_j$.  Then, the wild bootstrap test statistic is
		
    \[
    T_{nk}= T\Big(\{\bm{X}_j, y_{jk}\}_{j=1}^n, \{\bm{x}_i\}_{i=1}^m\Big).
    \]
	\item Define the Monte Carlo $p$-value as\footnote{Since we underestimate the level of the errors in Step 1 by a factor of roughly $n^{-2/(4+d)}$, for the theoretical development, we address this bias issue by modifying the $p$-value to be $p_n = \frac{1}{B} \sum_{k=1}^B \mathbf{1}_{\{T_n \le T_{nk} + \Delta_n\}}$, where  $\Delta_n = O(n^{-2/(4+d)} \log n)$. Note that if we fix $m$ and pick $h = O(n^{-\eta})$ for $\eta \in (\frac{1}{4+d},\frac{1}{d})$, then $\Delta_n/ T_{nk} = o_p(1)$ as $n \rightarrow \infty$, i.e. this correction has a negligible effect.  Indeed, our experience suggests that this modification offers little improvement in terms of finite sample performance in our simulation study.} 
\[
p_n = \frac{1}{B} \sum_{k=1}^B \mathbf{1}_{\{T_n \le T_{nk} \}}.
\]
For a test of size $\alpha \in (0,1)$, we reject $H_0$ if $p_n < \alpha$. 
\end{enumerate}	

A few remarks are in order. 

First, here we conveniently implemented the simplest wild bootstrap scheme to simplify our analysis, in line with the work of \citet{davidson2008}. Instead of imposing the Rademacher distribution on $u_{kj}$, we can also use any distribution with zero-mean and unit-variance. One popular choice suggested by \citet{mammen1993} is
\[
u_{jk}=\left\{{\begin{matrix}-\frac{{\sqrt {5}}-1}{2}&{\mbox{with probability }}\frac{5+\sqrt {5}}{10}\\\frac{{\sqrt {5}}+1}{2}&{\mbox{with probability }}\frac{5-\sqrt {5}}{10}\end{matrix}}\right..
\]

Second, note that the definition of $y_{jk}$ in Step~2  makes this a test of the residuals, i.e.,  when drawing bootstrap samples, we use $y_{jk} =  u_{jk}\ \tilde{\epsilon}_j$ instead of $y_{jk} = \hat{g}_n(\bm{X}_j)+ u_{jk}\ \tilde{\epsilon}_j$. From this perspective, our test is similar to the univariate monotonicity test in \cite{hall2000testing}. One reason behind this choice is to avoid the boundary inconsistency of the bootstrap procedure. See \citet{Andrews2000} and \citet{GNA2017} who addressed this issue in a much simpler setup.  Generally speaking, 
testing the null hypothesis becomes harder when $g_0$ is on the boundary of $G_2$.  
In practice, we could use $y_{jk} = \hat{g}_n(\bm{X}_j)+ u_{jk}\ \tilde{\epsilon}_j$ in certain scenarios (e.g. when testing $g_0$ is a strictly increasing and strictly concave function against $g_0 \notin G_2$), and slight improvements are observed in terms of finite-sample performance.

We now look into the theoretical properties of our procedure under both $H_0$ and $H_1$.
%in three different scenarios:
%\begin{enumerate}[(A)]
%    \item ``Interior'': $g_0 \in G_2$ and $g_0$ is strictly concave;
%    \item ``Boundary'': $g_0 \in G_2$ and $g_0$ is linear;
%    \item ``Misspecification'': $g_0 \notin G_2$.
%\end{enumerate}
See Appendix \ref{App:AppendixC} for the proof. 
	\begin{theorem}
		\label{thm:5.SCKLS-TEST1}
	 Suppose that Assumptions \ref{ass:1}(\ref{ass:1.1}),(\ref{ass:1.3})--(\ref{ass:1.5}) and  \ref{ass:2}(\ref{ass:2.1})  hold, and the conditional error distribution (i.e. $\epsilon_j|\bm{X}_j$) is symmetric. Furthermore, assume that $g_0$ is continuously twice-differentiable and let $h = O(n^{-\eta})$ for some fixed $\eta \in (\frac{1}{4+d},\frac{1}{d})$. Let $B := B(n) \rightarrow \infty$ as $n \rightarrow \infty$. Then, for any given $\alpha \in (0,1)$, 
	\begin{itemize}[--]
    \item Type I error: for any $g_0 \in G_2$,  $\limsup_{n \rightarrow \infty} P(p_n < \alpha) \le \alpha$; 
	%\item Type II error: for any $g_0 \notin G_2$, with a sufficiently large $m$, $1 - P(p_n < \alpha) \rightarrow 0$. 
	\item Type II error: for any $g_0 \notin G_2$, $\limsup_{n \rightarrow \infty} \Big\{ 1 - P(p_n < \alpha) \Big\} = 0$. 
	\end{itemize}
	In addition, if we replace Assumption \ref{ass:2}(\ref{ass:2.1}) by Assumption \ref{ass:2}(\ref{ass:2.2}), the same conclusions hold for sufficiently large $m$.
	\end{theorem}
	
See also Section~\ref{sec:5.simulation} for the finite-sample performance of our test in a simulation study, where we demonstrate that the proposed test controls both Type I and Type II errors reasonably well. Additionally, Appendix \ref{sec:5.2test} describes our procedure for testing affinity using SCKLS. 
    
%\iffalse	
	\section{Simulation study}
	\label{sec:5.simulation}
	\subsection{Numerical experiments on estimation}
	\subsubsection{The setup}
	We now examine the finite sample performance and robustness of the proposed estimator through Monte Carlo simulations. We run our experiments on a computer with Intel Core2 Quad CPU 3.00 GHz and 8GB RAM. We compare the performance of SCKLS is compared with that of CNLS and LL. See Appendix~\ref{App:CompResults} for a comparisons of SCKLS with CWB. For the SCKLS and the CNLS estimator, we solve the quadratic programming problems with MATLAB using the built-in quadratic programming solver, \texttt{quadprog}. We run two sets of experiments varying the number of observations ($n$), the number of evaluation points ($m$), and the number of the inputs ($d$). %Experiment \ref{exp:1} considers uniformly distributed input variables with a low noise level. Experiment \ref{exp:2} increases the noise level in the data generation process (DGP). %Experiment \ref{exp:3} considers non-uniform distributed input variables to validate the robustness and the benefits of the extension of SCKLS: variable bandwidth with $k$-NN approach and non-uniform grid. 
	%Experiment \ref{exp:4} considers uniform input and vary the number of evaluation points to assess how the number of evaluation points affect the performance and computational difficulty of SCKLS. Experiment \ref{exp:5} considers alternative concave/monotonic function which is generalized McFadden function to test the robustness of the SCKLS estimator for different shape of functions. 
	We also run additional experiments to show the robust performance of the SCKLS estimator under alternative conditions. See Appendix~\ref{App:CompResults} for the results. %In Appendix~\ref{App:CompResults}, We also implemented the CWB estimator with global concavity where these results are not available in the existing paper.
	
	We measure the estimator's performance using Root Mean Squared Errors (RMSE) based on two criteria: the distance from the estimated function to the true function measured 1) at the observed points and 2) at the evaluation points constructed on an uniform grid %\footnote{The evaluation points are defined by the equally spaced percentile of the observed data.}
	, respectively. As CNLS estimates hyperplanes at observation points, we use linear interpolation to obtain the RMSE of CNLS\footnote{The CNLS estimates include the second stage linear programming estimation procedure described in \cite{kuosmanen2012stochastic} to find the minimum extrapolated production function.}. %We also measure how the iterative algorithm for SCKLS helps alleviate computational difficulty by reporting the percentage of constraints which are used to obtain the optimal solution and the computational time in seconds. We do not use the constraint reduction algorithm for CNLS. 
	We replicate each scenario 10 times and report the average and standard deviation.
	
	\subsubsection{Choosing of the tuning parameters}
	For the SCKLS estimator, we use the Gaussian kernel function $K(\cdot)$ and leave-one-out cross-validation (LOOCV) for bandwidth selection. LOOCV is a data-driven method, and has been shown to perform well for unconstrained kernel estimators such as local linear \citep{stone1977consistent}. We apply LOOCV procedure on unconstrained estimates (i.e. local linear) to select the bandwidth for SCKLS to reduce the computational burden and because SCKLS is relatively insensitive to the bandwidth choice (see for example Section~\ref{sec:5.1.exp1}). For further computational improvements, we apply the iterative algorithm described in Appendix \ref{App:AppendixD}.	
		
	\subsubsection{Results}
	\paragraph{Fixed number of evaluation points}
	\label{sec:5.1.exp1}
	\begin{experiment}
	\label{exp:1}
	We consider a Cobb--Douglas production function with $d$-inputs and one-output, $g_0(x_1,\ldots,x_d)=\prod_{k=1}^{d}x_k^{\frac{0.8}{d}}$. For each pair $(\bm{X}_j,y_j)$, each component of the input, $\bm{X}_{jk}$, is randomly and independently drawn from uniform distribution $unif[1,10]$, and the additive noise, $\epsilon_j$, is randomly sampled from a normal distribution, $N(0,0.7^2)$. We consider 15 different scenarios with different numbers of observations (100, 200, 300, 400 and 500) and input dimensions (2, 3 and 4). The structure and data generation process of Experiment \ref{exp:1} follows \cite{lee2013more}. We fix the number of evaluation points at approximately 400 and locate them on a uniform grid.
	\end{experiment}
	
	For this experiment, we compare the following four estimators: SCKLS, CNLS, Local Linear Kernel (LL), and parametric Cobb--Douglas estimator. The latter estimator serves as a baseline because it is correctly specified parametric form. Tables \ref{tab:1.exp1obs} and \ref{tab:2.exp1grd} show for Experiment \ref{exp:1} the RMSE measured on observation points and evaluation points, respectively. The number in parentheses is the standard deviation of RMSE values computed by 10 replications. Note the standard derivations are generally small compared to the parameter estimates, which indicates low variability even after only 10 replications. A more extensive set of results for this experiment is summarized in  Appendix \ref{App:CompResults}. The SCKLS estimator has the lowest RMSE in most scenarios even when RMSE is measured on observation points (note that the SCKLS estimator imposes the global shape constraints via evaluation points in Equation~(\ref{eq:6.SCKLS})). Also as expected, the performance of SCKLS estimator improves as the number of observation points increases. Moreover, the SCKLS estimator performs better than the LL estimator particularly in higher dimensional functional estimation. This provides empirical evidence that the shape constraints in SCKLS are helpful in improving the finite sample performance as compared to LL. Note that LL appears to have larger RMSE values on evaluation points which are located in input space regions with sparse observations. This implies that the SCKLS estimator has more robust out-of-sample performance than the LL estimator due to the shape constraints. We also observe that the performance of the CNLS estimator measured at the evaluation points is worse than that measured at the observations. CNLS often has ill-defined hyperplanes which are very steep/shallow at the edge of the observed data, and this over-fitting leads to poor out-of-sample performance. In contrast, the SCKLS estimator performs similarly for both the observation points and evaluation points, because the construction of the grid that completely covers the observed data makes the SCKLS estimator more robust. %Although the SCKLS and the CWB estimators are relatively competitive when the number of observations is large, the SCKLS estimator has significantly better performance when the observations are smaller and in higher dimension.
	
	% Table generated by Excel2LaTeX from sheet 'Table1'
	\renewcommand{\arraystretch}{0.65}
	\begin{table}[!b]
		\footnotesize
		\centering
		\caption{RMSE on observation points for Experiment \ref{exp:1}.}
		\begin{tabular}{llccccc}
			\toprule
			& & \multicolumn{5}{c}{Average of RMSE on observation points} \\
			\multicolumn{2}{c}{Number of observations} &
			\multicolumn{1}{c}{100} & \multicolumn{1}{c}{200} & \multicolumn{1}{c}{300} & \multicolumn{1}{c}{400} & \multicolumn{1}{c}{500} \\
			\midrule
			\multirow{7}[2]{*}{2-input} & \multicolumn{1}{l}{SCKLS} & \bf 0.193 & 0.171 & 0.141 & \bf 0.132 & 0.118 \\
			& \multicolumn{1}{l}{} & (0.053) & (0.047) & (0.032) & (0.029) & (0.017) \\
			& \multicolumn{1}{l}{CNLS} & 0.229 & \bf 0.163 & \bf 0.137 & 0.138 & \bf 0.116 \\
			& \multicolumn{1}{l}{} & (0.042) & (0.037) & (0.010) & (0.027) & (0.016) \\
			%& \multicolumn{1}{l}{CWB in $p$-space} & \bf 0.189 & 0.167 & 0.158 & 0.140 & 0.129 \\
			%& \multicolumn{1}{l}{} & (0.055) & (0.049) & (0.040) & (0.026) & (0.019) \\
			& \multicolumn{1}{l}{LL} & 0.212 & 0.166 & 0.149 & 0.152 & 0.140 \\
			& \multicolumn{1}{l}{} & (0.079) & (0.042) & (0.028) & (0.028) & (0.028) \\
			\cmidrule{2-7}
			& \multicolumn{1}{l}{Cobb--Douglas } & 0.078 & 0.075 & 0.048 & 0.039 & 0.043 \\
			\midrule
			\multirow{7}[2]{*}{3-input} & \multicolumn{1}{l}{SCKLS} &\bf 0.230 & \bf 0.187 & \bf 0.183 & \bf 0.152 & \bf 0.165 \\
			& \multicolumn{1}{l}{} & (0.050) & (0.026) & (0.032) & (0.019) & (0.031) \\
			& \multicolumn{1}{l}{CNLS} & 0.294 & 0.202 & 0.189 & 0.173 & 0.168 \\
			& \multicolumn{1}{l}{} & (0.048) & (0.035) & (0.020) & (0.014) & (0.020) \\
			%& \multicolumn{1}{l}{CWB in $p$-space} & \bf 0.228 & 0.221 & 0.210 & 0.183 & 0.172 \\
			%& \multicolumn{1}{l}{} & (0.043) & (0.039) & (0.037) & (0.040) & (0.024) \\
			& \multicolumn{1}{l}{LL} & 0.250 & 0.230 & 0.235 & 0.203 & 0.181 \\
			& \multicolumn{1}{l}{} & (0.068) & (0.050) & (0.052) & (0.050) & (0.021) \\
			\cmidrule{2-7}
			& \multicolumn{1}{l}{Cobb--Douglas } & 0.104 & 0.089 & 0.070 & 0.047 & 0.041 \\
			\midrule
			\multirow{7}[2]{*}{4-input} & \multicolumn{1}{l}{SCKLS} & \bf 0.225 & \bf 0.248 & \bf 0.228 & \bf 0.203 & \bf 0.198 \\
			& \multicolumn{1}{l}{} & (0.038) & (0.020) & (0.037) & (0.042) & (0.028) \\
			& \multicolumn{1}{l}{CNLS} & 0.315 & 0.294 & 0.246 & 0.235 & 0.214 \\
			& \multicolumn{1}{l}{} & (0.039) & (0.027) & (0.024) & (0.029) & (0.015) \\
			%& \multicolumn{1}{l}{CWB in $p$-space} & 0.238 & 0.262 & 0.231 & 0.234 & 0.198 \\
			%& \multicolumn{1}{l}{} & (0.038) & (0.056) & (0.039) & (0.076) & (0.030) \\
			& \multicolumn{1}{l}{LL} & 0.256 & 0.297 & 0.252 & 0.240 & 0.226 \\
			& \multicolumn{1}{l}{} & (0.044) & (0.057) & (0.056) & (0.060) & (0.038) \\
			\cmidrule{2-7}
			& \multicolumn{1}{l}{Cobb--Douglas } & 0.120 & 0.073 & 0.091 & 0.067 & 0.063 \\
			\bottomrule
		\end{tabular}%
		\label{tab:1.exp1obs}%
	\end{table}%
	
	% Table generated by Excel2LaTeX from sheet 'Table2'
	\begin{table}
		\footnotesize
		\centering
		\caption{RMSE on evaluation points for Experiment \ref{exp:1}.}
		\begin{tabular}{llccccc}
			\toprule
			& & \multicolumn{5}{c}{Average of RMSE on evaluation points} \\
			\multicolumn{2}{c}{Number of observations} &
			\multicolumn{1}{c}{100} & \multicolumn{1}{c}{200} & \multicolumn{1}{c}{300} & \multicolumn{1}{c}{400} & \multicolumn{1}{c}{500} \\
			\midrule
			\multirow{7}[0]{*}{2-input} & \multicolumn{1}{l}{SCKLS} &\bf 0.219 & 0.189 & \bf0.150 & \bf0.147 & \bf0.128 \\
			& \multicolumn{1}{l}{} & (0.053) & (0.057) & (0.034) & (0.030) & (0.021) \\
			& \multicolumn{1}{l}{CNLS} & 0.350 & 0.299 & 0.260 & 0.284 & 0.265 \\
			& \multicolumn{1}{l}{} & (0.082) & (0.093) & (0.109) & (0.119) & (0.078) \\
			%& \multicolumn{1}{l}{CWB in $p$-space} & \bf0.206 & 0.186 & 0.174 & 0.154 & 0.143 \\
			%& \multicolumn{1}{l}{} & (0.049) & (0.062) & (0.043) & (0.026) & (0.021) \\
			& \multicolumn{1}{l}{LL} & 0.247 & \bf0.182 & 0.167 & 0.171 & 0.156 \\
			& \multicolumn{1}{l}{} & (0.101) & (0.053) & (0.030) & (0.030) & (0.034) \\
			\cmidrule{2-7}
			& \multicolumn{1}{l}{Cobb--Douglas } & 0.076 & 0.076 & 0.049 & 0.040 & 0.043 \\
			\midrule
			\multirow{7}[0]{*}{3-input} & \multicolumn{1}{l}{SCKLS} & \bf0.283 & \bf0.231 & \bf0.238 & \bf0.213 & \bf0.215 \\
			& \multicolumn{1}{l}{} & (0.072) & (0.033) & (0.030) & (0.029) & (0.034) \\
			& \multicolumn{1}{l}{CNLS} & 0.529 & 0.587 & 0.540 & 0.589 & 0.598 \\
			& \multicolumn{1}{l}{} & (0.112) & (0.243) & (0.161) & (0.109) & (0.143) \\
			%& \multicolumn{1}{l}{CWB in $p$-space} & 0.291 & 0.289 & 0.269 & 0.252 & 0.233 \\
			%& \multicolumn{1}{l}{} & (0.069) & (0.052) & (0.035) & (0.052) & (0.031) \\
			& \multicolumn{1}{l}{LL} & 0.336 & 0.340 & 0.360 & 0.326 & 0.264 \\
			& \multicolumn{1}{l}{} & (0.085) & (0.093) & (0.108) & (0.086) & (0.042) \\
			\cmidrule{2-7}
			& \multicolumn{1}{l}{Cobb--Douglas } & 0.116 & 0.098 & 0.080 & 0.052 & 0.046 \\
			\midrule
			\multirow{7}[0]{*}{4-input} & \multicolumn{1}{l}{SCKLS} & \bf0.321 & \bf0.357 & \bf0.329 & \bf0.308 & \bf0.290 \\
			& \multicolumn{1}{l}{} & (0.046) & (0.065) & (0.049) & (0.084) & (0.044) \\
			& \multicolumn{1}{l}{CNLS} & 0.845 & 0.873 & 0.901 & 0.827 & 0.792 \\
			& \multicolumn{1}{l}{} & (0.188) & (0.137) & (0.151) & (0.235) & (0.091) \\
			%& \multicolumn{1}{l}{CWB in $p$-space} & 0.360 & 0.385 & 0.358 & 0.361 & 0.325 \\
			%& \multicolumn{1}{l}{} & (0.039) & (0.077) & (0.068) & (0.138) & (0.077) \\
			& \multicolumn{1}{l}{LL} & 0.482 & 0.527 & 0.483 & 0.495 & 0.445 \\
			& \multicolumn{1}{l}{} & (0.115) & (0.125) & (0.146) & (0.153) & (0.074) \\
			\cmidrule{2-7}
			& \multicolumn{1}{l}{Cobb--Douglas } & 0.146 & 0.091 & 0.115 & 0.081 & 0.080 \\
			\bottomrule
		\end{tabular}%
		\label{tab:2.exp1grd}%
	\end{table}%
	
% 	Table \ref{tab:3.exp1time} shows the computational time of Experiment \ref{exp:1} for the shape constrained estimators which are more computationally intensive than the Local Linear estimator. The value inside parenthesis indicates the percentage of constraints included in the optimization problem that resulted in the optimal solution. While CNLS is the fastest estimator when the number of observations is small, the SCKLS estimator could be  faster in the case of larger datasets or high dimensional scenarios. Since we alleviate the curse of dimensionality by using the iterative algorithm, we can analyze larger dataset to estimate production function. The iterative algorithm is more effective for estimating low dimensionality functions because in higher dimensional spaces the number of adjacent grid points is relatively large.
	
	We also conduct simulations with different bandwidths to analyze the sensitivity of each estimator to bandwidths. We compare SCKLS and LL with bandwidth $h\in[0,10]$ with an increment by 0.01 for the 1-input setting, and we use bandwidth $\bm{h}\in[0,5] \times [0,5]$ with an increment by 0.25 in each coordinate for the 2-input setting. We simulate 100 datasets to compute the RMSE for each bandwidth as well as for the bandwidth via LOOCV. Figure \ref{fig:A1.RMSE} displays the average RMSE of each estimator. The histogram shows the distribution of bandwidths selected by LOOCV. The instances when SCKLS and LL provide the lowest RMSE are shown in light gray and dark gray respectively. For the one-input scenario, the SCKLS estimator performs better than the LL estimator for bandwidth between 0.25 - 2.25 as shown in (a). For the two-input scenario, the SCKLS estimator performs better for most of the LOOCV values as shown by the majority of the histogram colored in light gray. This indicates that LOOCV, calculated using the unconstrained estimator, provides bandwidths that work well for the SCKLS estimator. Importantly, the SCKLS estimator does not appear to be very sensitive to the bandwidth selection method since, heuristically, the shape constraints help reduce the variance of the estimator.  Finally, we note that similar results can be obtained in experimental settings with lower signal-to-noise level, or with non-uniform input. See Appendix~\ref{App:CompResults} for more details.
	\begin{figure}
		\centering
		\subfloat[One-input]{\includegraphics[width=0.5\textwidth]{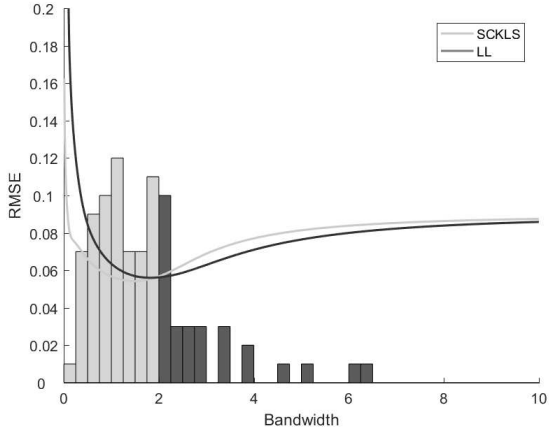}\label{fig:RMSE1}}
		\hfill
		\subfloat[Two-input]{\includegraphics[width=0.5\textwidth]{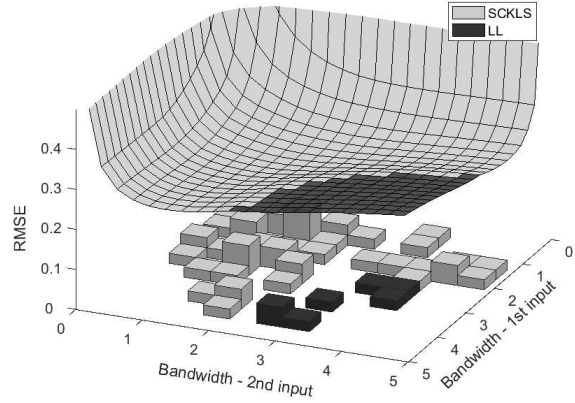}\label{fig:RMSE2}}
		\caption{The histogram shows the distribution of bandwidths selected by LOOCV. The curves show the relative performance of each estimator.}
		\label{fig:A1.RMSE}
	\end{figure}

	\paragraph{Different numbers of evaluation points}
	\label{sec:5.1.3.2 exp4}
	\begin{experiment}
		\label{exp:4}
		The setting is the same as Experiment \ref{exp:1}. However, now we consider 9 different scenarios with different numbers of evaluation points (100, 300 and 500) and input dimensions (2, 3 and 4). We fix the number of observed points at 400.
	\end{experiment}
	
	We show the performance of SCKLS. Table \ref{tab:8.exp4obs} and \ref{tab:9.exp4grd} shows for Experiment \ref{exp:4} the RMSE measured on observations and evaluation points respectively. %The comprehensive results of this experiment is summarized in  Appendix \ref{App:AppendixG}. 
	Both tables show that empirically even if we increase the number of evaluation points, the RMSE value does not change significantly. %Thus, we conclude that the SCKLS estimator performs well even on a relatively rough grid of points. 
	This has important implications for the running time. Specifically, we can reduce the calculation time by using a rough grid without sacrificing too much in terms of RMSE performance of the estimator.
	
	\newcolumntype{L}[1]{>{\raggedright\arraybackslash}p{#1}}
	\newcolumntype{C}[1]{>{\centering\arraybackslash}p{#1}}
	\newcolumntype{R}[1]{>{\raggedleft\arraybackslash}p{#1}}
	
	% Table generated by Excel2LaTeX from sheet 'Table8'
	\begin{table}[!htp]
		\footnotesize
		\centering\arraybackslash
		\caption{RMSE on observation points for Experiment \ref{exp:4}.}
		\begin{tabular}{rlR{0.7in}R{0.7in}R{0.7in}}
			\toprule
			&       & \multicolumn{3}{c}{Average of RMSE on observation points} \\
			\multicolumn{2}{c}{Number of evaluation points} & 100   & 300   & 500 \\
			\midrule
			\multicolumn{1}{c}{2-input} & SCKLS & 0.142 & 0.141 & 0.141 \\
			%\multicolumn{1}{c}{} & CNLS &  &  &  \\
			%\multicolumn{1}{c}{} & Local Linear&  &  &  \\
			\midrule
			\multicolumn{1}{c}{3-input} & SCKLS & 0.198 & 0.203 & 0.197 \\
			%\multicolumn{1}{c}{} & CNLS &  &  &  \\
			%\multicolumn{1}{c}{} & Local Linear &  &  &  \\
			\midrule
			\multicolumn{1}{c}{4-input} & SCKLS & 0.239 & 0.207 & 0.206 \\
			%\multicolumn{1}{c}{} & CNLS &  &  &  \\
			%\multicolumn{1}{c}{} & Local Linear &  &  &  \\
			\bottomrule
		\end{tabular}%
		\label{tab:8.exp4obs}%
	\end{table}%
	
	% Table generated by Excel2LaTeX from sheet 'Table9'
	\begin{table}[!htp]
		\footnotesize
		\centering
		\caption{RMSE on evaluation points for Experiment \ref{exp:4}.}
		\begin{tabular}{rlR{0.7in}R{0.7in}R{0.7in}}
			\toprule
			& & \multicolumn{3}{c}{Average of RMSE on evaluation points} \\
			\multicolumn{2}{c}{Number of evaluation points} & 100   & 300   & 500 \\
			\midrule
			\multicolumn{1}{c}{2-input} & SCKLS  & 0.181 & 0.164 & 0.158 \\
			%\multicolumn{1}{c}{} & CNLS &  &  &  \\
			%\multicolumn{1}{c}{} & Local Linear &  &  &  \\
			\midrule
			\multicolumn{1}{c}{3-input} & SCKLS  & 0.304 & 0.267 & 0.257 \\
			%\multicolumn{1}{c}{} & CNLS &  &  &  \\
			%\multicolumn{1}{c}{} & Local Linear &  &  &  \\
			\midrule
			\multicolumn{1}{c}{4-input} & SCKLS  & 0.383 & 0.296 & 0.270 \\
			%\multicolumn{1}{c}{} & CNLS &  &  &  \\
			%\multicolumn{1}{c}{} & Local Linear &  &  &  \\
			\bottomrule
		\end{tabular}%
		\label{tab:9.exp4grd}%
	\end{table}%

	\subsection{Numerical experiments on testing the imposed shape}
	%\subsubsection{Testing for monotonicity and concavity}
	\begin{experiment}
	\label{exp:test}
    We test monotonicity and concavity for data generated from the following single-input and single-output DGP:
    \begin{equation}
	    g_0(x) = x^p
  		\label{eq:DGP1}
    \end{equation}
    and 
    \begin{equation}
	    g_0(x)  = \frac{1}{1+\exp(-5\log (2x))}.
	    \label{eq:DGP2}
    \end{equation}
    With $n$ observations, for each pair $(X_j,y_j)$, each input, $X_{j}$, is randomly and independently drawn from uniform distribution $unif[0,1]$. In this simulation, we use the following multiplicative noise to validate whether the wild bootstrap can handle non-homogeneous noise. 
    \[y_j =  g_0(X_j)+(X_j+1)\cdot\epsilon_j ,\]
    where $\epsilon_j$, is randomly and independently sampled from a normal distribution, $N(0,\sigma^2)$. We use three different DGP scenarios A, B and C. For scenarios A and B, we use function (\ref{eq:DGP1}) where the exponent parameter $p$ defines whether the function $g_0$ is an element of the class of functions $G_2$ or not. We use $p=\{0,2\}$ for scenarios A and B respectively, where $g_0\in G_2$ if $p=0$, and $g_0\notin G_2$ if $p=2$ since $g_0$ is strictly convex. For scenario C, we consider an ``S''-shape function defined by (\ref{eq:DGP2}) which violates both global concavity and convexity. We consider different sample sizes $n=\{100,300,500\}$ and standard deviation of the noise $\sigma=\{0.1,0.2\}$, and perform 500 simulations to compute the rejection rate for each scenario. We assume that we do not know the distribution of the noise in advance and use the wild bootstrap procedure described in Section~\ref{sec:5.1test} with $B=200$. %The comprehensive results of this experiment is summarized in Appendix \ref{App:TestResult}.
    \end{experiment}

    Table \ref{tab:test1} shows the rejection rate for each DGP. For high signal-to-noise ratio scenarios ($\sigma=0.1$), the test works well even with a small sample size. Our test is able to control the Type I error, as illustrated in scenario A. In addition, the Type I\hspace{-.1em}I error of our test is small for the scenarios B and C where shape constraints are violated by the DGP. Furthermore, for low signal-to-noise ratio scenarios ($\sigma=0.2$), the rejection rate for scenarios B and C significantly improves when the sample size is increased from 100 to 300. Indeed, for larger noise scenarios more data is required for the test to have power. Thus, our test seems informative enough to guide users to avoid imposing shape constraints on the data generated from misspecified functions.
	
	 % Table generated by Excel2LaTeX from sheet 'Test1'
	 \begin{table}[htbp]
	 	\centering
	 	\caption{Rejection rate (\%) of the test for monotonicity and concavity}
	 	\begin{tabular}{cc|rr|rr}
	 		\toprule
	 		\multirow{2}[0]{*}{Sample size} & \multicolumn{1}{c|}{\multirow{3}[0]{*}{DGP Scenario}} &
	 		\multicolumn{4}{c}{Power of the Test ($\alpha$)} \\
	 		& & 
	 		\multicolumn{1}{r}{0.05} & \multicolumn{1}{r|}{0.01} & \multicolumn{1}{r}{0.05} & \multicolumn{1}{r}{0.01} \\
	 		\multicolumn{1}{c}{($n$)} &  & \multicolumn{2}{c|}{$\sigma=0.1$} & \multicolumn{2}{c}{$\sigma=0.2$} \\
	 		\midrule
	 		\multirow{4}[0]{*}{100} 
	 		& A ($H_0$)  & 5.8  & 2.0  & 8.0 & 2.6 \\
	 		%& 0.5   & 0.2 & 0.0 & 1.0 & 0.8 \\
	 		%& 1.0   & 2.0 & 0.4 & 1.6 & 0.4 \\
	 		& B ($H_1$)  & 98.6 & 94.6 & 55.0 & 36.2 \\
	 		%& C ($H_1$)  & 100.0 & 100.0 & 87.8 & 75.0 \\
	 		& C ($H_1$)  & 98.6 & 94.4 & 42.6 & 24.2 \\
	 		\midrule
	 		\multirow{4}[0]{*}{300} 
	 		& A ($H_0$)  & 6.8  & 1.8  & 6.6 & 3.0 \\
	 		%& 0.5   & 0.4  & 0.0  & 0.0 & 0.0 \\
	 		%& 1.0   & 1.0  & 0.8  & 2.6 & 0.4 \\
	 		& B ($H_1$)  & 100.0  & 100.0  & 92.0 & 83.2 \\
	 		%& C ($H_1$)  & 100.0  & 100.0  & 100.0 & 99.4 \\
	 		& C ($H_1$)  & 100.0  & 100.0  & 97.0 & 86.8 \\
	 		\midrule
	 		\multirow{4}[0]{*}{500} 
	 		& A ($H_0$)  & 5.4  & 1.6  & 5.6 & 1.4 \\
	 		%& 0.5   & 0.0  & 0.0  & 0.2 & 0.0 \\
	 		%& 1.0   & 1.2  & 0.2  & 1.8 & 0.6 \\
	 		& B ($H_1$)  & 100.0  & 100.0  & 99.4 & 97.2 \\
	 		%& C ($H_1$)  & 100.0  & 100.0  & 100.0 & 100.0 \\
	 		& C ($H_1$)  & 100.0  & 100.0  & 99.8 & 99.4 \\
	 		\bottomrule
	 	\end{tabular}%
	 	\label{tab:test1}%
	 \end{table}%

	\newpage
	\section{Application}
	\label{sec:6.application}
	We apply the proposed method to estimate the production function for two large industries in Chile: plastic (2520) and wood manufacturing (2010) where the values inside the parentheses indicate the CIIU3 industry code. There are some existing studies which analyze the productivity of Chilean data, see for example \cite{pavcnik2002trade}, who analyzed the effect of trade liberalization on productivity improvements. Other researchers have analyzed the productivity of Chilean manufacturing including \cite{benavente2006role}, \cite{alvarez2009} and \cite{levinsohn2003estimating}. However, the above-cited work use strong parametric assumptions and older data. Most studies use the Cobb--Douglas functional form which restricts the elasticity of substitution to be 1. When diminishing marginal productivity of inputs characterizes the data, the Cobb--Douglas functional form imposes that the most productive scale size is at the origin. We relax the parametric assumptions and estimate a shape constrained production function nonparametrically using data from 2010. We examine the marginal productivity, marginal rate of substitution, and most productive scale size (MPSS) to analyze the structure of the industries. We also investigate how productivity differs between exporting and non-exporting firms, as exporting has become an important source of revenue in Chile\footnote{Note that firms' decisions, i.e., selecting labor and capital levels with considerations for productivity levels or whether to export, are potentially endogenous. Solutions to this issue are to instrument or build a structural model based on timing assumptions. Our estimator can be embedded within the estimation procedures such as those described in \cite{Ackerberg2015} to address this issue.}. See Appendix \ref{App:application} for the details of estimation and comparison across different estimators.
	
	\subsection{The census of Chilean manufacturing plants}
	\label{sec:6.1.data}
	We use the Chilean Annual Industrial Survey provided by Chile's National Institute of Statistics\footnote{The data are available at \url{http://www.ine.cl/estadisticas/economicas/manufactura}.}. The survey covers manufacturing establishments with ten or more employees. We define Capital and Labor as the input variables and Value Added as the output variable of the production function\footnote{The definition of Labor includes full-time, part-time, and outsourced labors. Capital is defined as a sum of the fixed assets balance such as buildings, machines, vehicles, furniture, and technical software. Value added is computed by subtracting the cost of raw materials and intermediate consumption from the total amount produced. Further details are available at \url{http://www.ine.cl/estadisticas/economicas/manufactura}.}. Capital and Value Added are measured in millions of Chilean peso while Labor is measured as the total man-hours per year. We use cross sectional data from the plastic and the wood industries.
	 
	Many researchers have found positive effects of exporting for other countries using parametric models. See for instance, \citet{de2007exports} and \citet{bernard2004exporting}. Here we use SCKLS to relax the parametric assumption for the production function. To capture the effects of exporting, we use a semi-parametric modeling extension of SCKLS. The partially linear model is represented as follows:
	\begin{equation}
	y_j=\bm{Z}_j'\bm{\gamma}+g_0(\bm{X}_j)+\epsilon_j ,
	\end{equation}
	where $\bm{Z}_j=(Z_{j1},Z_{j2})'$ denotes contextual variables and $\bm{\gamma}=(\gamma_1,\gamma_2)'$ is the coefficient of contextual variables. We model exporting with two variables: a dummy variable indicating the establishments that are exporting and the share of output being exported. For more details see Appendix \ref{App:AppendixI}. %Our empirical results indicate exporting establishments are more productive than non-exporting establishments. 
	
	Table \ref{tab:11.sumstat} presents the summary of statistics for each industry by exporter/non-exporter. We find that exporters are typically larger than non-exporter in terms of labor and capital. Input variables are positively skewed, indicating there exist many small and few large establishments. Since SCKLS with variable bandwidth ($k$-nearest neighbor) and non-uniform grid performed the best in our simulation scenarios with non-uniform input data (as indicated in Appendix~\ref{App:CompResults}), we use these options. We choose the smoothing parameter $k$ via leave-one-out cross validation. Appendix \ref{App:AppendixA} explains the details of our implementation of K-NN for the SCKLS estimator.
	
	% Table generated by Excel2LaTeX from sheet 'Table11'
	\begin{table}
	\footnotesize
		\centering
		\caption{Statistics of Chilean manufacturing data.}
		\begin{tabular}{lrrrrrrrr}
			\toprule			\multicolumn{1}{c}{\multirow{4}[0]{*}{\shortstack{Plastic\\(2520)}}} & \multicolumn{3}{c}{Non-exporters ($n=173$)} & \multicolumn{1}{c}{} & \multicolumn{4}{c}{Exporters ($n=72$)} \\
			\cmidrule{2-4}
			\cmidrule{6-9}
			\multicolumn{1}{c}{} & Labor & \shortstack{Capital\\(million)} & \shortstack{Value\\Added\\(million)} &       & Labor & \shortstack{Capital\\(million)} & \shortstack{Value\\Added\\(million)} & \shortstack{Share of\\Exports} \\
			\midrule
			mean  & 92155 & 725.85 & 546.93 &       & 240890 & 2859  & 1733.9 & 0.147 \\
			median & 55220 & 258.41 & 247.05 &       & 180330 & 1329.1 & 1054.9 & 0.0524 \\
			std   & 106530 & 1574  & 1068.1 &       & 212480 & 3840.2 & 1678.8 & 0.201 \\
			skewness & 3.301 & 5.2052 & 5.9214 &       & 1.3681 & 2.4594 & 1.0678 & -0.303 \\
			\midrule
			\multicolumn{1}{c}{\multirow{3}[0]{*}{\shortstack{Wood\\(2010)}}} & \multicolumn{3}{c}{Non-exporters ($n=97$)} & \multicolumn{1}{c}{} & \multicolumn{4}{c}{Exporters ($n=35$)} \\
			\cmidrule{2-4}
			\cmidrule{6-9}
			\multicolumn{1}{c}{} & Labor & Capital &\shortstack{Value\\Added} &       & Labor & Capital & \shortstack{Value\\Added} & \shortstack{Share of\\Exports} \\
			\midrule
			mean  & 76561 & 364.93 & 334.83 &       & 501470 & 3063.4 & 4524.1 & 0.542 \\
			median & 44087 & 109.48 & 115.39 &       & 378000 & 2195.4 & 2673.5 & 0.648 \\
			std   & 78057 & 702.35 & 555.87 &       & 436100 & 2510.3 & 4466.3 & 0.355 \\
			skewness & 2.243 & 3.5155 & 3.432 &       & 0.81454 & 0.63943 & 1.0556 & -0.303 \\
			\bottomrule
		\end{tabular}%
		\label{tab:11.sumstat}%
	\end{table}%
	
	Figure \ref{fig:3.input} is a plot of labor and capital for each industry and shows input data is sparse for large establishments. \cite{beresteanu2005nonparametric} proposed to include shape constraints only for the evaluation points that are close to the observations. Thus, in addition to using a percentile grid of evaluation points, we propose to use the evaluation points that are inside the convex hull of observed input $\{\bm{X}_j\}_{j=1}^n$. See Appendix \ref{App:application} for details.
% 	\cite{beresteanu2005nonparametric} proposed to include shape constraints only for the evaluation points that are inside the convex full of observed input $\{X_j\}_{j=1}^n$. See details in Appendix \ref{App:application}. %As such, in addition to using a non-uniform grid of evaluation points, we also apply Beresteanu criteria to avoid using redundant evaluation points and over smoothing.
	
	\begin{figure}
		\centering
		\subfloat[Plastic (2520)]{\includegraphics[width=0.5\textwidth]{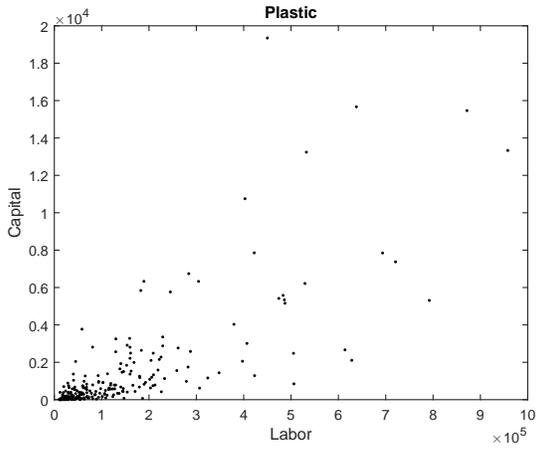}\label{fig:f1b}}
		\hfill
		\subfloat[Wood (2010)]{\includegraphics[width=0.5\textwidth]{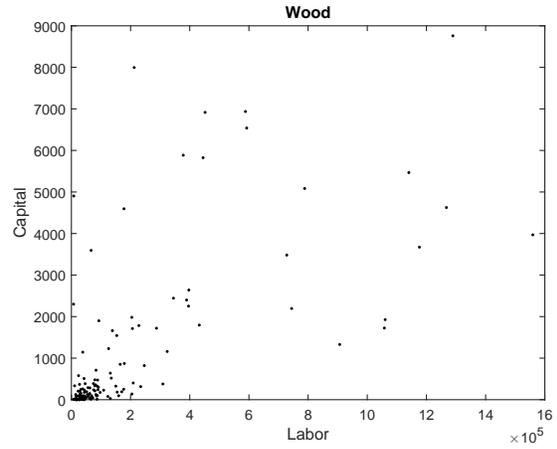}\label{fig:f2b}}
		\caption{Labor and Capital of each industry.}
		\label{fig:3.input}
	\end{figure}
	
	We begin by testing if the Cobb--Douglas production function is appropriate for our data. We use the hypothesis test for correct parametric specification described in \cite{henderson2015applied}\footnote{We apply a Cobb--Douglas OLS to the second stage data $\{\bm{X}_j, y_j-\bm{Z}_j\bm{\gamma}\}_{j=1}^n$ which removes the effect of contextual variables from observed output. See Appendix \ref{App:AppendixI} for details.}. The resulting $p$-value is 0.092 for the plastic industry and 0.007 for the wood industry, respectively. Therefore, the Cobb--Douglas parametric specification is likely to be wrong, particularly applied to the wood industry. 
	
	Next, we apply the test proposed in Section \ref{sec:5.1test} to determine if imposing global concavity and monotonicity shape constraints is appropriate. We estimate a $p$-value of $0.302$ for the plastic industry and $0.841$ for the wood industry, respectively. For both industries, the estimated $p$-value is not small enough to reject $H_0$, which means that the observed data is likely to satisfy the shape constraints imposed.
	
	\subsection{Estimated production function and interpretation}
	\label{sec:6.2.est}
	
	We estimate a semi-parametric model with a nonparametric shape constrained production function, a linear model for exporting share of sales, and a dummy variable for exporting. Table~\ref{tab:12.fit} shows the goodness of fit $(R^2)$ of the production function: 71.1\% of variance is explained in the plastic industry while 43.8\% of variance is explained in the wood industry.
	
	% Table generated by Excel2LaTeX from sheet 'Table12'
	\begin{table}[!htbp]
		\footnotesize
		\centering
		\caption{SCKLS fitting statistics for cross sectional data.}
		\begin{tabular}{lcc}
			\toprule
			Industry & \shortstack{Number of\\observations} & \shortstack{$R^2$}  \\
			\midrule
			Plastic & 245   & 71.1\%  \\
			Wood  & 132   & 43.8\%  \\
			\bottomrule
		\end{tabular}%
		\label{tab:12.fit}%
	\end{table}%
	
	Table \ref{tab:13.est} reports additional information characterizing the production function: the marginal productivity and the marginal rates of substitution at the 10, 25, 50, 75 and 90 percentiles are reported for both measures. Here, the rate of substitution indicates how much labor is required to maintain the same level of output when we decrease a unit of capital. When comparing the two industries, we find that the wood industry has a larger marginal rate of substitution than the plastic industry. This indicates that capital is more critical in the wood industry than the plastic industry.
	
	We also compare the estimated production function by the local linear and the SCKLS estimators. Figure~\ref{fig:prod_func_plastic} and Figure~\ref{fig:prod_func_wood} show the estimated production function within the convex hull of observations for plastic and wood industries, respectively. Visually, the production function estimated by the LL estimator is difficult to interpret and the values of important economic quantities such as marginal products and marginal rates of substitution are also hard to interpret. In particular, it is not possible to identify most productive scale size.  
	
	% Table generated by Excel2LaTeX from sheet 'Table13'
	\begin{table}[!htbp]
		\footnotesize
		\centering
		\caption{Characteristics of the production function.}
		\begin{tabular}{lrrr}
			\toprule
			& \multicolumn{3}{c}{Plastic (2520)} \\
			\cmidrule{2-4}
			& \multicolumn{2}{c}{Marginal Productivity} & \multicolumn{1}{c}{Marginal Rate of Substitution}  \\
			& \multicolumn{1}{c}{Labor ($=b_l$)} & \multicolumn{1}{c}{Capital ($=b_k$)} & \multicolumn{1}{c}{($=b_k/b_l$)} \\
			& \multicolumn{1}{c}{(million peso/man hours)} & \multicolumn{1}{c}{(peso/peso)} &  \\
			\midrule
			%\multicolumn{1}{l}{0th percentile} & 0.000182 & 3.4$\times10^{-15}$ & 1.89$\times10^{-11}$ \\
			\multicolumn{1}{l}{10th percentile} & 0.00396 & 0.111 & 23.3 \\
			\multicolumn{1}{l}{25th percentile} & 0.00523 & 0.139 & 23.9 \\
			\multicolumn{1}{l}{50th percentile} & 0.00579 & 0.139 & 24.0 \\
			\multicolumn{1}{l}{75th percentile} & 0.00579 & 0.139 & 35.3 \\
			\multicolumn{1}{l}{90th  percentile} & 0.00579 & 0.260 & 44.8 \\
			\midrule
			& \multicolumn{3}{c}{Wood (2010)} \\
			\cmidrule{2-4}
			& \multicolumn{2}{c}{Marginal Productivity} & \multicolumn{1}{c}{Marginal Rate of Substitution}  \\
			& \multicolumn{1}{c}{Labor ($=b_l$)} & \multicolumn{1}{c}{Capital ($=b_k$)} & \multicolumn{1}{c}{($=b_k/b_l$)} \\
			\midrule
			%0th percentile & 3.89$\times10^{-20}$ & 0.194 & 82.9 \\
			10th percentile & 1.46$\times10^{-18}$ & 0.816 & 760 \\
			25th percentile & 8.55$\times10^{-16}$ & 0.816 & 760 \\
			50th percentile & 0.00133 & 1.01  & 760 \\
			75th percentile & 0.00133 & 1.01  & 9.73$\times10^{14}$ \\
			90th  percentile & 0.00133 & 1.01  & 5.59$\times10^{17}$ \\
			%100th percentile & 0.0370 & 11.8  & 2.10$\times10^{19}$ \\
			\bottomrule
		\end{tabular}%
		\label{tab:13.est}%
	\end{table}%
	
% 	\begin{figure}[!htbp]
% 		\centering
% 		\subfloat[Plastic (2520)]{\includegraphics[width=0.5\textwidth]{fig4a.eps}\label{fig:f1a}}
% 		\hfill
% 		\subfloat[Wood (2010)]{\includegraphics[width=0.5\textwidth]{fig4b.eps}\label{fig:f2a}}
% 		\caption{Estimated Production Function.}
% 		\label{fig:4.est}
% 	\end{figure}

	\begin{figure}
		\centering
		\subfloat[Local Linear]{\includegraphics[width=0.5\textwidth]{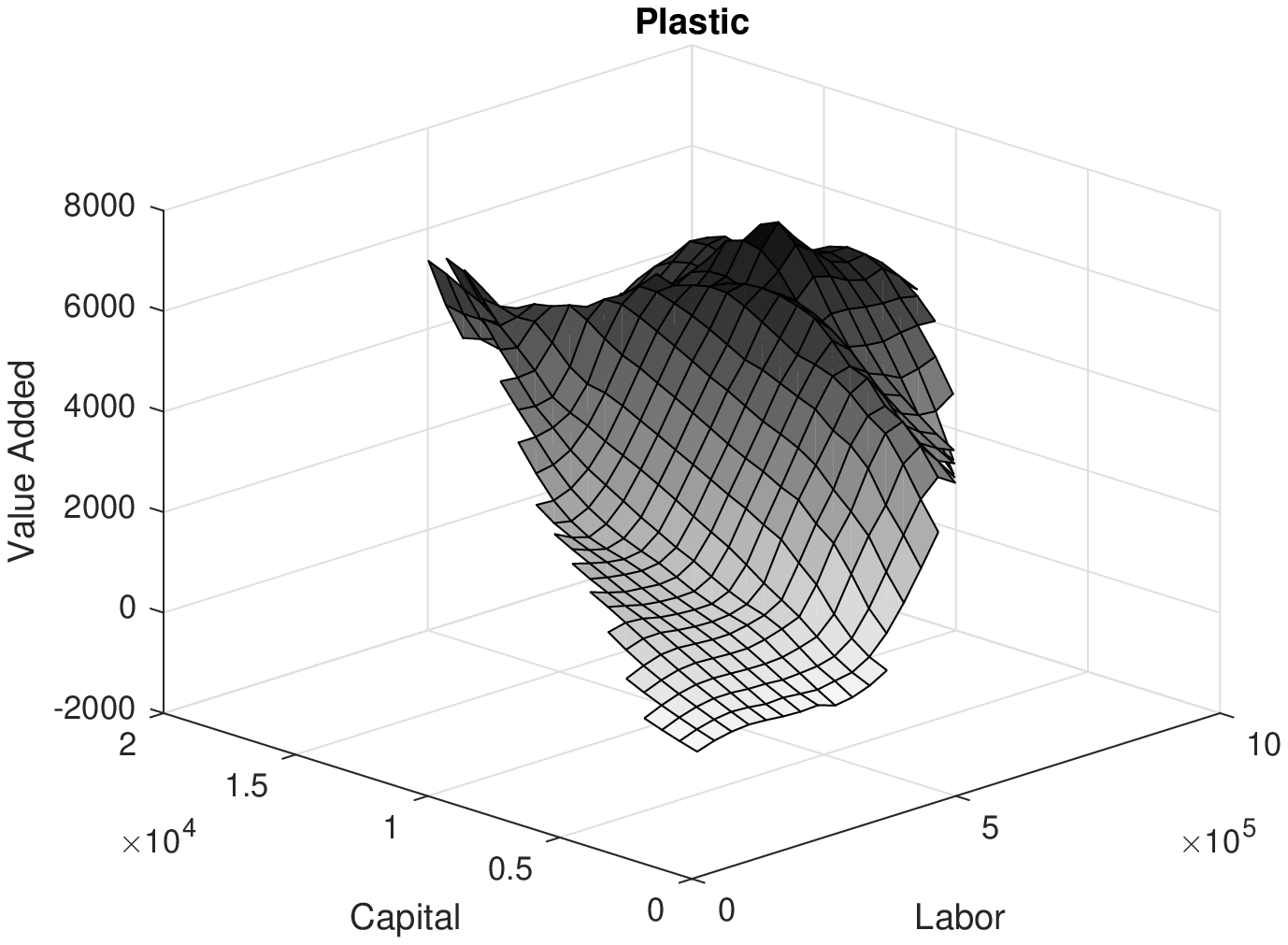}}
		\hfill
		\subfloat[SCKLS]{\includegraphics[width=0.5\textwidth]{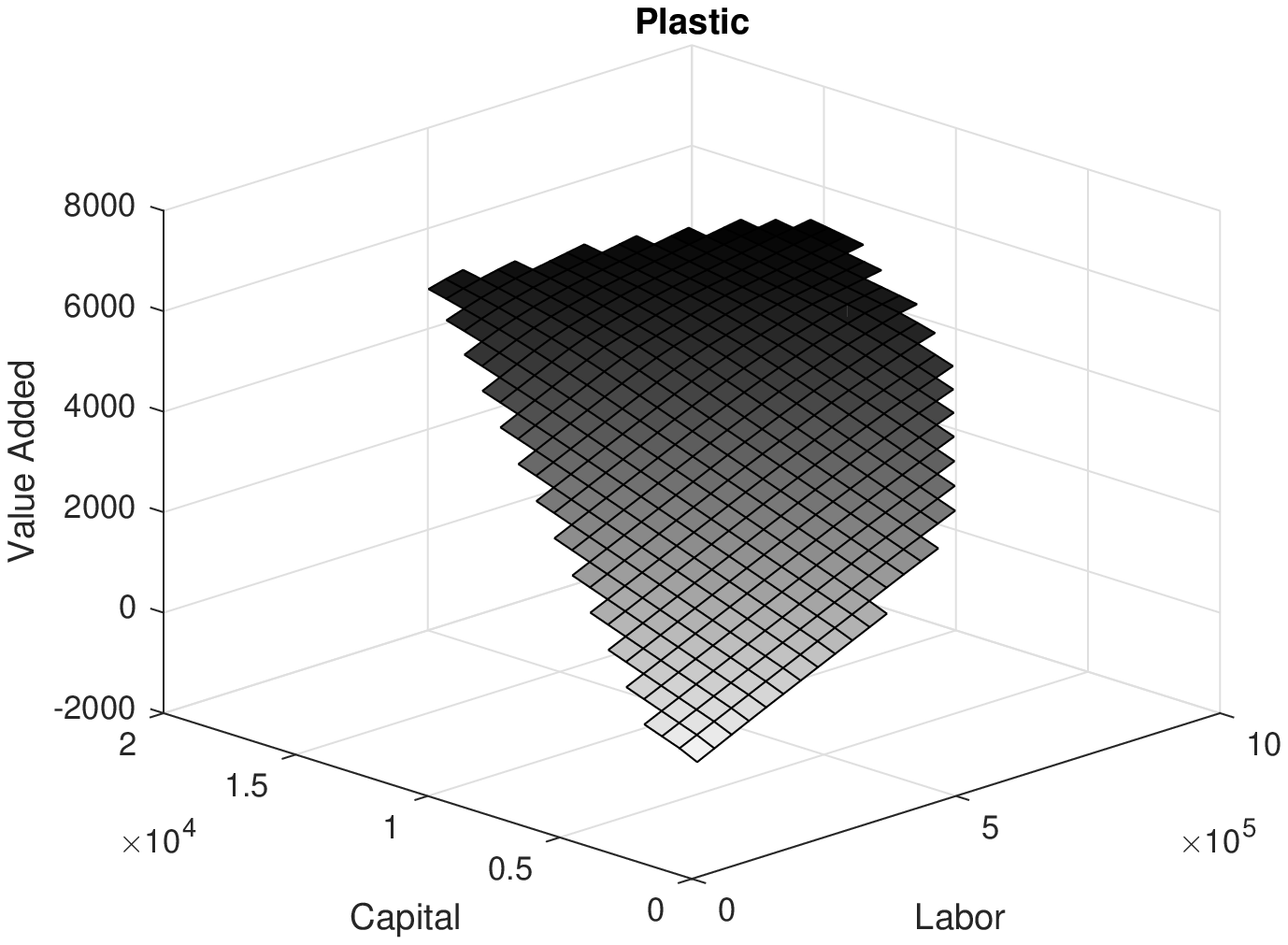}}
		\caption{Production function estimated by LL and SCKLS for the plastic industry (2520)}
		\label{fig:prod_func_plastic}
	\end{figure}
	\begin{figure}
		\centering
		\subfloat[Local Linear]{\includegraphics[width=0.5\textwidth]{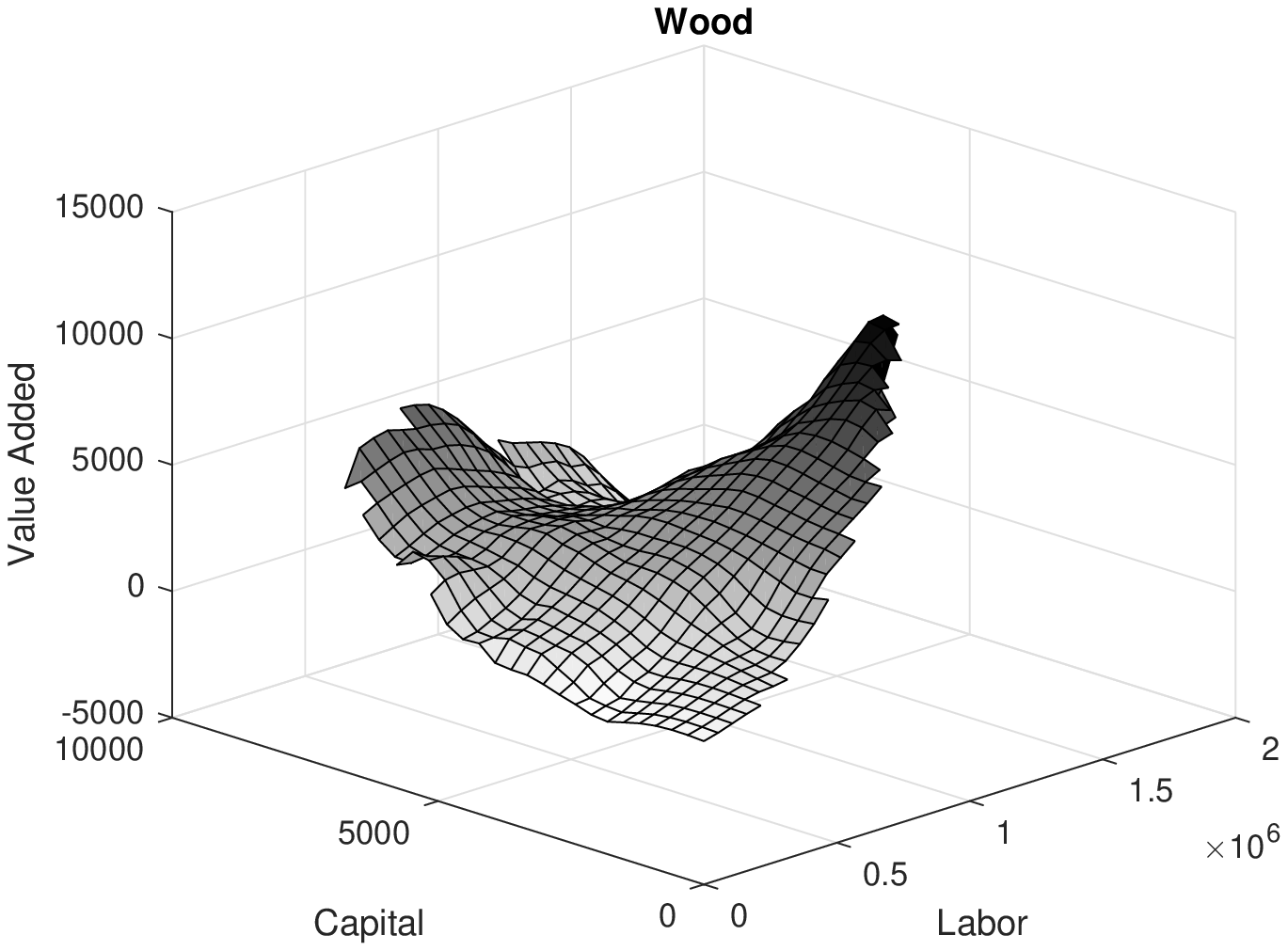}}
		\hfill
		\subfloat[SCKLS]{\includegraphics[width=0.5\textwidth]{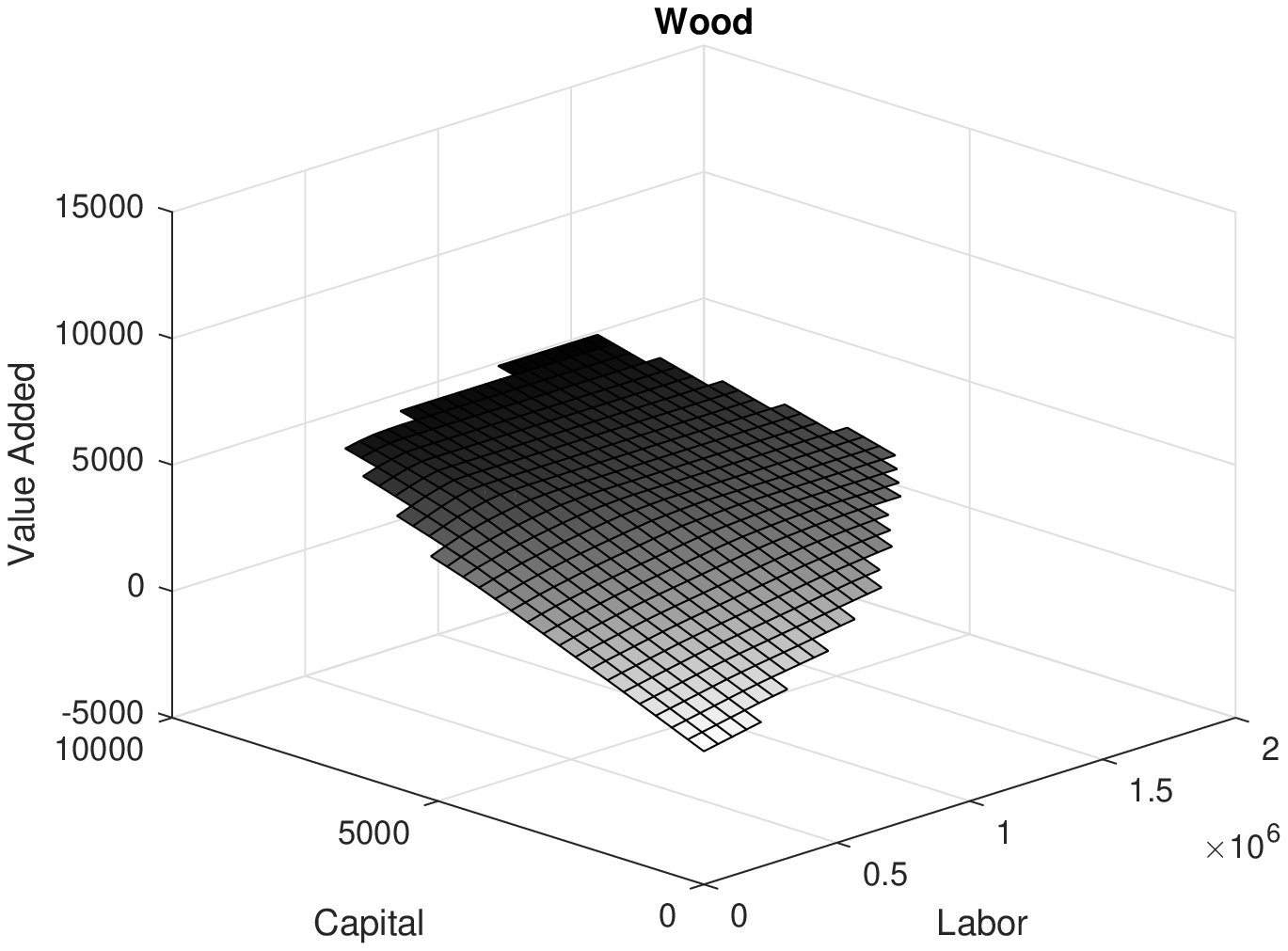}}
		\caption{Production function estimated by LL and SCKLS shape constraints for the wood industry (2010)}
		\label{fig:prod_func_wood}
	\end{figure}
	
	Table \ref{tab:14.Contextual var} reports the estimated coefficients for the exporting variables. In the plastic industry, the dummy variable for exporting is significant and positive while exports' share of sales is not. This indicates that the plants that export tend to produce more output than plants that do not export regardless of the export quantity. In contrast, the coefficient on the exports' share of sales is significant and positive in the wood industry while the dummy variable for exporting is not significant, indicating that establishments in the wood industry tend to be more productive the more they export. Thus, in both industries we find evidence of increased productivity for exporting firms.

	% Table generated by Excel2LaTeX from sheet 'Table14'
	\begin{table}[!htbp]
		\small
		\centering
		\caption{Coefficient of contextual variables from a 2-stage model.}
		\begin{tabular}{lrrrrr}
			\toprule
			& \multicolumn{2}{c}{Plastic (2520)} & \multicolumn{1}{c}{} & \multicolumn{2}{c}{Wood (2010)} \\
			\cmidrule{2-3}
			\cmidrule{5-6}
			& \multicolumn{1}{c}{\shortstack{Dummy of\\exporting}} & \multicolumn{1}{c}{\shortstack{Share of exporting\\in sales}} & \multicolumn{1}{c}{} & \multicolumn{1}{c}{\shortstack{Dummy of\\exporting}} & \multicolumn{1}{c}{\shortstack{Share of exporting\\in sales}} \\
			\midrule
			\multicolumn{1}{l}{Point estimate} & 334.5 & 303.7 &       & -763.0 & 4114 \\
			\multicolumn{1}{l}{95\% lower bound} & 148.7 & -334.3 &       & -1944 & 2568 \\
			\multicolumn{1}{l}{95\% upper bound} & 520.3 & 941.8 &       & 417.7 & 5660 \\
			\multicolumn{1}{l}{$p$-value} & 4.70$\times10^{-4}$ & 0.3493 &       & 0.2033 & 5.64$\times10^{-7}$ \\
			\bottomrule
		\end{tabular}%
		\label{tab:14.Contextual var}%
	\end{table}%

	Table \ref{tab:15.MPSS} reports the most productive scale size for the 10, 25, 50, 75, 90 percentiles of Capital/Labor ratio distribution of observed input. In both industries, the observed value added output is the largest for establishments with high capital to labor ratios, indicating that capital-intensive establishments have increased actual output. Furthermore, labor-intensive establishments have smaller most productive scale size in both industries. This is consistent with the theory of the firm, i.e. firms grow and become more capital intensive over time by automating processes with capital and using less labor.

	% Table generated by Excel2LaTeX from sheet 'Table15'
	\begin{table}[!htbp]
		\small
		\centering
		\caption{Most productive scale size for each capital/labor ratio.}
		\begin{tabular}{rrrr}
			\toprule
			& \multicolumn{3}{c}{Plastic (2520)} \\
			\cmidrule{2-4}
			& \multicolumn{1}{c}{\multirow{2}[0]{*}{MPSS Labor}} & \multicolumn{1}{c}{\multirow{2}[0]{*}{MPSS Capital}} & \multicolumn{1}{c}{Output} \\
			\multicolumn{1}{l}{Capital/Labor percentile} & \multicolumn{1}{c}{} & \multicolumn{1}{c}{} & \multicolumn{1}{c}{ (Value added)} \\
			\midrule
			\multicolumn{1}{l}{10th percentile} & 619580 & 519.1 & 3290 \\
			\multicolumn{1}{l}{25th percentile} & 529980 & 1344  & 3010 \\
			\multicolumn{1}{l}{50th percentile} & 529980 & 2604  & 3185 \\
			\multicolumn{1}{l}{75th percentile} & 529980 & 5617  & 3602 \\
			\multicolumn{1}{l}{90th  percentile} & 529980 & 10270 & 4248 \\
			\midrule
			\multicolumn{1}{l}{} & \multicolumn{3}{c}{Wood (2010)} \\
			\cmidrule{2-4}
			\multicolumn{1}{l}{} & \multicolumn{1}{c}{\multirow{2}[0]{*}{MPSS Labor}} & \multicolumn{1}{c}{\multirow{2}[0]{*}{MPSS Capital}} & \multicolumn{1}{c}{Output} \\
			\multicolumn{1}{l}{Capital/Labor percentile} & \multicolumn{1}{c}{} & \multicolumn{1}{c}{} & \multicolumn{1}{c}{ (Value added)} \\
			\midrule
			\multicolumn{1}{l}{10th percentile} & 2531100 & 741.6 & 1659 \\
			\multicolumn{1}{l}{25th percentile} & 1045000 & 1200  & 2142 \\
			\multicolumn{1}{l}{50th percentile} & 867250 & 2712  & 3470 \\
			\multicolumn{1}{l}{75th percentile} & 662700 & 4179  & 4682 \\
			\multicolumn{1}{l}{90th  percentile} & 458150 & 5644  & 5893 \\
			\bottomrule
		\end{tabular}%
		\label{tab:15.MPSS}%
	\end{table}%

	\section{Conclusion}
	\label{sec:7.conc}
	
	This paper proposed the SCKLS estimator that imposes shape constraints on a local polynomial estimator. We show the consistency and convergence rate of this new estimator under monotonicity and concavity constraints, as well as its relationship with CNLS and CWB. We also illustrate how to use SCKLS to validate the imposed shape constraints. %, and to test whether a function is affine. %Specifically, we show that CNLS is a special case of the SCKLS estimator. 
	In applications where out-of-sample performance is less critical and the boundary behavior is of less concern, such as regulation applications, the CNLS estimator may be preferable because of its simplicity. In contrast, in cases where out-of-sample performance is important, such as survey data, the SCKLS estimator appears to be more robust. Simulation results reveal the SCKLS estimator outperforms CNLS and LL in most scenarios. We propose and validate the usefulness of several extensions, including variable bandwidth and non-uniform griding, which are important to estimate functions with non-uniform input data set which is common in manufacturing survey and census data. We also propose a test for the imposed shape constraints based on SCKLS. Finally, we demonstrate the SCKLS estimator empirically using Chilean manufacturing data. We compute marginal productivity, marginal rate of substitution, most productive scale size and the effects of exporting, and provide several economic insights.
	
	One limitation of the proposed SCKLS estimator is its computation efficiency due to the large number of constraints. The algorithm we proposed for reducing constraints performs well, and we demonstrate the ability to solve large problems instances within a reasonable time. Furthermore, our simulation results show good functional estimates even with a rough grid. Consequently, we can make use of the flexibility of the evaluation points to reduce the computational time of the estimator.
	
	Potential future research could focus on the bandwidth selection methods. Typically, optimal bandwidth selection methods without shape constraints try to trade bias and variance to find the best estimator in terms of RMSE. Since the imposed shape restrictions already constrain the variance of the estimator to some extent, we expect that the optimal bandwidth in the SCKLS estimator will be smaller than the optimal unconstrained estimator. Further, if systematic inefficiency is present in the data, deconvoluting the residuals following the stochastic frontier literature would allow the investigation of a production frontier.
	
	%\section*{Acknowledgements}
	%We thank two anonymous reviewers and the Associate Editor for providing useful suggestions that helped improve this manuscript. We also thank Chris Parmeter and Jeff Racine for their helpful comments.

	\bigskip
	\begin{center}
		{\large\bf SUPPLEMENTARY MATERIAL}
	\end{center}
	
	\begin{description}
		
		\newcounter{app_section}
		\setcounter{app_section}{1}
		\refstepcounter{app_section}
		
		\item[Appendix:] The document contains: 
		 \begin{inparaenum}[(A)]
			  	\item extensions and the relationship between estimators;
			  	%\item comparison between piece-wise linear and smooth estimates;
			  	\item technical proofs of the theoretical results;
			  	\item a test of affinity using SCKLS;	
			  	\item an algorithm for SCKLS computational performance;
			  	\item comprehensive results of existing and additional numerical experiments;
			  	\item semiparametric model to integrate contextual variable; and
			  	\item details of the application to the Chilean manufacturing data.
		 \end{inparaenum}
		%\item[MATLAB code:] MATLAB code to compute SCKLS estimator described in the article. The package also contains some examples.
		
	\end{description}
%	\fi	
%	\newpage

    % Appendix Starts here----------------------------------------------------------------------------------
\setcounter{table}{0}
\setcounter{figure}{0}
\setcounter{equation}{0}
\setcounter{theorem}{0}
\setcounter{lemma}{0}
\setcounter{assumption}{0}
\renewcommand{\theequation}{A.\arabic{equation}} 
\renewcommand{\thefigure}{A.\arabic{figure}}
\renewcommand{\thetable}{A.\arabic{table}} 
\newcolumntype{P}[1]{>{\centering\arraybackslash}p{#1}}
\newtheorem{prop}{Proposition}

\AtAppendix{\counterwithin{lemma}{section}}
\AtAppendix{\counterwithin{prop}{section}}
\AtAppendix{\counterwithin{assumption}{section}}
%\AtAppendix{\counterwithin{experiment}{section}}

\setcounter{secnumdepth}{4}

    \newpage
    \appendix

\noindent{\huge\textbf{Appendix}}
	
    \vspace{1cm}
	 This appendix includes:
    \begin{itemize}
    	\item Extensions to SCKLS and a description of the relationship between SCKLS, CNLS and CWB (Appendix \ref{App:AppendixA}),
    	%\item Comparison between piece-wise linear and smooth estimates (Appendix \ref{App:AppendixB}).
    	\item Technical proofs of the theoretical results (Appendix \ref{App:AppendixC}).
        \item A test of affinity based on SCKLS (Appendix \ref{sec:5.2test})	
    	\item An algorithm for SCKLS computational performance (Appendix \ref{App:AppendixD}).
    	\item Comprehensive results of existing and additional numerical experiments (Appendix \ref{App:CompResults}).
    	\item Description of a semiparametric partially linear model to integrate contextual variable (Appendix \ref{App:AppendixI}).
    	%\item Additional experiments to show the robustness of the SCKLS estimator (Appendix \ref{App:AppendixJ}).
    	\item Details about the application to the Chilean manufacturing data (Appendix \ref{App:application})
    \end{itemize}
	
	\noindent%
	
	\newpage
	\spacingset{1.45} % DON'T change the spacing!
	
	\appendix
	
	\section{More on SCKLS,  CNLS and CWB} \label{App:AppendixA}
	% the \\ insures the section title is centered below the phrase: AppendixA
	In this section, we first give details on the extensions and practical considerations to SCKLS. We then mention some recently proposed estimators that are related to SCKLS, and make connections and comparisons among these methods. 
	
		\subsection{More on practical considerations and extensions to SCKLS}
	\label{App:Extensions}

	\subsubsection{SCKLS with general constraints}
    We focus on global concavity/convexity and monotonicity constraints in the main manuscript. But the SCKLS estimator can handle any types of shape constrained by imposing constraints on decision variables $\{a_i,\bm{b}_i\}_{i=1}^m$. We re-define the SCKLS estimator as
    \begin{equation}
	\begin{aligned}
	\label{eq:5.SCKLS_general}
	& \min_{\bm{a},\bm{b}}
	& & \sum_{i=1}^{m}\sum_{j=1}^{n}(y_j-a_i-(\bm{X}_j-\bm{x}_i)'\bm{b}_i)^2K\left(\frac{\bm{X}_j-\bm{x}_i}{\bm{h}}\right)\\
	& \mbox{subject to}
	& & l(\bm{x}_i)\leq\hat{g}^{(\bm{s})}(\bm{x}_i|\bm{\bm{a},\bm{b}})\leq u(\bm{x}_i), \; & i=1,\ldots,m\\
	\end{aligned}
	\end{equation}
	where $\bm{a} = (a_1,\ldots,a_m)'$ and $\bm{b} = (\bm{b}_1',\ldots,\bm{b}_m')'$. $l(\cdot)$ and $u(\cdot)$ represent lower and upper bounds at each evaluation point respectively. $\bm{s}$ denotes the order of partial derivative to each evaluation point $\bm{x}_i$.

	\subsubsection{SCKLS with Local Polynomial}
	\label{sec:3.1.SCKLSwLP}	
	With the proposed estimator in (\ref{eq:5.SCKLS_general}), we are only able to impose the constraints by using the functional estimate and/or first partial derivatives. For constraints involving a higher order of derivatives, we need to formulate SCKLS estimator with a higher order local polynomial function. For the multivariate local polynomial, we borrow the following notation from \citet{Masry1996multivariatelocal}.
	\[
	\begin{aligned}
	%\begin{split}
	&\bm{r}=(r_1,\ldots,r_d),\quad &&\bm{r}!=r_1!\times \cdots r_d!,\quad &&&\bar{\bm{r}}=\sum_{k=1}^{d}r_k, \\
	&\bm{x}^{\bm{r}} = x_1^{r_1}\times\cdots x_d^{r_d},\quad &&\sum_{0\leq\bar{\bm{r}}\leq p}=\sum_{k=0}^{p}\sum_{r_1=0}^{k}\cdots\sum_{r_d=0}^{k},\quad&&&\mbox{and}\\
	\quad&\left(D^{\bm{r}} g\right)\left(\bm{x}\right)=\frac{\partial^{\bm{r}}g(\bm{x})}{\partial x_1^{r_1}\cdots \partial x_d^{r_d}}
	%\end{split}
	\end{aligned}
	\]
	With this notation, we can approximate any function $g:\mathbb{R}^d \rightarrow \mathbb{R}$ locally (around any $\bm{x}$) using a multivariate polynomial of total order $p$, given by
	\begin{equation}
	\label{eq:LP}
	g(\bm{z}):=\sum_{0\leq \bar{\bm{r}}\leq p}\frac{1}{\bm{r}!}\left(D^{\bar{\bm{r}}}g\right)(\bm{x}) \left(\bm{z}-\bm{x}\right)^{\bar{\bm{r}}}.
	\end{equation}
	We now define the SCKLS estimator with a local polynomial function of order $p$ as follows:
	\begin{equation}
	\begin{aligned}
	\label{eq:6.2.SCKLS_general_LP}
	& \min_{\bm{b_i}}
	& & \sum_{i=1}^{m}\sum_{j=1}^{n}\left(y_j-\sum_{0\leq\bar{\bm{r}}\leq p}\bm{b}_{i}'(\bm{X}_j-\bm{x}_i)^{\bar{\bm{r}}}\right)^2K\left(\frac{\bm{X}_j-\bm{x}_i}{\bm{h}}\right)\\
	& \mbox{subject to}
	& & l(\bm{x}_i)\leq\hat{g}^{(\bm{s})}(\bm{x}_i|\bm{\bm{b}})\leq u(\bm{x}_i), \; & i=1,\ldots,m\\
	\end{aligned}
	\end{equation}
	where $\bm{b}_{i}$ is the functional or derivative estimates at each evaluation points and $\bm{b} = (\bm{b}_1',\ldots,\bm{b}_m')'$. When we select $p=1$, then the problem becomes exactly same as the proposed estimator in (\ref{eq:5.SCKLS_general}). This extension allows us to make the proposed methods more general and applicable for other applications of shape restricted functional estimation in which higher order derivative restricts may be required. From a computational complexity point of view, it is still optimizing a quadratic objective function within a convex solution space, and thus, the problem is still typically solvable within polynomial time.
	
	As demonstrated in \citet{li2007nonparametric}, the rate of convergence of local polynomial estimator is the same for $p=1$ and $p=2$. From a theoretical perspective, one could attempt to select a polynomial estimator with $p \ge 3$ to improve its convergence performance (at least theoretical). But that would require much stronger assumption on the smoothness of $g_0$, and would lead to additional computational burden\footnote{While the optimization problem is still polynomial time solvable, the number of decision variables would increase and the constraint matrix would become significantly more dense, lending to computational challenges.}. Our experience suggests that SCKLS inherits these properties from the local polynomial method. Therefore, in practice, with only monotonicity and concavity/convexity constraints, we feel that it suffices to consider SCKLS with $p=1$ (i.e. local linear). 
	
	\subsubsection{SCKLS with $k$-nearest neighbor}
	Our primary application of interest is production functions estimated for census manufacturing data where the input distributions are often highly skewed meaning there are many small establishments, but relatively few large establishments\footnote{An establishment is defined as a single physical location where business is conducted or where services or industrial operations are performed.}. To address this issue, we propose to use a $k$-nearest neighbor ($k$-NN) approach in SCKLS which we will refer to as SCKLS $k$-NN which is in spirit similar to the extension to the CWB-type estimator proposed by \cite{li2016nonparametric}. The $k$-NN approach uses a smaller bandwidth for smoothing in dense data regions and a larger bandwidth when the data is sparse. For a further description of the method, see for example \cite{li2007nonparametric}. For any given $k$, the formulation of SCKLS $k$-NN with monotonicity and concavity constraints leads to a different weighting scheme in the objective function, as illustrated in the following.
	
	\begin{equation}
	\begin{aligned}
	\label{eq:7.SCKLSkNN}
	& \min_{a_i,\bm{b_i}}
	& & \sum_{i=1}^{m}\sum_{j=1}^{n}(y_j-a_i-(\bm{X_j}-\bm{x_i})'\bm{b_i})^2w\left(\frac{\|\bm{X_j}-\bm{x_i}\|}{R_{\bm{x_i}}}\right)\\
	& \mbox{subject to}
	& & a_i-a_l\geq \bm{b_i}'(\bm{x_i}-\bm{x_l}), \; & i,l=1,\ldots,m\\
	&
	& & \bm{b_i}\geq 0, \; & i=1,\ldots,m
	\end{aligned}
	\end{equation}
	where $w(\cdot)$ is a general weight function, $\|\cdot\|$ is the Euclidean norm and $R_{\bm{x_i}}$ denotes the Euclidean distance between $\bm{x_i}$ and $k$-th nearest neighbor of $\bm{x}_i$ among the set of all covariates $\{\bm{X}_j\}_{j=1}^n$. In practice, $k$ can be chosen by leave-one-out cross validation (LOOCV).

	\subsubsection{SCKLS with non-uniform grid}
	\label{App:nonunif}
		As noted in the paper, the SCKLS estimator requires the user to specify the number and locations of the evaluation points. We can also address the input skewness issue by constructing the evaluation points differently, using a non-uniform grid method. To do so, we first use kernel density estimation to estimate the density function for each input dimension. Then we take the equally spaced percentiles of the estimated density function and construct non-uniform grid. Figure~\ref{fig:2.nonunif grid} demonstrates how the non-uniform grid are constructed for the 2-dimensional case. In this example, we set the minimum and maximum of the observed inputs (with respect to each coordinate) as the edge of the grid, and compute equally spaced percentile. When the support of the covariates is non-regular (e.g. not a hyperrectangle), we shall limit ourselves to evaluation points inside the convex hull of $\{\bm{X}_j\}_{j=1}^n$.
	\begin{figure}[!ht]
		\begin{center}
			\includegraphics[width=3.8in]{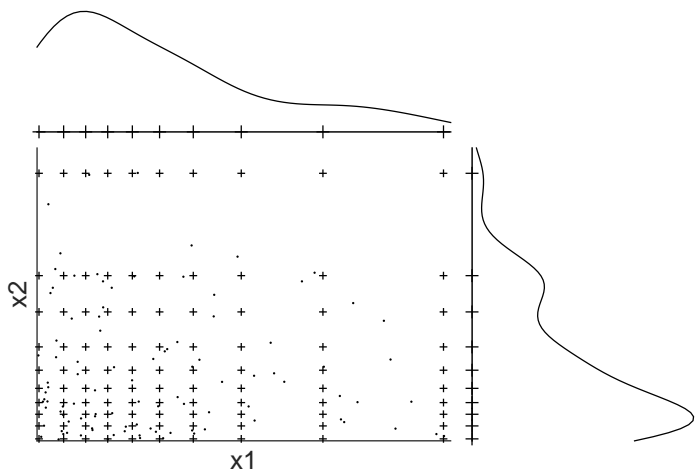}
		\end{center}
		\caption{Example of non-uniform grid with kernel density estimation. \label{fig:2.nonunif grid}}
	\end{figure}
	
	\newpage
	\subsection{Some related work}
		
	\subsubsection{Convex Nonparametric Least Squares (CNLS)}
	\label{sec:2.2.CNLS}
	
	\cite{kuosmanen2008representation} extends Hildreth's least squares approach to the multivariate setting with a multivariate input vector, and coins the term ``Convex Nonparametric Least Squares'' (CNLS)\footnote{A related maximum likelihood formulation was proposed by \cite{banker1992maximum}, with its consistency proved by \cite{sarath1997consistency}.}. CNLS builds upon the assumption that the true but unknown production function $g_0$ belongs to the class of monotonically increasing and globally concave functions, denoted by $G_2$ in this paper. Given the observations $\{\bm{X}_j,y_j\}_{j=1}^n$, a set of unique fitted values, $\hat{y}_j=\hat{\alpha}_j+\bm{\hat{\beta}_j}\bm{X}_j$, can be found by solving the quadratic programming (QP) problem
	\begin{equation}
	\begin{aligned}
	\label{eq:1.CNLS}
	& \min_{\alpha,\bm{\beta}}
	& & \sum_{j=1}^{n}(y_j-(\alpha_j+\bm{\beta}_j'\bm{X}_j))^2\\
	& \mbox{subject to}
	& & \alpha_j+\bm{\beta}_j'\bm{X}_j\leq\alpha_l+\bm{\beta}_l'\bm{X}_j, \; & j,l=1,\ldots,n\\
	&
	& & \bm{\beta_j}\geq 0, \; &j=1,\ldots,n
	\end{aligned}
	\end{equation}
	where $\alpha_j$ and $\bm{\beta}_j$ define the intercept and slope parameters that characterize the estimated set of hyperplanes. The inequality constraints in (\ref{eq:1.CNLS}) can be interpreted as a system of Afriat inequalities \citep{afriat1972efficiency, varian1984nonparametric} to impose concavity constraints. We emphasize that CNLS does not assume or restrict the domain $G_2$ to only  piece-wise affine functions.
	% or assume the smoothness of the function. 
	%Thus, the estimate resulting from (\ref{eq:1.CNLS}) is the most flexible\footnote{Typically, ``smooth'' refers to the existence of a certain number of derivatives. The CNLS estimator does not make assumptions on smoothness conditions, both implicitly and explicitly, thus is ``more flexible'' compared to other methods in this perspective.} estimator that satisfies the shapes constraints and that the data supports. 
	We also note that the functional estimates resulting from  (\ref{eq:1.CNLS}) is unique only at the observed data points. %Though the resulting estimate from (\ref{eq:1.CNLS}) is typically piece-wise affine, we note that piece-wise affine function can be used to approximate functions of many different forms. 
	In addition, when $d=1$, \cite{chen2016convex} and \cite{ghosal2016univariate} proved that the CNLS-type estimator attains $n^{-1/2}$ pointwise rate of convergence if the true function is piece-wise linear.
	
	Finally, we remark that CNLS is related to the method of sieves \citep{grenander1981abstract, ChenQiu2016} in the following way. The estimator could be rewritten as 
	\[
	\hat{g}_n \in \argmin_{g \in \mathcal{G}^n} \frac{1}{n} \sum_{j=1}^n(y_j - g(\bm{X}_j))^2,
	\]
	where 
	$
	\mathcal{G}^n = \{g:\mathbb{R}^d\rightarrow\mathbb{R}\ |\ g(\bm{x}) = \min_{j\in\{1\ldots,n\}} (\alpha_j + \bm{\beta}_j'\bm{x}), \mbox{ with } \bm{\beta}_j \ge 0 \mbox{ for } j = 1,\ldots, n\}.
	$
	However, since the sets $\mathcal{G}^1,\mathcal{G}^2, \ldots$ are not compact, most known results on sieves do not directly apply here.

	\subsubsection{Constrained Weighted Bootstrap (CWB)}
	\label{sec:2.3.CWB}
	\paragraph{Introduction}
	
	\cite{hall2001nonparametric} proposed the monotone kernel regression method in univariate function. \cite{du2013nonparametric} generalized this model to handle multiple general shape constraints for multivariate functions, which they refer to as Constrained Weighted Bootstrap (CWB). CWB estimator
	is constructed by introducing weights for each observed data point. The weights are selected to minimize the distance to unconstrained estimator while satisfying the shape constraints. The function is estimated as
	\begin{equation}
	\begin{aligned}
	\label{eq:2.def_g}
	\hat{g}(\bm{x}|\bm{p})=\sum_{j=1}^{n}p_jA_j(\bm{x})y_j
	\end{aligned}
	\end{equation}
	where $\bm{p}=(p_1,\ldots,p_n)'$, $p_j$ is the weights introduced for each observation and $A_j(\bm{x})$ is a local weighting matrix (e.g. local linear kernel weighting matrix). \cite{du2013nonparametric} relaxed the restriction imposed by \cite{hall2001nonparametric}  that $p_j$ is non-negative and propose to calculate $\bm{p}$ by minimizing its distance to unrestricted weights, $\bm{p}_u=(1/n,\ldots,1/n)'$, under derivative-based shape constraints\footnote{The use of the equality constraint $\sum_{j}p_j=1$ in \cite{du2013nonparametric} is a typo, and this condition is not used by them. In fact, it may harm the estimation procedure. Our empirical results show that this equality constraint only makes difference in very few cases and  the difference is typically small.}. The problem is formulated as follows.
	\begin{equation}
	\begin{aligned}
	\label{eq:3.CWB}
	& \min_{\bm{p}}
	& & D(\bm{p})=\sum_{j=1}^{n}(p_j-p_u)^2=\sum_{j=1}^{n}(p_j-1/n)^2\\
	& \mbox{subject to}
	& & l(\bm{x}_i)\leq\hat{g}^{(\bm{s})}(\bm{x}_i|\bm{p})\leq u(\bm{x}_i), \; & i=1,\ldots,m\\
	\end{aligned}
	\end{equation}	
	where $\bm{x_i}$ represents a set of points for evaluating constraints, the elements of $\bm{s}$ represent the order of partial derivative, and $g^{\bm{s}}(\bm{x})=[\partial^{s_1}g(\bm{x})\cdots\partial^{s_r}g(\bm{x})]/[\partial x_1^{s_1}\cdots\partial x_r^{s_r}]$ for $\bm{s}=(s_1,s_2,\ldots,s_r)$. %Here we use the notation, $l(\bm{x})$ and $u(\bm{x})$, to emphasize that the bounds needed to impose restrictions on functional derivatives are typically functions of the entire data set $\bm{x}$. 
	Here the shape restrictions (e.g. concavity/convexity and monotonicity constraints) are imposed at a set of evaluation points $\{\bm{x}_i\}_{i=1}^m$ through setting appropriate lower and upper bounds to the corresponding partial derivatives of the function.  
	One way to interpret the CWB estimator is as a two-step process: 1) estimate an unconstrained kernel estimator; 2) find the shape constrained function that is as close as possible (as measured by the Euclidean distance in $p$-space) to the unconstrained kernel estimator. Based on our experience, CWB tends to suffer from computational difficulties and occasionally poor estimates in small samples. We suggest  changing the objective function to minimize the distance from the estimated function to the observed data. This modification seems to improve the estimates empirically as shown in Appendix~\ref{App:CompResults}. %Since to our best knowledge, there does not exist a publicly available computing package for CWB, we now elaborate on how we implement the CWB for the purpose of conducting some comparisons.
	
	\paragraph{CWB estimator that minimize the distance from the observed data}\label{App:AppendixA2}
	
	%Since the objective function of the CWB estimator potentially pulls the functional estimate away from the true function and towards the unrestricted estimate, the CWB estimator may suffer from a finite sample bias. To avoid this problem, 
	We propose an extension of the CWB estimator by converting the objective function from $p$-space to $y$-space. Instead of minimizing the distance between the unconstrained estimator and the shape restricted functional estimate by minimizing the distance between the two functions in $p$-space, we propose to minimize the distance between the observed vector of $\bm{y}$ and the shape restricted functional estimates in $y$-space. The estimator,  which we shall refer to as CWB in $y$-space, is formulated as follows:
	\begin{equation}
		\begin{aligned}
		\label{eq:1.CWB}
		&  \min_{\bm{p}}
		& & D_y(\bm{p})=\sum_{j=1}^{n}(y_j-\hat{g}(\bm{X}_j|\bm{p}))^2\\
		& \mbox{subject to}
		& & l(\bm{x}_i)\leq\hat{g}^{(\bm{s})}(\bm{x}_i|\bm{p})\leq u(\bm{x}_i), \; & i=1,\ldots,m,\\
		&
		& & \sum_{j=1}^{n}p_j=1.
		\end{aligned}
	\end{equation}
	
	Since the objective function is not necessarily convex in $\bm{p}$, this problem is a general nonlinear optimization problem which is harder to solve.

	\paragraph{Calculating the first partial derivative of $\hat{g}(\bm{x}|\bm{p})$ for CWB}

	\cite{du2013nonparametric} proposed the CWB estimator which requires estimating the first partial derivatives of unconstrained functional estimates, $\hat{g}^{(1)}(\bm{x}|\bm{p})$. Here, we test two different methods of calculating the partial derivatives. The first method is to calculate the numerical derivative, $\hat{g}^{(1)}(\bm{x}|\bm{p})=\frac{\hat{g}(\bm{x}+\Delta|\bm{p})-\hat{g}(\bm{x}|\bm{p})}{\Delta}$, to obtain the approximated derivative estimate. \cite{racine2016local} shows that the numerical derivative is very close to the analytic derivative. The second method is to use the slope estimates of local linear estimator directly as a proxy for the first partial derivative. We evaluate the performance of CWB in $p$-space estimator with these two different methods. Table~\ref{tab:A1.CWB derivative} and Table~\ref{tab:A2.CWB derivative} summarize the RMSE performance against the true function on the observed points and the evaluation points respectively. The experimental setting is based on Experiment~\ref{exp:1} in Section~\ref{sec:5.simulation}.
	
	% Table generated by Excel2LaTeX from sheet 'TableA1'
	\begin{table}[ht]\
		\small
		\centering
		\caption{RMSE on observation points for different methods to obtain $\hat{g}^{(1)}(\bm{x}|\bm{p})$.}
		\begin{tabular}{llrrrrr}
			\toprule
			&       & \multicolumn{5}{c}{Average RMSE on the observation points} \\
			\multicolumn{2}{c}{Number of observations} & 100   & 200   & 300   & 400   & 500 \\
			\midrule
			\multicolumn{1}{c}{\multirow{2}[0]{*}{2-input}} & Numerical derivative & \textbf{0.260} & \textbf{0.163} & \textbf{0.143} & \textbf{0.153} & \textbf{0.164} \\
			\multicolumn{1}{c}{} & Slope estimates of LL & 0.421 & 0.357 & 0.284 & 0.306 & 0.293 \\
			\midrule
			\multicolumn{1}{c}{\multirow{2}[0]{*}{3-input}} & Numerical derivative & \textbf{0.236} & \textbf{0.256} & \textbf{0.208} & \textbf{0.246} & \textbf{0.240} \\
			\multicolumn{1}{c}{} & Slope estimates of LL & 0.356 & 0.427 & 0.336 & 0.294 & 0.279 \\
			\midrule
			\multicolumn{1}{c}{\multirow{2}[0]{*}{4-input}} & Numerical derivative & \textbf{0.259} & \textbf{0.226} & \textbf{0.222} & \textbf{0.216} & \textbf{0.210} \\
			\multicolumn{1}{c}{} & Slope estimates of LL & 0.388 & 0.397 & 0.276 & 0.261 & 0.259 \\
			\bottomrule
		\end{tabular}%
		\label{tab:A1.CWB derivative}%
	\end{table}%
	
	% Table generated by Excel2LaTeX from sheet 'TableA2'
	\begin{table}[ht]
		\small
		\centering
		\caption{RMSE on evaluation points for different methods to obtain $\hat{g}^{(1)}(\bm{x}|\bm{p})$.}
		\begin{tabular}{llrrrrr}
			\toprule
			&       & \multicolumn{5}{c}{Average RMSE on the evaluation points} \\
			\multicolumn{2}{c}{Number of observations} & 100   & 200   & 300   & 400   & 500 \\
			\midrule
			\multicolumn{1}{c}{\multirow{2}[0]{*}{2-input}} & Numerical derivative & \textbf{0.284} & \textbf{0.188} & \textbf{0.157} & \textbf{0.176} & \textbf{0.193} \\
			\multicolumn{1}{c}{} & Slope estimates of LL & 0.445 & 0.387 & 0.321 & 0.334 & 0.323 \\
			\midrule
			\multicolumn{1}{c}{\multirow{2}[0]{*}{3-input}} & Numerical derivative & \textbf{0.309} & \textbf{0.355} & \textbf{0.272} & \textbf{0.331} & \textbf{0.271} \\
			\multicolumn{1}{c}{} & Slope estimates of LL & 0.438 & 0.507 & 0.403 & 0.371 & 0.363 \\
			\midrule
			\multicolumn{1}{c}{\multirow{2}[0]{*}{4-input}} & Numerical derivative & \textbf{0.408} & \textbf{0.381} & \textbf{0.354} & \textbf{0.333} & \textbf{0.308} \\
			\multicolumn{1}{c}{} & Slope estimates of LL & 0.530 & 0.535 & 0.396 & 0.387 & 0.368 \\
			\bottomrule
		\end{tabular}%
		\label{tab:A2.CWB derivative}%
	\end{table}%

	The results show that CWB using the numerical derivative performs better than CWB using the slope estimates from the local linear kernel estimator particularly when the sample size is small.

	\subsection{A comparison between SCKLS, CNLS and CWB}\label{App:AppendixA3}
	
	Figure~\ref{fig:1.comparison} is meant to be illustrative of the relationship between the SCKLS, CNLS and CWB estimators in a two-dimensional estimated $\epsilon$-space where there are more than two observations, but for the rest of the $n-2$ observations, their estimated $\epsilon_j$s are held fix. The gray area indicates the cone of concave and monotonic functions. CNLS estimates a monotonic and concave function while minimizing the sum of squared errors, that is, minimizing the distance from the origin to the cone in the estimated $\epsilon$-space. CWB estimates a monotonic and concave function by finding the closest point, measured in $p$-space, on the cone of concave and monotonic functions to unconstrained kernel estimate. SCKLS minimizes a weighted function of estimated errors, and therefore avoids overfitting the observed data. However, as shown in \ref{App:AppendixC1}, SCKLS can be interpreted as minimizing the weighted distance from the unconstrained local linear kernel estimator to the cone of concave and monotonic functions. 
	\begin{figure}
		\begin{center}
			\includegraphics[width=4.5in]{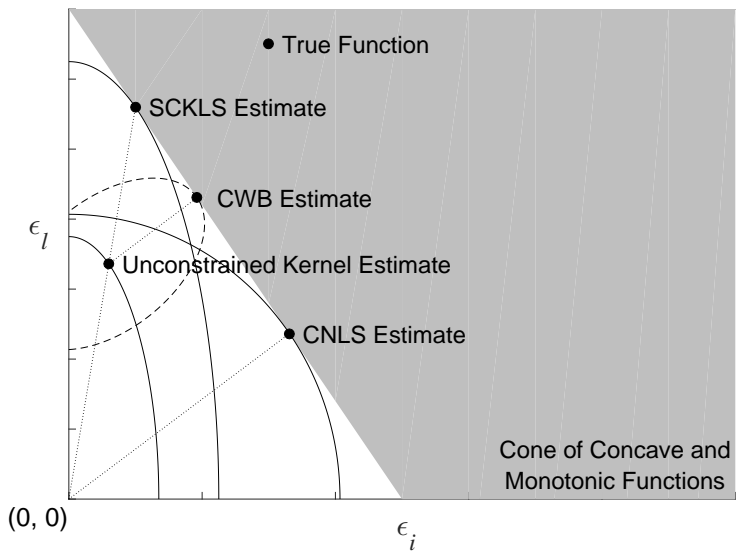}
		\end{center}
		\caption{Comparison of different estimators in the estimated-$\epsilon$-space. \label{fig:1.comparison}}
	\end{figure}
	
	\subsubsection{CNLS as a Special Cases of SCKLS}
	\label{sec:4.2.special}
	%Theorem \ref{thm:1.consistency} establishes SCKLS is a consistent estimator. 
	Let $\hat{g}_n$ and $\hat{g}_n^{CNLS}$ denote the SCKLS estimator and the CNLS estimator respectively. We will next examine the relationship between them. 
	\begin{assumption}
		\leavevmode
		\label{ass:4}	
		The set of evaluation points is equal to the set of sample input vectors, i.e. $m=n$ and $\bm{x}_i=\bm{X}_i$ for $i=1,\ldots,n$.
	\end{assumption}
	\begin{prop}
		\label{thm:4.SCKLS-CNLS}
		Suppose that Assumption~\ref{ass:4} holds. Then, for any $n$, when the vector of bandwidth goes to zero, i.e. $\|\bm{h}\|\rightarrow \mathbf{0}$ (where $\bm{h}=(h_1,\ldots,h_d)'$), the SCKLS estimator $\hat{g}_n$ converges to the CNLS estimator $\hat{g}_n^{CNLS}$ pointwise at $\bm{X}_1,\ldots,\bm{X}_n$.
	\end{prop}	
		
   Proposition~\ref{thm:4.SCKLS-CNLS} essentially says that CNLS can be viewed as a special case of SCKLS. Note that in comparison to the CNLS estimator, our SCKLS estimator has tuning parameters, which to some extent control the bias--variance tradeoff (in a non-trivial way given the shape restrictions). For reasonable values of these tuning parameters, SCKLS estimator performs better than CNLS. See also Section~\ref{sec:5.simulation} of the main manuscript. 
	This is especially true for the estimates close to the boundary of the input space, where imposing the shape constraint alone could lead to severe overfitting of the data, and thus biased estimates. Indeed, in view of Theorem~\ref{thm:3.uconsistency} (from the main manuscript), we have that $\sup_{\bm{S}}\big|\hat{g}_n(\bm{x})-g_0(\bm{x})\big| = o_p(1)$, while on the other hand,  $\sup_{\bm{S}}\big|\hat{g}_n^{CNLS}(\bm{x})-g_0(\bm{x})\big|$ does not converge to zero in probability.
	
	Additional equivalence results can also be shown. Proposition~\ref{thm:5.SCKLS-OLS} shows the equivalence of linear regression subject to monotonicity constraints and the SCKLS estimator when the bandwidth vector approaches infinity.
	\begin{prop}
		\label{thm:5.SCKLS-OLS}
		Given Assumption~\ref{ass:1}(\ref{ass:1.5}). For any given $n$, when the bandwidth vector goes to infinity (i.e. $\min_{k=1,\ldots,d} h_k \rightarrow \infty$), the SCKLS estimator converges to the least squares estimator of the linear regression model subject to monotonicity constraints.
	\end{prop}

	\subsubsection{CWB in y-space as a Special Cases of SCKLS}
	\label{sec:1.1}
	%Theorem \ref{thm:1.consistency} establishes SCKLS is a consistent estimator. 
	Let $\hat{g}_n$ and $\hat{g}_n^{CWBY}$ denote the SCKLS estimator and the CWB y-space estimator respectively. We will next examine the relationship between them. 
% 	\begin{assumption}
% 		\leavevmode
% 		\label{ass:4}	
% 		The set of evaluation points is equal to the set of sample input vectors, i.e. $m=n$ and $\bm{x}_i=\bm{X}_i$ for $i=1,\ldots,n$.
% 	\end{assumption}
	\begin{prop}
		\label{thm:1.SCKLS-CWBY}
		Suppose that Assumption~\ref{ass:4} holds. Then, for any $n$, when the vector of bandwidth goes to zero for both the SCKLS estimator and the CWB in y-space estimator, i.e. $\|\bm{h}\|\rightarrow \mathbf{0}$ (where $\bm{h}=(h_1,\ldots,h_d)'$), the SCKLS estimator $\hat{g}_n$ converges to the CWB in y-space estimator $\hat{g}_n^{CWBY}$ pointwise at $\bm{X}_1,\ldots,\bm{X}_n$.
	\end{prop}
	Proposition~\ref{thm:1.SCKLS-CWBY} states that SCKLS and CWB in $y$-space estimators converge to the same estimates as $\|\bm{h}\|\rightarrow \mathbf{0}$. Combining with Proposition~\ref{thm:4.SCKLS-CNLS}, CNLS can be viewed as a special case of SCKLS and CWB in $y$-space.
		
 \subsubsection{The relationship between CWB in $p$-space and SCKLS}
	\label{sec:1.2}

Again start from the SCKLS estimator, and in view of Assumption \ref{ass:1}~(\ref{ass:1.5}), for any sufficiently small $\bm{h}$, we have
\[
K\left(\frac{\bm{X}_j-\bm{x}_i}{\bm{h}}\right) = \begin{cases}
0 & \mbox{if } \bm{x}_i\neq\bm{X}_j, \\
K(\bm{0}) & \mbox{if } \bm{x}_i=\bm{X}_j,
\end{cases} \mbox{ for }  \forall i,j.	
\]
Then, the objective function of the SCKLS estimator (3) is equal to $\sum_{j=1}^{n}(y_j-a_j)^2 K(\bm{0})$, and thus 
\begin{equation*}
%\label{eq:9.objcnv}
\argmin_{a_1,\bm{b}_1,\ldots,a_n,\bm{b}_n}\sum_{j=1}^{n}(y_j-a_j)^2 K(\bm{0}) =\argmin_{a_1,\ldots,a_n}\sum_{j=1}^{n}(y_j-a_j)^2 =\argmin_{a_1,\ldots,a_n} L(g(a_j))
\end{equation*}

where $L(\cdot) = \sum_{j=1}^{n}(\cdot)^2 $ is the squared error loss function, $g(a_j) = y_j-a_j$	the definition of the residual. 

Alternatively now consider the objective function of CWB, specifically $D(\bm{p})=\sum_{j=1}^{n}(p_u-p_j)^2=\sum_{j=1}^{n}(1/n-p_j)^2 = L(m(g(p_j)))$. And let $L(\cdot)$ continue to be defined as above as the squared error lost function and $g(p_j)$ as the definition of the residual. This implies that $m(\cdot) = \frac{\cdot}{y_j n}$. Therefore, the CWB estimator can be interpreted as a projection of a local polynomial estimator to the cone of functions which are monotonic and concave in which the direction of projection minimizes a specific weighting of the unconstrained local polynomial residuals in which the weights are defined as $\frac{1}{y_j n}$. Therefore, even if the vector of bandwidth goes to zero for the CWB in $p$-space estimator, i.e. $\|\bm{h}\|\rightarrow \mathbf{0}$ (where $\bm{h}=(h_1,\ldots,h_d)'$), the CWB estimator and CNLS are not equivalent because the $y_j$ in the denominator of the weights is not a function of the bandwidth. 
 
	\subsubsection{On the computational aspects}
	
	We also compare the computational burden of each estimators. Table \ref{tab:A4.sizeQuad} shows the size of quadratic programming problems of each estimators: SCKLS, CNLS and CWB. The size of a quadratic programming problem of the SCKLS estimator is fully controllable because the number of decision variables and constraints is a function of the number of evaluation points and independent of the number of observed points. Because of this, we can solve large-scale problems with $n>100,000$ using the SCKLS estimator while other shape constrained nonparametric estimators might face prohibitive computational difficulties without any data pre-processing.
	\begin{table}[ht]
		\begin{doublespace}
		\small
		\centering
		\caption{The size of quadratic programming problems of each estimator.}
		\begin{tabular}{p{2.7in}P{.6in}P{.6in}P{.6in}}
			\toprule
			& SCKLS & CNLS  & CWB  \\
			\midrule
			Number of decision variables  & $m(d+1)$ & $n(d+1)$ & $n$ \\
			Number of global concavity constraints & $m(m-1)$ & $n(n-1)$ & $m(m-1)$ \\ 
			\bottomrule
		\end{tabular}%
		\label{tab:A4.sizeQuad}%
		\end{doublespace}
	\end{table}%

	\section{Technical proofs}
	\label{App:AppendixC}
    \subsection{Summary of the proof strategy}	
    \label{App:AppendixC_summary}
	Theorems~\ref{thm:2.rate}-- \ref{thm:4.missconsistency} concern the consistency and convergence rate of the SCKLS estimator and serve as the primary results in our theoretical development. As such, before presenting the technical details,  we summarize our proof strategy as follows:
\begin{enumerate}
    \item We rewrite the SCKLS estimator, after some manipulations, as the projection of the local linear estimator to a convex cone of monotonic and concave functions under a certain norm. More precisely, the SCKLS estimator 
	\[
	\hat{g}_n \in \mathrm{argmin}_{g \in G_2} 	\|g - \tilde{g}_n\|^2_{n,m},
	\]
	where $\tilde{g}_n$ is the local linear estimator, $G_2$ is the set that contains all the concave and increasing functions, and $\|\cdot\|_{n,m}$ is a norm defined in detail later in Appendix B.2.
	
    \item (Theorem~\ref{thm:2.rate}). Let $\hat{g}_n$ be the SCKLS estimator and $g_0 \in G_2$ be the truth. Using the new formulation of SCKLS above, we see that 
    \[
    \|\hat{g}_n - \tilde{g}_n\|_{n,m} \le \|g_0 - \tilde{g}_n\|_{n,m}.
    \]
    Moreover,  by the triangular inequality, we have that 
    \[
    \|\hat{g}_n - g_0\|_{n,m} \le  \|\hat{g}_n - \tilde{g}_n\|_{n,m}  +  \|\tilde{g}_n - g_0\|_{n,m}  \le 2\| \tilde{g}_n - g_0\|_{n,m}.
    \]
    Using the results on the uniform consistency of the local linear estimator (e.g. \citet{fan2016}, see our Lemma~\ref{lem:1} and Lemma~\ref{lem:2}), we can bound the RHS of the triangle inequality equation by $O_p(n^{-2/(4+d)}\log n) = o_p(1)$. Consequently, $\|\hat{g}_n - g_0\|_{n,m}$ converges to zero at the same rate. To complete the proof, we show that the discrete $L_2$ distance between $\hat{g}_n$ and $g_0$ is bounded above by a constant times $\|\hat{g}_n - g_0\|_{n,m}$. 
    
    % As explained below, this also implies that $\hat{g}_n$ converges to $g_0$ in $L_2$.
    % we note that %Note that this is where we directly used \citet{fan2016} in this proof.
    
    \item (Theorem~\ref{thm:1.consistency}). Building upon Theorem~\ref{thm:2.rate}, we then make use of the concavity of $\hat{g}_n$ and $g_0$ to establish uniform consistency. % on any compact subset in the interior of the domain $\bm{S}$.
    Loosely speaking, this relies on the fact that the convergence in $L_2$ for a sequence of Lipschitz (and concave) functions implies the uniform convergence in the interior of the domain. See	Lemma~\ref{lem:3} and Lemma~\ref{lem:4} below for more detail. Note that we only look at $\hat{g}_n$ on the a compact subset interior of its domain, in order to make sure that $\hat{g}_n$ is Lipschitz there. That is also why we do not have consistency on the boundary from the current proof strategy.
    % we look at its are differ at some points by $\epsilon$, then there exists a convex region with positive measure over which these two functions are differ by at least $\epsilon/2$, and thus $\|\hat{g}_n - g_0\|_{n,m}  > c$ in probability for some positive constant $c>0$, which leads to a contradiction. To elaborate further, here we (a) use the approach of \citet{lim2012consistency} to show Lipschitz-ness of $\hat{g}_n$ (again, over a pre-given subset of the interior of the domain), and (b) translate the randomness of $\hat{g}_n$ to the randomness of the set $\{\bm{x}: \; |\hat{g}_n(\bm{x}) - g_0(\bm{x})| > \epsilon/2 \}$ (in order to bound $\|\hat{g}_n - g_0\|_{n,m}$ from below). The latter is then handled by \citet{Rao1962} (which can be viewed as a version of Glivenko--Cantelli theorem).
    
    \item (Theorem~\ref{thm:3.uconsistency}). If we let the number of evaluation points, $m$, grow at a certain rate slower than $n$, we can extend the uniform consistency result to the entire support of $\bm{X}$. The assumption on the rate of growth of $m$ makes sure that the first partial derivative of SCKLS, $\frac{\partial \hat{g}_n}{\partial\bm{x}}(\bm{x})$, is bounded for some positive constant, so the SCKLS is Lipschitz over the entire domain. 
    
    \item (Theorem~\ref{thm:4.missconsistency}). This can be viewed as a generalization of Theorem~\ref{thm:1.consistency}. The main ingredient of its proof is to establish $\|\hat{g}_n - g_0^*\|_{n,m}=o_p(1)$. Then the uniform consistency follows from the concavity of $\hat{g}_n$ and $g_0^*$ via Lemma~\ref{lem:4}.

\end{enumerate}

	\subsection{Alternative definition of SCKLS}
	\label{App:AppendixC1}
	Recall that given observations $\{\bm{X}_j,y_j\}_{j=1}^n$ and evaluation points $\{\bm{x}_i\}_{i=1}^m$, the (unconstrained) local linear estimator at $\bm{x}_i$ is $(\tilde{a}_i,\tilde{\bm{b}}_i)$ for $i = 1,\ldots,m$, where $(\tilde{a}_1,\tilde{\bm{b}}_1,\ldots,\tilde{a}_m,\tilde{\bm{b}}_m)$ is the (unique) minimizer of 
	\[ 
	\sum_{i=1}^{m}\sum_{j=1}^{n}(y_j-a_i-(\bm{X}_j-\bm{x}_i)'\bm{b}_i)^2K\left(\frac{\bm{X}_j-\bm{x}_i}{\bm{h}}\right)\\.
	\]	
	For simplicity, we assume that the bandwidth is equal for all input dimensions, i.e. $\bm{h} = (h,\ldots,h)'$. Since the objective function is quadratic, for any $(a_1,\bm{b}_1,\ldots,a_m,\bm{b}_m)$, its value equals
	\[
	nh^d \sum_{i=1}^{m} \big(\tilde{a}_i-a_i, (\tilde{\bm{b}}_i-\bm{b}_i)'{h}\big) \boldsymbol{\Sigma}_i \begin{pmatrix}  \tilde{a}_i-a_i  \\ (\tilde{\bm{b}}_i-\bm{b}_i){h} \end{pmatrix}  + \mathrm{Const}
	\]	
	where 
	\[
	\boldsymbol{\Sigma}_i = \frac{1}{n h^d}\sum_{j=1}^n U\Big(\frac{\bm{X}_j-\bm{x}_i}{{h}}\Big) \Big\{U\Big(\frac{\bm{X}_j-\bm{x}_i}{{h}}\Big)\Big\}' K\left(\frac{\bm{X}_j-\bm{x}_i}{{h}}\right)  
	\]
	with $U(\bm{x})$ being the vector $(1, \bm{x}')'$ and 
	\[
	\mathrm{Const} = \sum_{i=1}^{m}\sum_{j=1}^{n}(y_j-\tilde{a}_i-(\bm{X}_j-\bm{x}_i)'\tilde{\bm{b}}_i)^2K\left(\frac{\bm{X}_j-\bm{x}_i}{{h}}\right).
	\]	
	Therefore, SCKLS can be simply viewed as a minimizer of 
	\[
	\sum_{i=1}^{m} \big(\tilde{a}_i-a_i, (\tilde{\bm{b}}_i-\bm{b}_i)'{h}\big) \boldsymbol{\Sigma}_i \begin{pmatrix}  \tilde{a}_i-a_i  \\ (\tilde{\bm{b}}_i-\bm{b}_i){h} \end{pmatrix}
	\]
	subject to the shape constraints imposed on $(a_1,\bm{b}_1,\ldots,a_m,\bm{b}_m)$. More generally, fixing $\{\bm{X}_1,\ldots,\bm{X}_n\}$, $\{\bm{x}_1,\ldots,\bm{x}_m\}$ and $h$, and define a new squared distance measure between two functions $g_1,g_2$ as
	\[
	\|g_1 - g_2\|_{n,m}^2 = \frac{1}{m}\sum_{i=1}^{m} \Big(g_1(\bm{x}_i)-g_2(\bm{x}_i), \big(\frac{\partial g_1}{\partial \bm{x}}(\bm{x}_i)-\frac{\partial g_2}{\partial \bm{x}}(\bm{x}_i)\big)'{h}\Big) \boldsymbol{\Sigma}_i \begin{pmatrix}  g_1(\bm{x}_i)-g_2(\bm{x}_i) \\\big(\frac{\partial g_1}{\partial \bm{x}}(\bm{x}_i)-\frac{\partial g_2}{\partial \bm{x}}(\bm{x}_i)\big)'{h} \end{pmatrix},
	\]
	then SCKLS belongs to\footnote{To be more precise technically,  if $g_1-g_2$ is not differentiable, then $\|g_1 - g_2\|_{n,m}$ needs to be taken as the infimum among all possible sub-gradients in the previous definition. Nevertheless, since we only consider the behavior of the functions at finitely many points, without loss of generality, here we can restrict ourselves to differentiable functions.}
	\[
	\argmin_{g \in G_2} 	\|g - \tilde{g}_n\|_{n,m}
	\]
	where $G_2$ is the set that contains all the concave and increasing functions from $\bm{S}$ to $\mathbb{R}$.
	
	Below, we list some useful results on the behaviors of $\boldsymbol{\Sigma}_i$ and $(\tilde{a}_i,\tilde{\bm{b}}_i)$. These results follow from \citet{fan2016}. %Lemma 5 and 11, and Proposition 7 of \citet{fan2016} respectively.
	\begin{lemma}[Lemma 5 of \citet{fan2016}, Page 508]
		\label{lem:1}
		Suppose that Assumption 1(i)-1(vi) hold, then
		with probability one, there exists $C>1$ such that the eigenvalues of $\boldsymbol{\Sigma}_i$ are in $[1/C,C]$ for all $i=1,\ldots,m$ for sufficiently large $n$. 
	\end{lemma}

	\begin{lemma}[Proposition 7 of \citet{fan2016}, Page 509]
		\label{lem:2}
		Suppose that Assumption 1(i)-1(vi) hold, then as $n \rightarrow \infty$,
		\[
		\sup_{i=1,\ldots,m} \Big(|\tilde{a}_i - g_0(\bm{x}_i)|^2, \Big\|h\Big\{\tilde{\bm{b}}_i - \frac{\partial g_0}{\partial \bm{x}}(\bm{x}_i)\Big\} \Big\|^2 \Big) = O_p(n^{-4/(4+d)}\log n).
		\]
	\end{lemma}

	%Its proof is extremely similar to that of Lemma~\ref{lem:3}, so is omitted for brevity.

	\subsection{Proof of Theorems in Section~\ref{sec:4.property}}

	\subsubsection{Proof of Theorem 1}

	\begin{proof}
		
		With a sufficiently large $n$, the uniqueness of the estimates of $\hat{g}_n(\bm{x}_i)$ and $\frac{\partial \hat{g}_n}{\partial \bm{x}} (\bm{x}_i)$ for $i = 1,\ldots,m$ is established because our objective function corresponds to is a quadratic programming problem with a positive definite (strictly convex) objective function with a feasible solution. See \cite{bertsekas1995nonlinear}.
		
		Based on our characterization of SCKLS in Appendix~\ref{App:AppendixC1}, we note that the objective function at the SCKLS estimate is smaller than or equal to that at the truth, and thus
		\[
    \|\hat{g}_n - \tilde{g}_n\|_{n,m}^2 \le \|g_0 - \tilde{g}_n\|_{n,m}^2.
    \]
    Moreover,  by the triangular inequality, we have that 
    \[
    \|\hat{g}_n - g_0\|_{n,m} \le  \|\hat{g}_n - \tilde{g}_n\|_{n,m}  +  \|\tilde{g}_n - g_0\|_{n,m}  \le 2\| \tilde{g}_n - g_0\|_{n,m}.
    \]
    As such,  
    \begin{align}
    \label{eq:proofthm1.1}
      \|\hat{g}_n - g_0\|_{n,m}^2 \le 4\| \tilde{g}_n - g_0\|_{n,m}^2. 
    \end{align}
    Recall that the (unconstrained) local linear estimator at $\bm{x}_i$ is $(\tilde{a}_i,\tilde{\bm{b}}_i)$ for $i = 1,\ldots,m$. It follows from Lemma~\ref{lem:2} that 
    \[
    \| \tilde{g}_n - g_0\|_{n,m}^2 =  \frac{1}{m}	\sum_{i=1}^{m} \big(\tilde{a}_i-g_0(\bm{x}_i), \big(\tilde{\bm{b}}_i-\frac{\partial g_0}{\partial \bm{x}} (\bm{x}_i) \big)'{h}\Big) \boldsymbol{\Sigma}_i \begin{pmatrix}  \tilde{a}_i- g_0(\bm{x}_i)  \\ \big(\tilde{\bm{b}}_i-\frac{\partial g_0}{\partial \bm{x}}(\bm{x}_i) \big){h} \end{pmatrix} = O_p(n^{-4/(4+d)}\log n)
    \]
    In addition, from Lemma~\ref{lem:1}, we have that
    \begin{align}
    \notag\|\hat{g}_n - g_0\|_{n,m}^2 &= \frac{1}{m}	\sum_{i=1}^{m} \big(\hat{g}_n(\bm{x}_i)-g_0(\bm{x}_i), \big(\frac{\partial \hat{g}_n}{\partial \bm{x}} (\bm{x}_i) -\frac{\partial g_0}{\partial \bm{x}} (\bm{x}_i) \big)'{h}\Big) \boldsymbol{\Sigma}_i \begin{pmatrix}  \hat{g}_n(\bm{x}_i) - g_0(\bm{x}_i)  \\ \big(\frac{\partial \hat{g}_n}{\partial \bm{x}} (\bm{x}_i)  -\frac{\partial g_0}{\partial \bm{x}}(\bm{x}_i) \big){h} \end{pmatrix}\\
    &\ge \frac{1}{C m} \sum_{i=1}^m (\hat{g}_n(\bm{x}_i)-g_0(\bm{x}_i))^2, \label{eq:proofthm1.2}
    \end{align}
    where $C$ is the constant mentioned in the statement of Lemma~\ref{lem:1}.
    
    Plugging the above two equations into (\ref{eq:proofthm1.1}) yields
    \begin{align*}
    %\label{eq:proofthm1.2}
    \frac{1}{m} \sum_{i=1}^m (\hat{g}_n(\bm{x}_i)-g_0(\bm{x}_i))^2 \le O_p(n^{-4/(4+d)}\log n) = o_p(1).
    \end{align*}
    \end{proof}

    \subsubsection{Proof of Theorem 2}
    For the sake of clarity, we have divided the proof of Theorem 2 into several parts. 
    
    \paragraph{Some useful lemmas}
	Here we list two useful lemmas on the convergence of convex functions. %The proof of Lemma~\ref{lem:3} is provided here for the sake of completeness.
	
	\begin{lemma}
		\label{lem:3}
		Suppose that $f_0,f_1,f_2,\ldots: \bm{C}'\rightarrow \mathbb{R}$ are Lipschitz and convex functions, where $\bm{C}' \subset \mathbb{R}^d$ is a compact and convex set. In addition, assume that these functions all have the same bound and Lipschitz constant. Then 
		\[
		\lim_{n \rightarrow \infty }\int_{\bm{C}'}\{f_n(\bm{x})-f_0(\bm{x})\}^2d\bm{x} = 0
		\]
		implies that 
		\[
		\lim_{n \rightarrow \infty }\sup_{\bm{x} \in \bm{C}}|f_n(\bm{x})-f_0(\bm{x})| = 0
		\]
		for any compact $\bm{C}$ in the interior of $ \bm{C}'$.
	\end{lemma}
	
	\begin{proof}
	Suppose that the common Lipschitz constant is $M >0$. Moreover, suppose that
	\[
	\sup_{\bm{x} \in \bm{C}'} \inf_{\bm{y} \in \bm{C}} \|\bm{x}-\bm{y}\| =: \delta.
	\]
	Essentially, that means that for any $\bm{x} \in \bm{C}'$, the ball of radius $\delta$ centered at $\bm{x}$ (denoted as $B_\delta(\bm{x})$) intersects with $\bm{C}$.
	
	Next, suppose that $\sup_{\bm{x} \in \bm{C}}|f_n(\bm{x})-f_0(\bm{x})|\ge \epsilon$ for some $\epsilon > 0$. Let 
	\[
	\bm{x}^* \in \mathrm{argmax}_{\bm{x} \in  \bm{C}} |f_n(\bm{x})-f_0(\bm{x})|.
	\]
	Then for any $\bm{x}$ that lies inside the ball of radius $\min\{\delta, \epsilon/(4M)\}$ centered at $\bm{x}^*$, we have that
	\begin{align*}
	|f_n(\bm{x})-f_0(\bm{x})| &= |f_n(\bm{x})-f_n(\bm{x}^*)+f_n(\bm{x}^*) -f_0(\bm{x}^*)+ f_0(\bm{x}^*)- f_0(\bm{x})|\\
	& \ge |f_n(\bm{x}^*) -f_0(\bm{x}^*)|- |f_n(\bm{x})-f_n(\bm{x}^*)|-|f_0(\bm{x}^*)- f_0(\bm{x})|\\
	& \ge \epsilon - \frac{\epsilon}{4M}M - \frac{\epsilon}{4M}M = \frac{\epsilon}{2},
	\end{align*}
	where we made use of the Lipschitz constant for $f_n$ and $f_0$ in the second last line above.
	Consequently,
	\[
	\int_{\bm{C}'}\{f_n(\bm{x})-f_0(\bm{x})\}^2d\bm{x} \ge \Big(\frac{\epsilon}{2}\Big)^2 \mathrm{Vol}(B_{\min\{\delta, \epsilon/(4M)\}}(\bm{x}^*)) = \mathrm{Const.} \times \epsilon^{d+2} 
	\]
	for any $0 < \epsilon < 4M\delta$. 
	
	But since $\epsilon >0$ is arbitrary, $\limsup_{n\rightarrow \infty}\sup_{\bm{x} \in \bm{C}}|f_n(\bm{x})-f_0(\bm{x})|\ge \epsilon$ for any sufficiently small $\epsilon$
	would imply
	\[
		\limsup_{n\rightarrow \infty}\int_{\bm{C}'}\{f_n(\bm{x})-f_0(\bm{x})\}^2d\bm{x} \ge  \mathrm{Const.} \times \epsilon^{d+2},
	\]
	violating
	\[
      \lim_{n\rightarrow \infty}\int_{\bm{C}'}\{f_n(\bm{x})-f_0(\bm{x})\}^2d\bm{x} = 0.
	\]
	Our proof is thus completed by contradiction.

	\end{proof}

	The following Lemma~\ref{lem:4} can be viewed as a small extension of Lemma~\ref{lem:3}. This is the version that we shall use in the proof of Theorem~\ref{thm:1.consistency}.
	
	\begin{lemma}
		\label{lem:4}
		Suppose that $f_0,f_1,f_2,\ldots: \bm{C}'\rightarrow \mathbb{R}$ are Lipschitz and convex functions (that could be random), where $\bm{C}' \subset \mathbb{R}^d$ is a compact and convex set. In addition, assume that these functions all have the same bound and Lipschitz constant. Furthermore, $q:\bm{C}'\rightarrow \mathbb{R}$ with $\inf_{\bm{x} \in \bm{C}'} q(\bm{x}) > 0$. Then, for any fixed compact set $\bm{C}$ in the interior of $ \bm{C}'$,
		\[
        \int_{\bm{C}'}\{f_n(\bm{x})-f_0(\bm{x})\}^2 q(\bm{x})d\bm{x}  \stackrel{p}{\rightarrow} 0
		\]
		implies that 
		\[
		\sup_{\bm{x} \in \bm{C}}|f_n(\bm{x})-f_0(\bm{x})|\stackrel{p}{\rightarrow} 0
		\]
	as $n \rightarrow \infty$.
	\end{lemma}
	
	\begin{proof}
	Following the arguments in the proof of Lemma~\ref{lem:3}, we see that $\sup_{\bm{x} \in \bm{C}}|f_n(\bm{x})-f_0(\bm{x})|\ge \epsilon$ would entail 
	\[
	\int_{\bm{C}'}\{f_n(\bm{x})-f_0(\bm{x})\}^2q(\bm{x}) d\bm{x} \ge \Big(\frac{\epsilon}{2}\Big)^2 \mathrm{Vol}(B_{\min\{\delta, \epsilon/(4M)\}}(\bm{x}^*)) \inf_{\bm{x} \in \bm{C}} q(\bm{x}) = \mathrm{Const.} \times \epsilon^{d+2} 
	\]
	for any sufficiently small $\epsilon$. Consequently, $
        \int_{\bm{C}'}\{f_n(\bm{x})-f_0(\bm{x})\}^2 q(\bm{x})d\bm{x}  \stackrel{p}{\rightarrow} 0$
		implies that 
		$\sup_{\bm{x} \in \bm{C}}|f_n(\bm{x})-f_0(\bm{x})|\stackrel{p}{\rightarrow} 0$.
	\end{proof}		
	
	\paragraph{Lipschitz continuity of SCKLS}
   	For the reasons that will become clearer later, it is useful to investigate the Lipschitz continuity of SCKLS before we present our proof of Theorem 2. Our finding is summarized in the following lemma. Its proof is similar to that of Proposition 4 of \citet[Page 201--202]{lim2012consistency}, or that of Theorem 1 of \citet[online supplementary material, Page 2--6]{ChenSamworth2016}.  We provide a concise version of the proof for the sake of completeness. To better illustrate its main idea and intuition, below we focus on the scenario of $d=1$.

	\begin{lemma}
	\label{lem:5}
    Under the assumptions of the first part of Theorem~2 (in the case where $m$ increases with $n$), for any convex and compact set $\bm{C}\subset \mathrm{int}(\bm{S})$ (where $\mathrm{int}(\cdot)$ denotes the interior of a set), there exists some constants $B > 0$ and $M > 0$ such that  $\hat{g}_n$ is $B$-bounded and  $M$-Lipschitz over $\bm{C}$ with probability one as $n\rightarrow \infty$.
    \end{lemma}
    
    \begin{proof}
    As explained before, here we focus on the scenario of $d=1$. Without loss of generality, we can take $\bm{S}=[0,1]$ and $\bm{C}=[\delta,1-\delta]$ for some $\delta \in (0,1/2)$.
    
    Let $B_0 = \sup_{[0,1]}|g_0(x)|$. First, we show that the event
    \[
    \sup_{x \in [\delta,1-\delta]}|\hat{g}_n(x)| \le 2B_0+1 =:B
    \]
    happens with probability one as $n \rightarrow \infty$. 
    
    Since $\hat{g}_n$ is increasing, $\sup_{x \in [\delta,1-\delta]}|\hat{g}_n(x)| = \max\Big(|\hat{g}_n(\delta)|,|\hat{g}_n(1-\delta)|\Big)$. In addition, due to the monotonicity of $\hat{g}_n$, suppose that $\hat{g}_n(\delta) \le 0$, then $|\hat{g}_n(x)| \ge |\hat{g}_n(\delta)|$ for $x \in [0, \delta]$; otherwise, if $\hat{g}_n(\delta) > 0$,  $|\hat{g}_n(x)| \ge |\hat{g}_n(\delta)|$ for $x \in [\delta, 2\delta]$ (actually, this statement is true for $x \in [\delta, 1]$; but for our purpose, it suffices to only consider $x \in [\delta, 2\delta]$). As such, $|\hat{g}_n(\delta)| > 2B_0+1$ would imply that
    \begin{align*}
    \frac{1}{m}  \sum_{i=1}^m (\hat{g}_n(\bm{x}_i)-g_0(\bm{x}_i))^2 &\ge \frac{ \mathbf{1}_{\{\hat{g}_n(\delta) \le 0\}}}{m} \sum_{i=1}^m (\hat{g}_n(x_i)-g_0(x_i))^2 \mathbf{1}_{\{x_i \in [0,\delta]\}} \\
    &\quad +  \frac{\mathbf{1}_{\{\hat{g}_n(\delta) > 0\}}}{m}  \sum_{i=1}^m (\hat{g}_n(x_i)-g_0(x_i))^2 \mathbf{1}_{\{x_i \in [\delta,2\delta]\}}\\
    & \ge (2B_0+1-B_0)^2 \bigg( \frac{\mathbf{1}_{\{\hat{g}_n(\delta) \le 0\}}}{m}  \sum_{i=1}^m\mathbf{1}_{\{x_i \in [0,\delta]\}} +  \frac{\mathbf{1}_{\{\hat{g}_n(\delta) > 0\}}}{m}  \sum_{i=1}^m \mathbf{1}_{\{x_i \in [\delta,2\delta]\}}\bigg)\\
    &\ge (B_0+1)^2 \min\bigg(\frac{1}{m}  \sum_{i=1}^m \mathbf{1}_{\{x_i \in [0,\delta]\}}, \frac{1}{m}  \sum_{i=1}^m \mathbf{1}_{\{x_i \in [\delta,2\delta]\}}\bigg)\\
    &\stackrel{n \rightarrow \infty}{\ge}  B_0^2  \delta \min_{[0,1]} q(x) > 0.
    \end{align*}
    where $q(\cdot)$ is the density function with respect to what the empirical distribution of $\{\bm{x}_1,\ldots,\bm{x}_m\}$ converges to (see Assumption~\ref{ass:2}(\ref{ass:2.1})). Here the  last line also follows from Assumption~\ref{ass:2}(\ref{ass:2.1}). Note that Theorem~1 says that $\frac{1}{m}  \sum_{i=1}^m (\hat{g}_n(\bm{x}_i)-g_0(\bm{x}_i))^2 = o_p(1)$, which would result in a contradiction. Therefore, $|\hat{g}_n(\delta)| \le 2B_0+1$. 
    
    Furthermore, we can reapply the above argument to show that  $|\hat{g}_n(1-\delta)| \le 2B_0+1$. Consequently,  
    \[
    \sup_{x \in [\delta,1-\delta]}|\hat{g}_n(x)| \le 2B_0+1 =B
    \]
    happens with probability one as $n \rightarrow \infty$. 
    
    Second, note that the above proof works for any $\delta \in (0,1/2)$. Therefore, we also have that 
    \[
    \sup_{x \in [\delta/2,1-\delta/2]}|\hat{g}_n(x)| \le 2B_0+1
    \]
    with probability one  as $n \rightarrow \infty$. 
    
    Finally, since $\hat{g}_n$ is concave, we note that the Lipschitz constant over $[\delta,1-\delta]$ is bounded above by
    \[
    \max \Big(\frac{|\hat{g}_n(\delta/2)-\hat{g}_n(\delta)|}{\delta/2},  \frac{|\hat{g}_n(1-\delta/2)-\hat{g}_n(1-\delta)|}{\delta/2}\Big) \le 4(2B_0+1)/\delta =:M.
    \]
    In other words, intuitively speaking, in terms of the Lipschitz constant, the most extreme case for concave functions always occurs on the boundary. For general cases (i.e. $d > 1$), see for instance, \citet[Page 165, Problem 7]{van1996weak}.
    \end{proof}
	
    \paragraph{Putting things together to prove Theorem 2}
	\begin{proof} $\; $
	
	\noindent \textbf{First claim: when $m$ increases with $n$.}
	
    Let $C'$ be a compact and convex set such that $\bm{C} \subset \mathrm{int}(\bm{C}')$ and $\bm{C}' \subset \mathrm{int}(\bm{S})$, where $\mathrm{int}(\cdot)$ denotes the interior of a set. 
    
    By Lemma~\ref{lem:5}, we have that $\hat{g}_n$ is $B$-bounded and $M$-Lipschitz over $\bm{C}'$ with probability one as $n\rightarrow \infty$. Therefore, 
    %$\hat{g}_n$ belongs to a class of %functions that is Glivenko--Cantelli %(see Theorem 2.7.11 and Theorem 2.4.1 of %\citet{van1996weak}). Consequently,  
    $\{\hat{g}_n(\bm{x})-g_0(\bm{x})\}^2 \mathbf{1}_{\{\bm{x} \in  \bm{C}'\}}$  belongs to the class of functions that is bounded and equicontinuous over $\bm{C}'$.
    %Glivenko--Cantelli as well (see Theorem 3 of \citet{vanderVaart2000}). 
    By Theorem~3.1 of \cite[Page 662]{Rao1962} (which can also be viewed as a generalization of the Uniform Law of Large Numbers; see also Chapter 2.4 of \citet{van1996weak}), we have that 
    \[
    \left|\frac{1}{m} \sum_{i=1}^m (\hat{g}_n(\bm{x}_i)-g_0(\bm{x}_i))^2 \mathbf{1}_{\{\bm{x}_i \in  \bm{C}'\}} - \int_{\bm{C}'} \{\hat{g}_n(\bm{x})-g_0(\bm{x})\}^2 q(\bm{x})d\bm{x}\right| \stackrel{p}{\rightarrow}0.
    \]
    In addition, it follows from Theorem~\ref{thm:2.rate} that
    \[
     o_p(1) = \frac{1}{m} \sum_{i=1}^m (\hat{g}_n(\bm{x}_i)-g_0(\bm{x}_i))^2 \ge \frac{1}{m} \sum_{i=1}^m (\hat{g}_n(\bm{x}_i)-g_0(\bm{x}_i))^2 \mathbf{1}_{\{\bm{x}_i \in  \bm{C}'\}}.
    \]
    %where $q(\cdot)$ is the density function with respect to what the empirical distribution of $\{\bm{x}_1,\ldots,\bm{x}_m\}$ converges to (see Assumption~\ref{ass:2}(\ref{ass:2.1})). 
    Combining the above two equations together yields
    \[
    \int_{\bm{C}'} \{\hat{g}_n(\bm{x})-g_0(\bm{x})\}^2 q(\bm{x})d\bm{x} = o_p(1).
    \]
    It then follows immediately from Lemma~\ref{lem:4} that as $ n \rightarrow \infty$,
    \[
    \sup_{\bm{x} \in \bm{C}}|\hat{g}_n(\bm{x})-g_0(\bm{x})|\stackrel{p}{\rightarrow} 0.
    \]
    
	\vspace{3mm}
	\noindent\textbf{Second claim: when $m$ is fixed.}
    
    In views of Lemma~\ref{lem:1} and  Theorem \ref{thm:2.rate},
	    \[
		 \frac{1}{C}\sum_{i=1}^{m}\left[|\hat{g}_n(\bm{x}_i)-g_0(\bm{x}_i)|^2 + \Big\|\Big(\frac{\partial \hat{g}_n}{\partial \bm{x}}(\bm{x}_i)-\frac{\partial g_0}{\partial \bm{x}}(\bm{x}_i)\Big)h\Big\|^2 \right] \le  \|\hat{g}_n - g_0\|_{n,m}^2 = O_p(n^{-4/(4+d)}\log n)
	    \]
	    where the first inequality is from Lemma~\ref{lem:1}, and the last equality is from Theorem \ref{thm:2.rate}.

	 Since $m$ is fixed and $h = O(n^{-1/(4+d)})$, it follows from that $| \hat{g}_n(\bm{x}_i)-g_0(\bm{x}_i)| = O_p(n^{-2/(4+d)}\log n) \stackrel{p}{\rightarrow} 0$ and $\|\frac{\partial \hat{g}_n}{\partial \bm{x}}(\bm{x}_i) - \frac{\partial {g}_0}{\partial \bm{x}}(\bm{x}_i) \| = O_p(n^{-1/(4+d)}\log n) \stackrel{p}{\rightarrow} 0$ for every $i=1,\ldots,m$. 
	\end{proof}

	\subsubsection{Proof of Theorem 3}
	\begin{proof}
	    Using Equation (\ref{eq:proofthm1.2}) but focusing on the difference between the derivatives instead, we have that
	    \begin{align*}
	     \frac{h^2}{C m} \sum_{i=1}^m \Big\|\Big(\frac{\partial \hat{g}_n}{\partial \bm{x}}(\bm{x}_i)-\frac{\partial g_0}{\partial \bm{x}}(\bm{x}_i)\Big)\Big\|^2 &\le \|\hat{g}_n - g_0\|_{n,m}^2 = O_p(n^{-4/(4+d)}\log n)
	    \end{align*}	  
	    as $n \rightarrow \infty$.
	    It then follows from $h=O(n^{-1/(4+d)})$ and Assumption~\ref{ass:3} that 
	    \begin{align*}
\sum_{i=1}^m \Big\|\frac{\partial \hat{g}_n}{\partial\bm{x}}(\bm{x}_i)-\frac{\partial g_0}{\partial\bm{x}}(\bm{x}_i)\Big\|^2 = O_p( h^{-2} m n^{-4/(4+d)} \log n) = o_p(1).
	  	    \end{align*}
  	    This implies that $\max_{i=1,\ldots,m} \Big\|\frac{\partial \hat{g}_n}{\partial\bm{x}}(\bm{x}_i)\Big\|_{\infty} \le \sup_{\bm{x} \in \bm{S}} \Big\|\frac{\partial {g}_0}{\partial\bm{x}}(\bm{x})\Big\|_{\infty} + o_p(1)$.
		Now since 
		\[
		\hat{g}_n(\bm{x})=\min_{i\in\{1,\ldots,m\}}\Big\{\hat{g}_n(\bm{x}_i)+(\bm{x}-\bm{x_i})' \frac{\partial \hat{g}_n}{\partial\bm{x}}(\bm{x}_i)	\Big\},\]
		we have that with probability one,
		\[
    	\sup_{\bm{x} \in \bm{S}} \Big\|\frac{\partial \hat{g}_n}{\partial\bm{x}}(\bm{x})\Big\|_{\infty} \le M
		\]
		for some $M > 0$, as $n \rightarrow \infty$.
		
		For any $\epsilon > 0$, we can always find a compact set $\bm{C}_\epsilon \subset \bm{S}$ such that $\sup_{\bm{x} \in \bm{S}} \inf_{\bm{y} \in \bm{C}_\epsilon} \|\bm{x} - \bm{y} \| < \frac{\epsilon}{2(M+M_{g_0})}$, where $M_{g_0}$ is the Lipschitz constant of $g_0$.  In view of Theorem 2, $\sup_{\bm{x} \in \bm{C}_\epsilon}|\hat{g}_n(\bm{x})-g_0(\bm{x})|\rightarrow 0$ in probability. Therefore,
		\begin{align*}
		   \sup_{\bm{x} \in \bm{S}} |\hat{g}_n(\bm{x}) - g_0(\bm{x})| \le \sup_{\bm{x} \in \bm{C}_\epsilon} |\hat{g}_n(\bm{x}) - g_0(\bm{x})| + (M+M_{g_0}) \Big \{\sup_{\bm{x} \in \bm{S}} \inf_{\bm{y} \in \bm{C}_\epsilon} \|\bm{x} - \bm{y} \| \Big \}  \le \epsilon
		\end{align*}
		as $n \rightarrow \infty$. Since $\epsilon$ is picked arbitrarily, we have shown the consistency of $\hat{g}_n$ over $\bm{S}$.		
		
	\end{proof}

	\subsection{Proof of Theorems in Section~\ref{sec:5.test}}

    \subsubsection{Proof of Theorem 4}

	\begin{proof}
	Using the definition of SCKLS in Appendix~\ref{App:AppendixC1} and the notation in the proofs of Theorem 1 and Theorem 2, we have that 
    \begin{align*}
    &\sum_{i=1}^{m} \Big(\tilde{a}_i-g_0^*(\bm{x}_i), \big(\tilde{\bm{b}}_i-\frac{\partial g_0^*}{\partial \bm{x}}(\bm{x}_i) \big)'{h}\Big) \boldsymbol{\Sigma}_i \begin{pmatrix}  \tilde{a}_i- g_0^*(\bm{x}_i)  \\ \big(\tilde{\bm{b}}_i-\frac{\partial g_0^*}{\partial \bm{x}}(\bm{x}_i) \big){h} \end{pmatrix} \\
    & \ge 
    \sum_{i=1}^{m} \big(\tilde{a}_i-\hat{a}_i, (\tilde{\bm{b}}_i-\hat{\bm{b}}_i)'h\big) \boldsymbol{\Sigma}_i \begin{pmatrix}  \tilde{a}_i-\hat{a}_i \\ 
	    	\big(\tilde{\bm{b}}_i-\hat{\bm{b}}_i\big){h} \end{pmatrix}\\
	& = \sum_{i=1}^{m} \Big(\tilde{a}_i-g_0^*(\bm{x}_i), \big(\tilde{\bm{b}}_i-\frac{\partial g_0^*}{\partial \bm{x}}(\bm{x}_i) \big)'{h}\Big) \boldsymbol{\Sigma}_i \begin{pmatrix}  \tilde{a}_i- g_0^*(\bm{x}_i)  \\ \big(\tilde{\bm{b}}_i-\frac{\partial g_0^*}{\partial \bm{x}}(\bm{x}_i) \big){h} \end{pmatrix} \\
	& \qquad + 2 \sum_{i=1}^m \big(\tilde{a}_i-g_0^*(\bm{x}_i), (\tilde{\bm{b}}_i-\frac{\partial g_0^*}{\partial \bm{x}}(\bm{x}_i))'h\big) \boldsymbol{\Sigma}_i \begin{pmatrix}  g_0^*(\bm{x}_i) - \hat{a}_i \\ 
	    	(\frac{\partial g_0^*}{\partial \bm{x}}(\bm{x}_i)-\hat{\bm{b}}_i\big){h} \end{pmatrix}\\
	& \qquad + \sum_{i=1}^{m} \big(g_0^*(\bm{x}_i) - \hat{a}_i, (\frac{\partial g_0^*}{\partial \bm{x}}(\bm{x}_i)-\hat{\bm{b}}_i)'h\big) \boldsymbol{\Sigma}_i \begin{pmatrix} g_0^*(\bm{x}_i) - \hat{a}_i \\ 
	    	\big(\frac{\partial g_0^*}{\partial \bm{x}}(\bm{x}_i)-\hat{\bm{b}}_i\big){h} \end{pmatrix}
	\end{align*}
	where we recall that $\hat{a}_i$ and $\hat{\bm{b}}_i$ are respectively the estimated value and its gradient from SCKLS at evaluation point $\bm{x}_i$, i.e., $\hat{a}_i = \hat{g}_n(\bm{x}_i)$ and $\hat{\bm{b}}_i = \frac{\partial\hat{g}_n}{\partial\bm{x}}(\bm{x}_i)$.
	
	Therefore, in view of Lemma~\ref{lem:2}, with probability one, for sufficiently large $n$,
	\begin{align}
	 \label{eq:proofthm4.1} &\frac{2}{m}\sum_{i=1}^m \big(\tilde{a}_i-g_0^*(\bm{x}_i), (\tilde{\bm{b}}_i-\frac{\partial g_0^*}{\partial \bm{x}}(\bm{x}_i))'h\big) \boldsymbol{\Sigma}_i \begin{pmatrix}  \hat{a}_i - g_0^*(\bm{x}_i) \\ (\hat{\bm{b}}_i- \frac{\partial g_0^*}{\partial \bm{x}}(\bm{x}_i)){h} \end{pmatrix} \\
	  \label{eq:proofthm4.2} &\ge \frac{1}{m} \sum_{i=1}^{m} \big(g_0^*(\bm{x}_i) - \hat{a}_i, (\frac{\partial g_0^*}{\partial \bm{x}}(\bm{x}_i)-\hat{\bm{b}}_i)'h\big) \boldsymbol{\Sigma}_i \begin{pmatrix} g_0^*(\bm{x}_i) - \hat{a}_i \\ \big(\frac{\partial g_0^*}{\partial \bm{x}}(\bm{x}_i)-\hat{\bm{b}}_i\big){h} \end{pmatrix} 
	    & \ge \frac{1}{mC} \sum_{i=1}^{m} (g_0^*(\bm{x_i}) - \hat{a}_i)^2
	\end{align}
	
	Next, we show that the quantity in (\ref{eq:proofthm4.1}) converges to zero in probability as $n \rightarrow \infty$. The proof can be divided into six steps:
	\begin{enumerate}[1.]
	    \item The contribution to (\ref{eq:proofthm4.1}) from evaluation points lying outside a carefully pre-chosen compact subset $\bm{S}'$ of the interior of $\bm{S}$ (denoted as $\mathrm{int}(\bm{S})$) can be made arbitrarily small. This follows from the Cauchy--Schwarz inequality that
	    \begin{align}
	    & \notag \frac{1}{m}\sum_{i=1}^m \big(\tilde{a}_i-g_0^*(\bm{x}_i), (\tilde{\bm{b}}_i-\frac{\partial g_0^*}{\partial \bm{x}}(\bm{x}_i))'h\big) \boldsymbol{\Sigma}_i \begin{pmatrix}  \hat{a}_i - g_0^*(\bm{x}_i) \\ (\hat{\bm{b}}_i- \frac{\partial g_0^*}{\partial \bm{x}}(\bm{x}_i)){h} \end{pmatrix} \mathbf{1}_{\{\bm{x} \notin \bm{S}'\}} \\
	    &  \label{eq:proofthm4.3} \le \sqrt{ \frac{1}{m}\sum_{i=1}^m \big(\tilde{a}_i-g_0^*(\bm{x}_i), (\tilde{\bm{b}}_i-\frac{\partial g_0^*}{\partial \bm{x}}(\bm{x}_i))'h\big) \boldsymbol{\Sigma}_i \begin{pmatrix}  \tilde{a}_i - g_0^*(\bm{x}_i) \\ (\tilde{\bm{b}}_i- \frac{\partial g_0^*}{\partial \bm{x}}(\bm{x}_i)){h} \end{pmatrix} \mathbf{1}_{\{\bm{x} \notin \bm{S}'\}}}\\
	    &  \label{eq:proofthm4.4} \qquad \times \sqrt{ \frac{1}{m}\sum_{i=1}^m \big(\hat{a}_i-g_0^*(\bm{x}_i), (\hat{\bm{b}}_i-\frac{\partial g_0^*}{\partial \bm{x}}(\bm{x}_i))'h\big) \boldsymbol{\Sigma}_i \begin{pmatrix}  \hat{a}_i - g_0^*(\bm{x}_i) \\ (\hat{\bm{b}}_i- \frac{\partial g_0^*}{\partial \bm{x}}(\bm{x}_i)){h} \end{pmatrix} }.
	   \end{align}
	   Because of Lemma~\ref{lem:1} and Assumption~\ref{ass:2}(\ref{ass:2.1}), the quantity in (\ref{eq:proofthm4.3}) can be made arbitrarily small by choosing $\bm{S}'$ sufficiently close to $\bm{S}$. In addition, applying the Cauchy--Schwarz inequality to (\ref{eq:proofthm4.1}) and comparing it to  (\ref{eq:proofthm4.2}) yields
	   \begin{align*}
	   	 &2\sqrt{\frac{1}{m}\sum_{i=1}^m \big(\tilde{a}_i-g_0^*(\bm{x}_i), (\tilde{\bm{b}}_i-\frac{\partial g_0^*}{\partial \bm{x}}(\bm{x}_i))'h\big) \boldsymbol{\Sigma}_i \begin{pmatrix}  \tilde{a}_i - g_0^*(\bm{x}_i) \\ (\tilde{\bm{b}}_i- \frac{\partial g_0^*}{\partial \bm{x}}(\bm{x}_i)){h} \end{pmatrix} }\\
	   	 & \qquad \times \sqrt{\frac{1}{m}\sum_{i=1}^m \big(\hat{a}_i-g_0^*(\bm{x}_i), (\hat{\bm{b}}_i-\frac{\partial g_0^*}{\partial \bm{x}}(\bm{x}_i))'h\big) \boldsymbol{\Sigma}_i \begin{pmatrix}  \hat{a}_i - g_0^*(\bm{x}_i) \\ (\hat{\bm{b}}_i- \frac{\partial g_0^*}{\partial \bm{x}}(\bm{x}_i)){h} \end{pmatrix} } \\  
	   	  &\ge \frac{1}{m} \sum_{i=1}^{m} \big(g_0^*(\bm{x}_i) - \hat{a}_i, (\frac{\partial g_0^*}{\partial \bm{x}}(\bm{x}_i)-\hat{\bm{b}}_i)'h\big) \boldsymbol{\Sigma}_i \begin{pmatrix} g_0^*(\bm{x}_i) - \hat{a}_i \\ \big(\frac{\partial g_0^*}{\partial \bm{x}}(\bm{x}_i)-\hat{\bm{b}}_i\big){h} \end{pmatrix},
	   \end{align*}
	   so (\ref{eq:proofthm4.4}) is no greater than
	   \begin{align*}
	   &2\sqrt{\frac{1}{m}\sum_{i=1}^m \big(\tilde{a}_i-g_0^*(\bm{x}_i), (\tilde{\bm{b}}_i-\frac{\partial g_0^*}{\partial \bm{x}}(\bm{x}_i))'h\big) \boldsymbol{\Sigma}_i \begin{pmatrix}  \tilde{a}_i - g_0^*(\bm{x}_i) \\ (\tilde{\bm{b}}_i- \frac{\partial g_0^*}{\partial \bm{x}}(\bm{x}_i)){h} \end{pmatrix} } \\
	   &\rightarrow 2 \Big\{\int_{\bm{S}} (g_0(\bm{x}) - g_0^*(\bm{x}))^2 Q(d\bm{x})\Big\}^{1/2} \le 2 \Big\{\int_{\bm{S}} g_0^2(\bm{x}) Q(d\bm{x})\Big\}^{1/2}.
	   \end{align*}
	   Consequently, the claim in this step is proved.
	   \item We now investigate the contribution to (\ref{eq:proofthm4.1}) from evaluation points lying inside $\bm{S}'$. Using Lemma~\ref{lem:5}, we have that $\hat{g}_n$ is bounded (i.e. from both below and above) and $M$-Lipschitz over $\bm{S'}$ in probability. 
	   
	   Combining this with Lemma~\ref{lem:1} implies that 
	   \begin{align*}
	   &\Bigg|\frac{1}{m}\sum_{i=1}^m \big(\tilde{a}_i-g_0^*(\bm{x}_i), (\tilde{\bm{b}}_i-\frac{\partial g_0^*}{\partial \bm{x}}(\bm{x}_i))'h\big) \boldsymbol{\Sigma}_i \begin{pmatrix}  \hat{a}_i - g_0^*(\bm{x}_i) \\ \big(\hat{\bm{b}}_i- \frac{\partial g_0^*}{\partial \bm{x}}(\bm{x}_i)\big){h} \end{pmatrix} \mathbf{1}_{\{\bm{x} \in \bm{S}'\}} \\
	   &\quad -  \frac{1}{m}\sum_{i=1}^m \bigg((g_0-g_0^*)(\bm{x}_i), (\frac{\partial (g_0 - g_0^*)}{\partial \bm{x}}(\bm{x}_i))'h\bigg) \boldsymbol{\Sigma}_i \begin{pmatrix}  \hat{a}_i - g_0^*(\bm{x}_i) \\ \big(\hat{\bm{b}}_i- \frac{\partial g_0^*}{\partial \bm{x}}(\bm{x}_i)\big){h} \end{pmatrix} \mathbf{1}_{\{\bm{x} \in \bm{S}'\}} \Bigg| \rightarrow 0
	   \end{align*}
	   in probability. As such, we can instead work on 
	   \begin{align}
	   \label{eq:proofthm4.5}
	         \frac{1}{m}\sum_{i=1}^m \bigg((g_0-g_0^*)(\bm{x}_i), (\frac{\partial (g_0 - g_0^*)}{\partial \bm{x}}(\bm{x}_i))'h\bigg) \boldsymbol{\Sigma}_i \begin{pmatrix}  \hat{a}_i - g_0^*(\bm{x}_i) \\ \big(\hat{\bm{b}}_i- \frac{\partial g_0^*}{\partial \bm{x}}(\bm{x}_i)\big){h} \end{pmatrix} \mathbf{1}_{\{\bm{x} \in \bm{S}'\}}
	   \end{align}
	   \item Next, we bound (and eliminate) the influence from the parts involving partial derivatives of $g_0$, $g_0^*$ and $\hat{g}_n$ in (\ref{eq:proofthm4.5}). Since $\hat{g}_n$ is bounded and $M$-Lipschitz over $\bm{S'}$ in probability, together with Lemma~\ref{lem:2}, we could bound (\ref{eq:proofthm4.5}) from above by 
	   \begin{align*}
	   \frac{1}{m}\sum_{i=1}^m (g_0(\bm{x}_i)-g_0^*(\bm{x}_i))( \hat{g}_n(\bm{x}_i)) - g_0^*(\bm{x}_i)) \mathbf{1}_{\{\bm{x} \in \bm{S}'\}} + O(h) + O(h^2), 
	   \end{align*}
	   which is arbitrarily close to $\frac{1}{m}\sum_{i=1}^m (g_0(\bm{x}_i)-g_0^*(\bm{x}_i))(\hat{g}_n(\bm{x}_i) - g_0^*(\bm{x}_i)) \mathbf{1}_{\{\bm{x} \in \bm{S}'\}} $ as $n \rightarrow \infty$ (i.e. $h \rightarrow$ 0). Here we also used the fact that $\sup_{i =1,\ldots,m} |\boldsymbol{\Sigma}_i^{(11)}-1| \rightarrow 0$, where $\boldsymbol{\Sigma}_i^{(11)}$ is the first diagonal entry of the matrix $\boldsymbol{\Sigma}_i$. 
	   \item  Now we re-expand $\hat{g}_n$ from $\bm{S'}$ to $\bm{S}$ as 
	   \[
	   \hat{g}^{\bm{S}'}_n(\bm{x}) = \min_{i\in\{1,\ldots,m | \bm{x}_i \in \bm{S}'\}, }\Big\{\hat{g}_n(\bm{x}_i)+(\bm{x}-\bm{x_i})' \frac{\partial \hat{g}_n}{\partial\bm{x}}(\bm{x}_i)	\Big\}.
	   \]
	   Three useful facts about $\hat{g}^{\bm{S}'}_n$ are listed below:
	   \begin{itemize}
	   \item $\hat{g}_n^{\bm{S}'} \ge \hat{g}_n$, with $\hat{g}_n^{\bm{S}'}(\bm{x}_i) = \hat{g}_n(\bm{x}_i)$ for any $\bm{x}_i \in \bm{S}'$.
	   \item  there exists some $B>0$ such that $\sup_{\bm{x} \in \bm{S}}\hat{g}_n^{\bm{S}'}(\bm{x}) \le B$ in probability. Importantly, given that there is a common compact and convex set $\bm{C}$ such that $\bm{C} \subset \bm{S}'$ for all the $\bm{S}'$ to be considered, the constant $B$ does not depend on the choice of $\bm{S}'$. To see this, we note that $\hat{g}_n^{\bm{C}} = \hat{g}_n$ over $\bm{C}$, which is also $B'$-bounded and $M'$-Lipschitz over $\bm{C}$ in probability via Lemma~\ref{lem:5}. Then it follows that
	   \[
	   \hat{g}_n^{\bm{S}'} \le \hat{g}_n^{\bm{C}} \le B'+M' \sup_{\bm{y}_1,\bm{y_2} \in \bm{S}} \|\bm{y}_1 - \bm{y}_2\| =: B
	   \]
	   in probability as $n \rightarrow \infty$.
	   
       \item The function $\{(g_0-g_0^*)(\hat{g}_n - g_0^*)\}(\cdot)$ is bounded and Lipschitz over $\bm{S}'$ in probability (where the constants do not depend on $n$). So is $\{(g_0-g_0^*)(\hat{g}^{\bm{S}'}_n - g_0^*)\}(\cdot)$ over $\bm{S}$. This also means that $\{(g_0-g_0^*)(\hat{g}^{\bm{S}'}_n - g_0^*)\}(\cdot)$ is equicontinuous over $\bm{S}$.
	   \end{itemize}

	   \item 
	    Returning to the quantity we mentioned at the end of Step 3, we note that 
	   \begin{align*}
	   &\frac{1}{m}\sum_{i=1}^m (g_0(\bm{x}_i)-g_0^*(\bm{x}_i))( \hat{a}_i - g_0^*(\bm{x}_i) \mathbf{1}_{\{\bm{x}_i \in \bm{S}'\}} \\ 
	   &= 	   \frac{1}{m}\sum_{i=1}^m \Big((g_0-g_0^*)(\hat{g}_n^{\bm{S}'} - g_0^*)\Big)(\bm{x}_i) - 	   \frac{1}{m}\sum_{i=1}^m \Big((g_0-g_0^*)(\hat{g}_n^{\bm{S}'} - g_0^*)\Big)(\bm{x}_i) \mathbf{1}_{\{\bm{x}_i \notin \bm{S}'\}} \\
	   &=    \frac{1}{m}\sum_{i=1}^m \Big((g_0-g_0^*)(\hat{g}_n^{\bm{S}'} - g_0^*)\Big)(\bm{x}_i) - \frac{1}{m}\sum_{i=1}^m \Big((g_0-g_0^*)(\hat{g}_n - g_0^*)\Big)(\bm{x}_i) \mathbf{1}_{\{\bm{x}_i \notin \bm{S}'\}}\\
	   & \qquad -  \frac{1}{m}\sum_{i=1}^m \Big((g_0-g_0^*)(\hat{g}_n - \hat{g}_n^{\bm{S}'})\Big)(\bm{x}_i) \mathbf{1}_{\{\bm{x}_i \notin \bm{S}'\}}\\
	   &= (\mathrm{I}) + (\mathrm{II}) + (\mathrm{III}).
	   \end{align*}
	  We deal with each of these items separately.
	  \begin{itemize}
	      \item By the third fact listed in the above Step~4 and Theorem 3.1 of \citet{Rao1962}, (I) in the limit (i.e. as $n \rightarrow \infty$) is at most
	      \[
	      \sup_{g \in G_2}\int_{\bm{S}}\{(g_0(\bm{x})-g_0^*(\bm{x})\}\{g(\bm{x})-g_0^*(\bm{x})\}q(\bm{x})d\bm{x} \le 0.
	      \]
	      Note that $g_0^*$ minimizes 
	      \[
	      \mathcal{G}(g):= \int_{\bm{S}}(g_0(\bm{x})-g(\bm{x}))^2q(\bm{x})d\bm{x}
	      \]
	      over all $g \in G_2$. The previous inequality thus follows by studying the functional derivative for the function $\mathcal{G}(\cdot)$ at $g_0^*$ in the direction of $g-g_0^*$ (N.B. $g_0^* + \epsilon (g-g_0^*) \in G_2$ for  $\epsilon \rightarrow 0$) for all $g \in G_2$.
	      \item Both $|(\mathrm{II})|$ and $|(\mathrm{III})|$ in the limit can be arbitrarily small for $\bm{S}'$ sufficiently close to $\bm{S}$. This follows from Cauchy--Schwarz inequality and an argument similar to that in Step~1.
	      
	  \end{itemize}
	   \item We now put things together by noting that in light of Steps 1 to 5, for any $\epsilon$, we can find some $\bm{S}'$ such that the quantity in (\ref{eq:proofthm4.1}) is no bigger than $\epsilon$ in probability as $n \rightarrow \infty$. Since the quantity in (\ref{eq:proofthm4.1}) is also non-negative, our claim that (\ref{eq:proofthm4.1}) converges to zero in probability is verified.
	\end{enumerate}
	
	Finally, uniform consistency over any $\bm{C}$ can be shown using exactly the same approach we demonstrated in the final stage of proving the first part of Theorem~\ref{thm:1.consistency} via Lemma~\ref{lem:4}.

	\end{proof}
	
	\subsubsection{Proof of Theorem 5}

	\begin{proof}
	
	Our proof can be divided into three parts.
	
	\noindent \textbf{1. The case of $g_0 = 0$.}
	
	Using the definition of SCKLS in Appendix~\ref{App:AppendixC1}, it is easy to verify that $T_n = \|\hat{g}_n - \tilde{g}_n \|_{n,m}$. For reasons that will become clear later, we denote $\hat{g}_{n}^\circ$ and $ \tilde{g}_{n}^\circ$ the SCKLS and LL estimators based on the same covariates, evaluation points and bandwidth used in calculating $T_n$, but with the response vector $(\epsilon_1,\ldots,\epsilon_n)'$ (instead of $\bm{y}_n$) and set $T_n^\circ = \|\tilde{g}^\circ_{n}   - \hat{g}^\circ_{n}\|_{n,m}$. Obviously, when $g_0 = 0$ (which is the case here), $\hat{g}_{n}^\circ = \hat{g}_{n}, \tilde{g}_{n}^\circ = \tilde{g}_{n}$ and $T_n^\circ = T_n$. 
	
	Now, for $k = 1,\ldots,B$,  $T_{nk} = \|\hat{g}_{nk} - \tilde{g}_{nk} \|_{n,m}$, where $\hat{g}_{nk}$ and $ \tilde{g}_{nk}$ are respectively the SCKLS and LL estimators based on the same covariates, evaluation points and bandwidth used in calculating $T_n$, but with the response vector $(u_{1k} \tilde{\epsilon}_1,\ldots,u_{nk}\tilde{\epsilon}_n)'$. Further, we define a slightly modified bootstrap version of the test statistic as $T_{nk}^\circ = \|\hat{g}_{nk}^\circ - \tilde{g}_{nk}^\circ \|_{n,m}$, where $\hat{g}_{nk}^\circ$ and $ \tilde{g}_{nk}^\circ$ are the SCKLS and LL estimators based on the same covariates, evaluation points and bandwidth used in calculating $T_n$, but with the response $(u_{1k} \epsilon_1,\ldots,u_{nk}\epsilon_n)'$. Let $\bm{e} = (|\epsilon_1|,\ldots,|\epsilon_n|)'$ and denote $p_n^\circ = \frac{1}{B}\sum_{i=1}^B \mathbf{1}_{\{T_n^\circ \le T_{nk}^\circ\}}$. Then, it follows from the symmetry of the error distribution that conditioning on the values of the absolute errors (i.e. $(|\epsilon_1|,\ldots,|\epsilon_n|)' = \bm{e}$), the quantities
	\[
	T_n^\circ, T_{n1}^\circ,\ldots,T_{nB}^\circ
	\]
	are exchangeable. Consequently, as $B \rightarrow \infty$,
	\[
	P(p_n^\circ \le \alpha) = E\Big\{P\Big(p_n^\circ \le \alpha \Big|(|\epsilon_1|,\ldots,|\epsilon_n|)' = \bm{e}\Big) \Big\} \le \frac{\lfloor B \alpha \rfloor+1}{1+B} \rightarrow \alpha.
	\]
	Back to the elements in the quantity $p_n$, our aim is to show that $\mathbf{1}_{\{T_n \le T_{nk}^\circ\}} \le \mathbf{1}_{\{T_n \le T_{nk} + \Delta_n\}}$ for large $n$. Note that 
	\[
	T_{nk} - T_{nk}^\circ  = \|\tilde{g}_{nk} - \hat{g}_{nk}\|_{n,m} - \|\tilde{g}^\circ_{nk} - \hat{g}^\circ_{nk}\|_{n,m}  \le \|\tilde{g}_{nk} - \hat{g}^\circ_{nk}\|_{n,m} - \|\tilde{g}^\circ_{nk} - \hat{g}^\circ_{nk}\|_{n,m}\le \|\tilde{g}_{nk} - \tilde{g}^\circ_{nk}\|_{n,m}
	\]
	Because we estimated the error vector in Step 1 using LL (without any shape restrictions), it follows from Proposition~7 of \citet{fan2016} that $\sup_{j}|\tilde{\epsilon_j}-\epsilon_j| \le O_p(n^{-2/(4+d)} \log^{1/2} n)$. By the linearity of the LL estimator (w.r.t. the response vector), we have that $\sup_k\|\tilde{g}_{nk} - \tilde{g}^\circ_{nk}\|^2_{n,m} = O_p(n^{-4/(4+d)} \log n)$. Consequently, with arbitrarily high probability,
	\begin{align*}
    \inf_{k=1,\ldots,B} (T_{nk}+ \Delta_n - T_{nk}^\circ) > 0
	\end{align*}
	for sufficiently large $n$. This yields $\mathbf{1}_{\{T_n^\circ \le T_{nk}^\circ\}} \le \mathbf{1}_{\{T_n \le T_{nk} + \Delta_n\}}$ and thus $p_n \ge p_n^\circ$. As a result, $P(p_n \le \alpha) \le P(p_n^\circ \le \alpha) \le \alpha$, as required.
	\newline
	
	\noindent \textbf{2. The general case of  $g_0 \in G_2$.}
	
	To relate $T_n$ to what we investigated before (i.e. $g_0 = 0$), we recall the definitions of $\hat{g}_{n}^\circ$ and $ \tilde{g}_{n}^\circ$  from the previous case, and define an additional quantity $\tilde{g}_{n}^{\dagger}$ to be the LL estimator in exactly the same setting, but is obtained using the response vector $(g_0(\bm{X}_1),\ldots,g_0(\bm{X}_n))'$. By the linearity of the LL, $\tilde{g}_{n} = \tilde{g}^\circ_{n} + \tilde{g}^{\dagger}_{n}$.  Since $g_0$ is continuously twice-differentiable, we have that 
	\[
	T_n = \|\tilde{g}_{n} - \hat{g}_{n}\|_{n,m} \le \|\tilde{g}^\circ_{n} + \tilde{g}^{\dagger}_{n}  - \hat{g}^\circ_{n} - g_0\|_{n,m} \le  \|\tilde{g}^\circ_{n}   - \hat{g}^\circ_{n} \|_{n,m} + \|\tilde{g}^{\dagger}_{n} - g_0\|_{n,m} = T_n^\circ + O_p(h^2).
	\]
	As a result, with arbitrarily high probability, for every $k = 1, \ldots, B$,
	\[
	T_{nk}+\Delta_n - T_n = T_{nk}^\circ - T_n^\circ + (T_{nk} -  T_{nk}^\circ) - (T_n - T_n^\circ) + \Delta_n \ge  T_{nk}^\circ - T_n^\circ
	\]
	for sufficiently large $n$. This also leads to $\mathbf{1}_{\{T_n^\circ \le T_{nk}^\circ\}} \le \mathbf{1}_{\{T_n \le T_{nk} + \Delta_n\}}$. We could then directly apply the argument from the previous case to conclude that  $P(p_n \le \alpha) \le \alpha$.
	\newline
	
	\noindent \textbf{3. The case of  $g_0 \notin G_2$}
	
	Here $g_0$ is assumed to be fixed and continuously twice-differentiable. 
	
	First, two situations are considered.
	\begin{itemize}
	    \item Under Assumption~\ref{ass:2}(\ref{ass:2.1}), we recall that
	\[
	g_0^* := \argmin_{g \in G_2} \int_{\bm{S}}\{g(\bm{x})-g_0(\bm{x})\}^2 Q(d\bm{x}).
	\]
    Since $g_0 \notin G_2$, there must exists some compact set $\bm{S}' \subset \mathrm{int}(\bm{S})$ such that $Q(\bm{S}') > 0$ and 
    \[
    \inf_{\bm{x} \in \bm{S}'} |g_0^*(\bm{x}) - g_0(\bm{x})| > \delta.
    \]
    Note that
    \[
    T_n^2 =  \|\hat{g}_n - \tilde{g}_n\|^2_{n,m} \ge  \frac{1}{m}\sum_{i=1}^{m} \Big(\hat{g}_n(\bm{x}_i)-\tilde{g}_n(\bm{x}_i), \big(\frac{\partial(g_1-g_2)}{\partial \bm{x}}(\bm{x}_i)\big)'{h}\Big) \boldsymbol{\Sigma}_i \begin{pmatrix}  \hat{g}_n(\bm{x}_i)-\tilde{g}_n(\bm{x}_i) \\\frac{\partial(\hat{g}_n-\tilde{g}_n)}{\partial \bm{x}}(\bm{x}_i)h \end{pmatrix}\mathbf{1}_{\{\bm{x}_i \in \bm{S}'\}}.
    \]
    Here we have that $\tilde{g}_n \rightarrow g_0$ by \citet{fan2016} and $\hat{g}_n \rightarrow g_0^*$ over $\bm{S}'$ by our Theorem~\ref{thm:4.missconsistency}. Since $\tilde{g}_n - \hat{g}_n$ is Lipschitz over $\bm{S}'$, it is easy to verify (see also Step 3 of the proof of Theorem~\ref{thm:4.missconsistency}) that the righthand side of the above display equation is bounded below by $\delta^2 Q(\bm{S}')$ in the limit as $n \rightarrow \infty$ (also $h \rightarrow 0$). Consequently, $T_n \ge c'$ in probability for some $c' >0$.
    \item Now under Assumption~\ref{ass:2}(\ref{ass:2.2}), since $g_0 \notin G_2$ and the evaluation points are reasonably well spread across $\bm{S}$ (i.e. Assumption~\ref{ass:2}(\ref{ass:2.2})), for sufficiently large and fixed $m$, we can always find some evaluation points where the imposed shape constraint is violated. This means that
    \[
	\inf_{g \in G_2} \|g - g_0\|_{n,m} \ge c
	\]
	in probability for some $c > 0$. So we still have that
	\[
	T_n =  \|\hat{g}_n - \tilde{g}_n\|_{n,m} \ge  \|\hat{g}_n - g_0\|_{n,m} -  \|\tilde{g}_n- g_0\|_{n,m} \ge 	\inf_{g \in G_2} \|g - g_0\|_{n,m} - o_p(1) \ge c'
	\]
	in probability for some $c' > 0$.	
    
	\end{itemize}

	Second, it follows from the proof for the case of $g_0=0$ that 
	\[
	T_{nk}  = T_{nk}^\circ + T_{nk} - T_{nk}^\circ \le \|\tilde{g}_{nk}^\circ\|_{n,m} + \|\tilde{g}_{nk} - \tilde{g}_{nk}^\circ\|_{n,m} =o_p(1).
	\]
	
    Finally, write $W_{nk} =  \mathbf{1}_{\{T_{nk} + \Delta_n > c'/2\}}$. We note that $W_{n1}, \ldots, W_{nB}$ are exchangeable. Thus, for any $\alpha \in (0,1)$, as $n \rightarrow \infty$,
	\begin{align*}
	    P(\mbox{Do not reject } H_0) &= P\Bigg(\frac{1}{B}\sum_{k=1}^B \mathbf{1}_{\{T_n \le T_{nk}+ \Delta_n\}} \ge \alpha\Bigg) \\
	    &\le P(T_n \le  c'/2)+ P\Bigg(T_n > c'/2,\, \frac{1}{B}\sum_{k=1}^B \mathbf{1}_{\{T_n \le T_{nk}+ \Delta_n\}} \ge \alpha\Bigg) \\
	    &\le P(T_n \le c'/2) + P\Bigg(\frac{1}{B}\sum_{k=1}^B W_{nk} \ge \alpha\Bigg) \\
	    &\le P(T_n \le c'/2) + \frac{E(W_{n1})}{\alpha} \rightarrow 0,
	\end{align*}
	where we used Markov's inequality in the final line above. So the Type II error at the alternative indeed converges to 0.

	\end{proof}

	\subsection{Proof of Propositions in Appendix~\ref{App:AppendixA3}}
	\subsubsection{Proof of Proposition \ref{thm:4.SCKLS-CNLS}}	
	
	%\begin{prop}
	%		\label{App:thm:4.SCKLS-CNLS}
	%		Suppose that Assumption 1(v) and Assumption 4 holds. Then, for any $n$, when the vector of bandwidth approaches zero, i.e. $\|\bm{h}\|\rightarrow \mathbf{0}$ (where $\bm{h}=(h_1,\ldots,h_d)'$) the SCKLS estimator $\hat{g}_n$ converges to the CNLS estimator $\tilde{g}_n$ pointwise at $\bm{X}_1,\ldots,\bm{X}_n$.
	%\end{prop}

		\begin{proof}
		In view of Assumption \ref{ass:1}~(\ref{ass:1.5}), for any sufficiently small $\bm{h}$, we have
		\[
		K\left(\frac{\bm{X}_j-\bm{x}_i}{\bm{h}}\right) = \begin{cases}
		0 & \mbox{if } \bm{x}_i\neq\bm{X}_j, \\
		K(\bm{0}) & \mbox{if } \bm{x}_i=\bm{X}_j,
		\end{cases} \mbox{ for }  \forall i,j.	
		\]
		Then, the objective function of (\ref{eq:6.SCKLS}) is equal to $\sum_{j=1}^{n}(y_j-a_j)^2 K(\bm{0})$, and thus 
		\begin{equation*}
		%\label{eq:9.objcnv}
		\argmin_{a_1,\bm{b}_1,\ldots,a_n,\bm{b}_n}\sum_{j=1}^{n}(y_j-a_j)^2 K(\bm{0}) =\argmin_{a_1,\ldots,a_n}\sum_{j=1}^{n}(y_j-a_j)^2
		\end{equation*}
		Writing $a_j=\alpha_j+\bm{\beta}_j'\bm{X}_j$ and $\bm{b}_j=\bm{\beta}_j$ for $j=1,\ldots,n$ by definition. Then, quadratic programming problem (\ref{eq:6.SCKLS}) can be rewritten as follows:
		\begin{equation*}
		\begin{aligned}
		%\label{eq:10.CNLS}
		& \min_{\alpha,\bm{\beta}}
		& & \sum_{j=1}^{n}(y_j-(\alpha_j+\bm{\beta}_j'\bm{X}_j))^2\\
		& \mbox{subject to}
		& & \alpha_j+\bm{\beta}_j'\bm{X}_j\leq\alpha_l+\bm{\beta}_l'\bm{X}_j, \; & j,l=1,\ldots,n\\
		&
		& & \bm{\beta}_j\geq 0, \; &j=1,\ldots,n
		\end{aligned}
		\end{equation*}
		which is equivalent to the formulation of the CNLS estimator (\ref{eq:1.CNLS}).
	\end{proof}	

\subsubsection{Proof of Proposition \ref{thm:5.SCKLS-OLS}}	
%	\begin{theorem}
%		\label{App:thm:5.SCKLS-OLS}
%		Given Assumption~1(v). For any given $n$, when the bandwidth approaches to infinity, the SCKLS estimator converges to the least squares estimator of the linear regression model subject to monotonicity constraints.
%	\end{theorem}	
	\begin{proof}
		When $\min_{k=1,\ldots,d} h_k\rightarrow\infty$, we have
		\begin{equation}
		\label{eq:12.K}
		K\left(\frac{\bm{X}_j-\bm{x}_i}{\bm{h}}\right) = K(\bm{0}) \; \quad \mbox{for } \forall i,j.
		\end{equation}
		By substituting (\ref{eq:12.K}) into the objective function of (\ref{eq:6.SCKLS}) converges to
		\[
		\sum_{i=1}^{m}\sum_{j=1}^{n}(y_j-a_i-(\bm{X}_j-\bm{x}_i)'\bm{b}_i)^2 K(\bm{0}).
		\]
		
		Next, we derive the minimum of the objective function in the limit. Let's consider 
		\begin{equation}
		\label{eq:13.objcnv}
    	\argmin_{a_1,\bm{b}_1,\ldots,a_m,\bm{b}_m}\sum_{i=1}^{m}\sum_{j=1}^{n}(y_j-a_i-(\bm{X}_j-\bm{x}_i)'\bm{b}_i)^2
		\end{equation}
		subject to constraints. Rewrite $a_i+(\bm{X}_j-\bm{x}_i)'\bm{b_i}=\alpha_i+\bm{\beta}_i'\bm{X}_j$ for  $i=1,\ldots,m$ and $j=1,\ldots,n$. Then the objective function of (\ref{eq:6.SCKLS}) can be rewritten as follows with (\ref{eq:13.objcnv}).	
		\begin{equation*}
		\begin{aligned}
		%\label{eq:14.OLS1}
		& \min_{\alpha_1,\bm{\beta}_1,\ldots,\alpha_m,\bm{\beta}_m}
		& & \sum_{i=1}^{m}\sum_{j=1}^{n}(y_j-(\alpha_i+\bm{\beta}_i'\bm{X}_j))^2\\
		& \mbox{subject to}
		& & \alpha_i+\bm{\beta}_i'\bm{x}_i\leq\alpha_l+\bm{\beta}_l'\bm{x}_i
		& i,l=1,\ldots,m\\
		& 
		& & \bm{\beta}_i\geq 0
		& i=1,\ldots,m
		\end{aligned}
		\end{equation*}	
		Here, since we do not impose any weight on the objective function, it is easy to see that $\alpha_1 = \cdots = \alpha_m$ and $\bm{\beta}_1 = \cdots = \bm{\beta}_m$. Then the Afriat constraints become redundant, resulting in
		\begin{equation*}
		\begin{aligned}
		%\label{eq:15.OLS2}
		& \min_{\alpha,\bm{\beta}}
		& & \sum_{j=1}^{n}(y_j-(\alpha+\bm{\beta}'\bm{X_j}))^2\\
		& \mbox{subject to}
		& & \bm{\beta}\geq 0.
		\end{aligned}
		\end{equation*}		
		
		%It remains to show that as $\bm{h} \rightarrow \infty, the minimum of objective function of (\ref{eq:6.SCKLS}) for each $\bm{h}$ also converges to the minimum of the objective function in the limit (with $\bm{h} = \infty$). This is indeed the case because the objective function is quadratic (both for any $\bm{h}$ and in the limit when $\bm{h} \rightarrow \infty$).  

	\end{proof}

	\subsubsection{Proof of Proposition \ref{thm:1.SCKLS-CWBY}}	
		\begin{proof}
		In view of Assumption \ref{ass:1}~(\ref{ass:1.5}), for any sufficiently small $\bm{h}$, we have
		\[
		K\left(\frac{\bm{X}_j-\bm{x}_i}{\bm{h}}\right) = \begin{cases}
		0 & \mbox{if } \bm{x}_i\neq\bm{X}_j, \\
		K(\bm{0}) & \mbox{if } \bm{x}_i=\bm{X}_j,
		\end{cases} \mbox{ for }  \forall i,j.	
		\]
		Then, the objective function of the SCKLS estimator (3) is equal to $\sum_{j=1}^{n}(y_j-a_j)^2 K(\bm{0})$, and thus 
		\begin{equation*}
		%\label{eq:9.objcnv}
		\argmin_{a_1,\bm{b}_1,\ldots,a_n,\bm{b}_n}\sum_{j=1}^{n}(y_j-a_j)^2 K(\bm{0}) =\argmin_{a_1,\ldots,a_n}\sum_{j=1}^{n}(y_j-a_j)^2
		\end{equation*}
		
		Also consider Assumption A1 (i) from \cite{du2013nonparametric}, we can say something similar for CWB in y-space. For any sufficiently small $\bm{h}$, we have
		
% 		\[ A_{j}(\bm{X_l}) = diag\left(\frac{1}{n},\frac{1}{n},\ldots,\frac{1}{n}\right) \] where $j,l=1,\ldots,n$.

        \[ A_{j}(\bm{x_i}) = \begin{cases}
		0 & \mbox{if } \bm{x}_i\neq\bm{X}_j, \\
		n & \mbox{if } \bm{x}_i=\bm{X}_j,
		\end{cases} \mbox{ for }  \forall i,j.
		\]
		
		and thus
		
		\begin{equation}
	    \begin{aligned}
	    \label{eq:1.ndef_g}
	    \hat{g}(\bm{x}_i|\bm{p})=\sum_{j=1}^{n}p_jA_{j}(\bm{X}_i)y_{j}=n p_i y_i~~ \forall i=1,\ldots,n.
	    \end{aligned}
	    \end{equation}
		
		Then we can rewrite the CWB in $y$-space estimator as follows:
	\begin{equation}
		\begin{aligned}
		\label{eq:1.CWB_rev}
		&  \min_{\bm{p}}
		& & D_y(\bm{p})=\sum_{i=1}^{n}(y_i-n p_i y_{i})^2\\
		& \mbox{subject to}
		& & l(\bm{x}_i)\leq\hat{g}^{(\bm{s})}(\bm{x}_i|\bm{p})\leq u(\bm{x}_i), \; & i=1,\ldots,n.
		\end{aligned}
	\end{equation}
		
		Recognize that if $\hat{g}_n = n p_{i}y_{i}$ is true, then SCKLS and CWB in y-space are equivalent. Take $\hat{g}_n$ as the solution to SCKLS estimator and let $p_i$ be a set of decision variables, we see $\hat{g}_n = n p_{i}y_{i}$ is simply a system of $n$ equations and $n$ unknowns.   
	\end{proof}

	\section{Testing for affinity using SCKLS}
	\label{sec:5.2test}
	
	\subsection{The procedure}
	To further illustrate the usefulness of SCKLS for testing other shapes, we study the problem of testing 
	\[
	H_0:\;\; g_0: \bm{S} \rightarrow \mathbb{R} \mbox{ is affine } \quad \mbox{against} \quad H_1:\;\; g_0: \bm{S} \rightarrow \mathbb{R} \mbox{ is not affine}. 
	\]
	%Our test statistic and testing procedure is again based on the SCKLS estimators presented above. 
	The main idea of our test is motivated by \cite{sen2016testing}. %We show that under some regularity assumptions, our proposed test is consistent, i.e. under $H_1$, its power goes to 1 as the sample size increases. Moreover, 
	The critical value of the test can be easily computed using Monte Carlo or bootstrap methods.
	
	To start of with, we define $\hat{g}_{n}^\mathrm{V}$, the SCKLS estimator with only a set of convexity constraints as
	\begin{equation*}
	\begin{aligned}
	%\label{eq:6.SCKLS}
	& \min_{a_i,\bm{b_i}}
	& & \sum_{i=1}^{m}\sum_{j=1}^{n}(y_j-a_i-(\bm{X}_j-\bm{x}_i)'\bm{b}_i)^2K\left(\frac{\bm{X}_j-\bm{x}_i}{\bm{h}}\right)\\
	& \mbox{subject to}
	& & a_i-a_l\leq \bm{b}_i'(\bm{x}_i-\bm{x}_l), \; & i,l=1,\ldots,m
	\end{aligned}
	\end{equation*}	
    Furthermore, $\hat{g}_{n}^\Lambda$, the SCKLS estimator using only a set of  concavity constraints is defined as
	\begin{equation*}
	\begin{aligned}
	%\label{eq:6.SCKLS}
	& \min_{a_i,\bm{b_i}}
	& & \sum_{i=1}^{m}\sum_{j=1}^{n}(y_j-a_i-(\bm{X}_j-\bm{x}_i)'\bm{b}_i)^2K\left(\frac{\bm{X}_j-\bm{x}_i}{\bm{h}}\right)\\
	& \mbox{subject to}
	& & a_i-a_l\geq \bm{b}_i'(\bm{x}_i-\bm{x}_l), \; & i,l=1,\ldots,m
	\end{aligned}
	\end{equation*}	
	
	We now describe our testing procedure as follows.  
	\begin{enumerate}
		\item	First, we run linear regression on the response against the covariates and call the least squares fit $g_{n}^L$. Next, we fit the data using SCKLS (with evaluation points at $\bm{x}_1,\ldots,\bm{x}_m$ and bandwidth $\bm{h}_n$). The resulting estimators are denoted by $\hat{g}_{n}^\mathrm{V}$  and  $\hat{g}_{n}^\Lambda$,  where $\hat{g}_{n}^\mathrm{V}$ is the SCKLS estimator using only a set of convexity constraints, while   $\hat{g}_{n}^\Lambda$ is the SCKLS estimator using only a set of  concavity constraints, all based on $\{\bm{X}_j,y_j\}_{j=1}^n$.%\footnote{For mathematical definitions of $\hat{g}_{n}^\mathrm{V}$  and  $\hat{g}_{n}^\Lambda$ see Appendix C.3}. 
		We then define the test statistics to be 
	\[
	T_n = \max\bigg[\frac{1}{m}\sum_{i=1}^m \{\hat{g}_{n}^\mathrm{V}(\bm{x}_i)-g_{n}^L(\bm{x}_i)\}^2,\frac{1}{m}\sum_{i=1}^m \{\hat{g}_{n}^\Lambda(\bm{x}_i)-g_{n}^L(\bm{x}_i)\}^2\bigg].
	\]

	\item We simulate the distributional behavior of the test statistics $B$ times under $H_0$. For $k = 1,\ldots, B$, we set the observations to be $\{\bm{X}_j,y_{jk}\}_{j=1}^n$ (i.e. no change in the values of the covariates), where $\bm{y}_{nk} = (y_{1k},\ldots,y_{nk})'$ is drawn using the wild bootstrap procedure as described in Section~\ref{sec:5.1test} (or the ordinary bootstrap procedure if we know that the errors are homogeneous). Then we run linear regression on $\bm{y}_{nk}$  against the covariates and denote the least squares fit by $g_{nk}^L$. Fitting the data using SCKLS (with the same set of evaluation points and the same bandwidth as before) leads to the resulting estimators $\hat{g}_{nk}^\mathrm{V}$  and  $\hat{g}_{nk}^\Lambda$,  where $\hat{g}_{nk}^\mathrm{V}$ is the SCKLS estimator using only the convexity constraint, while   $\hat{g}_{nk}^\Lambda$ is the SCKLS estimator using only the concavity constraint, all based on $\{\bm{X}_j,y_{jk}\}_{j=1}^n$. So
	\[
		T_{nk} = \max\bigg[\frac{1}{m}\sum_{i=1}^m \{\hat{g}_{nk}^\mathrm{V}(\bm{x}_i)-g_{nk}^L(\bm{x}_i)\}^2,\frac{1}{m}\sum_{i=1}^m \{\hat{g}_{nk}^\Lambda(\bm{x}_i)-g_{nk}^L(\bm{x}_i)\}^2\bigg].
	\]

	\item The Monte Carlo $p$-value is defined as 
\[
p_n = \frac{1}{B} \sum_{k=1}^B \mathbf{1}_{\{T_n \le T_{nk} \}}.
\]
For a test of size $\alpha \in (0,1)$, we reject $H_0$ if $p_n < \alpha$. 
	\end{enumerate}
	
	The intuition of the test is as follows. First, an affine function is both convex and concave. Therefore under $H_0$, both SCKLS estimates, $\hat{g}_{n}^\mathrm{V}$  and  $\hat{g}_{n}^\Lambda$, should be close to the linear fit $g_{n}^L$, so the value of $T_n$ should be small. Second, a function is both convex and concave only if it is affine. So given enough observations, we should be able to reject the null hypothesis under $H_1$. Third, we used the fact that $T_n$ based on $\{\bm{X}_j,y_j\}_{j=1}^n$ and $ \{\bm{X}_j,\epsilon_j\}_{j=1}^n$ are exactly the same under $H_0$ when simulating the distributional behavior of $T_n$.  %(see also Appendix \ref{App:AppendixC}). 
	
	Finally, we remark that in case we know that $g_0$ is monotonically increasing a priori, we could test $H_0':$ $g_0$ is monotonically increasing and affine using essentially the same procedure with only minor modifications described in the following: we instead run linear regression with signed constraints in both Step 1 and Step 2, replace $\hat{g}_{n}^\mathrm{V}$ by the SCKLS with both the convexity and monotonicity constraints, and replace $\hat{g}_{n}^\Lambda$ by the SCKLS with both the concavity and monotonicity constraints.

	\subsection{A simulation study}
	
	 We now examine the finite-sample performance of the affinity test using data generated from the following DGP:
    \begin{equation}
	    g_0(\bm{x}) = \frac{1}{d}\sum_{k=1}^{d}x_k^p
  		%\label{eq:DGP1}
    \end{equation}
     where $\bm{x} = (x_1,\ldots,x_d)'$. With $n$ observations, for each pair $(\bm{X}_j,y_j)$, each component of the input, $\bm{X}_{jk}$, is randomly and independently drawn from uniform distribution $unif[0,1]$, and the additive noise, $\epsilon_j$, is randomly and independently sampled from a normal distribution, $N(0,0.1)$. 
    
    We considered different sample sizes $n \in \{100,300,500\}$ and vary the number of inputs $d \in \{1,2\}$, and perform 100 simulations to compute the rejection rate for each scenario. We used the 
    %assume that we do not know the distribution of the noise in advance and use the 
    ordinary bootstrap method with $B=500$.
%     \end{experiment}
    
      In the scenarios we considered $g_0$ is affine if $p=1.0$, and is non-linear if $p \in \{0.2, 0.5,2,5\}$. Table~\ref{tab:App:test2} show the rejection rate for each scenario with one-input and two-input at $\alpha = 0.05$.  We conclude that the proposed test works well with a moderate sample size. 	 
% 	 % Table generated by Excel2LaTeX from sheet 'Test'
 	 \begin{table}[htbp]
 	 	\centering
 	 	\caption{Rejection rate of the affinity test using SCKLS at $\alpha = 0.05$}
 	 	\begin{tabular}{cc|c|c}
 	 		\toprule
 	 		\multirow{2}[0]{*}{Sample size ($n$)} & \multicolumn{1}{c|}{\multirow{2}[0]{*}{Shape Parameter ($p$)}} & \multicolumn{2}{c}{Power of the Test} \\
 	 		&       &$d=1$ &$d=2$ \\
 	 		\midrule
 	 		\multirow{5}[0]{*}{100} 
 	 		& 0.2   & 0.99    & 0.74  \\
 	 		& 0.5   & 0.97    & 0.79  \\
 	 		& 1.0   & 0.05    & 0.02  \\
 	 		& 2.0   & 1.00    & 1.00  \\
 	 		& 5.0   & 1.00    & 1.00  \\   		
 	 		\midrule
 	 		\multirow{5}[0]{*}{300} 
 	 		& 0.2   & 1.00  &  1.00   \\
 	 		& 0.5   & 1.00  & 0.99  \\
 	 		& 1.0   & 0.05  &  0.01  \\
 	 		& 2.0   & 1.00  &  1.00   \\
 	 		& 5.0   & 1.00  &  1.00   \\
 	 		\midrule
 	 		\multirow{5}[0]{*}{500}
 	 		& 0.2   & 1.00    & 1.00  \\
 	 		& 0.5   & 1.00    & 1.00   \\
 	 		& 1.0   & 0.08   & 0.01   \\
 	 		& 2.0   & 1.00    & 1.00  \\
 	 		& 5.0   & 1.00    & 1.00   \\
 	 		\bottomrule
 	 	\end{tabular}%
 	 	\label{tab:App:test2}%
 	 \end{table}%
	
	%Theoretical property of our proposed test are presented below. Its proof can be found in Appendix \ref{App:AppendixC}.
	
%	\begin{theorem}
%		\label{thm:5.SCKLS-TEST2}
%		Suppose that Assumption \ref{ass:1}(\ref{ass:1.1}),(\ref{ass:1.3})--(\ref{ass:1.6}) and Assumption \ref{ass:2}(\ref{ass:2.2}) hold. Furthermore, assume that $g_0$ is twice-differentiable, %and the error distribution $Z$ is known. 
%		then for any $\alpha \in (0,1)$, as $n\rightarrow \infty$ and $B \rightarrow \infty$, 
		%\begin{enumerate}
		%	\item if $H_0$ is true, then $P(p_n < \alpha) \rightarrow \alpha$; 
		%	\item 
%		$P((p_n < \alpha) \rightarrow 1$ under $H_1$. 
		%\end{enumerate}
		%	In other words, the test is approximately at most of size $\alpha$ and is consistent.    
%	\end{theorem}
	
	%We remark that with proper choice of the bandwidth vector (e.g. see Assumption  \ref{ass:1}(\ref{ass:1.6})), our procedure alleviates the boundary issue of CNLS. Consequently, in our test procedure, there is no need to restrict  further to the class of convex/concave functions with a fixed Lipschitz constant, a technical condition imposed by \cite{sen2016testing}. 

	\section{An algorithm for SCKLS computational performance} \label{App:AppendixD}	
	
	For a given number of evaluation points, $m$, SCKLS requires $m(m-1)$ concavity constraints. Larger values of $m$ provide a more flexible functional estimate, but also increase the number of constraints quadratically, thus, the amount of time needed to solve the quadratic program also increases quadratically. Since one can select the number of evaluation points in SCKLS,  by selecting $m$ the computational complexity can be potentially reduced relative to CNLS or estimates on denser grids, i.e. with $m(m-1)\ll n(n-1)$. 
	
	Further, \cite{dantzig1954solution,dantzig1959linear} proposed an iterative approach that reduces the size of large-scale problems by relaxing a subset of the constraints and solving the relaxed model with only a subset $V$ of constraints, checking which of the excluded constraints are violated, and iteratively adding violated constraints to the relaxed model until an optimal solution satisfies all constraints. \cite{lee2013more}, who applied the approach to CNLS, found a significant reduction in computational time. Computational performances also improves if a subset of the constraints can be identified which are likely to be needed in the model. \cite{lee2013more} find the concavity constraints corresponding to pairs of observations that are close in terms of the $\ell_2$ norm measured over input vectors and more likely to be binding than those corresponding to the distant observations. We use this insight to develop a strategy for identifying constraints to include in the initial subset $V$, when solving SCKLS as described below.
	
	Given a grid to evaluate the constraints of the SCKLS estimator, we define the initial subset of constraints $V$ as those constraints constructed by adjacent grid points as shown in Figure~\ref{fig:A2.grid}. Further, we summarize our implementation of the algorithm proposed in \cite{lee2013more} below and label it as Algorithm 1.

	\begin{figure}[ht]
		\begin{center}
			\includegraphics[width=3in]{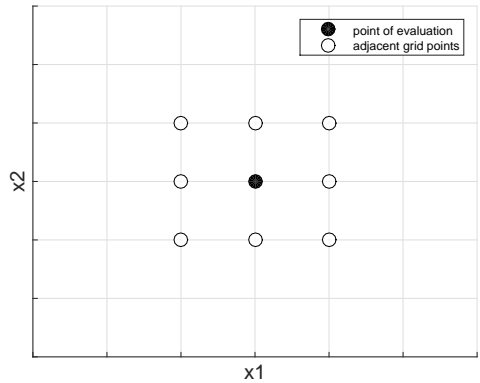}
		\end{center}
		\caption{Definition of adjacent grid in two-dimensional case. \label{fig:A2.grid}}
	\end{figure}

	\begin{algorithm}[ht]   
		\footnotesize          
		\caption{Iterative approach for SCKLS computational speedup}         
		\label{alg1}                          
		\begin{algorithmic}                  
			\STATE $t \Leftarrow 0$
			\STATE $V \Leftarrow \{(i,l):\bm{x}_i \mbox{ and } \bm{x}_l \mbox{ are adjacent, }i<l\}$
			\STATE 	Solve relaxed SCKLS with $V$ to find initial solution $\{a_i^{(0)},\bm{b}_i^{(0)}\}_{i=1}^m$
			\WHILE{$\{a_i^{(t)},\bm{b}_i^{(t)}\}_{i=1}^m$ satisfies all constraints in (\ref{eq:6.SCKLS})}
			\STATE $t \Leftarrow t+1$
			\STATE $U \Leftarrow \{(i,l): \bm{x_i} \mbox{ and } \bm{x}_l \mbox{ do not satisfy constraints in (\ref{eq:6.SCKLS})\}}$
			\STATE $V \Leftarrow V \cup U$
			\STATE Solve relaxed SCKLS with $V$ to find solution $\{a_i^{(t)},\bm{b}_i^{(t)}\}_{i=1}^m$
			\ENDWHILE
			\RETURN $\{a_i^{(t)},\bm{b}_i^{(t)}\}_{i=1}^m$
		\end{algorithmic}
	\end{algorithm}

	\clearpage
	\section{Comprehensive results of existing and additional numerical experiments}
	\label{App:CompResults}
	We show the comprehensive results of experiments in Section~\ref{sec:5.simulation} and additional experiments to show the performance of the SCKLS estimator and its extensions. For the CWB estimator, we use the convex optimization solver \texttt{SeDuMi} because \texttt{quadprog} was not able to solve CWB\footnote{For CWB,  \texttt{SeDuMi} provides a better solution than \texttt{quadprog}, while both \texttt{SeDuMi} and \texttt{quadprog} give exactly the same solution for SCKLS.}. 
	
	For CWB estimator, we use a local linear estimator to obtain the weighting matrix $A_j(\bm{x})$ in (\ref{eq:2.def_g}). The first partial derivative of $\hat{g}(\bm{x}|\bm{p})$ is obtained by approximating the derivatives through numerical differentiation $\hat{g}^{(1)}(\bm{x}|\bm{p})=\frac{\hat{g}(\bm{x}+\Delta|\bm{p})-\hat{g}(\bm{x}|\bm{p})}{\Delta}$, where $\Delta$ is a small positive constant\footnote{\cite{du2013nonparametric} proposes to use an analytical derivative for the first partial derivative of $\hat{g}(\bm{x}|\bm{p})$; however, the analytical derivative performs similarly to numerical differentiation as shown in \cite{racine2016local}. We propose two alternative methods to compute the first partial derivative, and compared them in Appendix~\ref{sec:2.3.CWB}.}.
	
	\subsection{Uniform input -- high signal-to-noise ratio (Experiment~\ref{exp:1})} \label{App:AppendixE}	
	
	We compare the following seven estimators: SCKLS with fixed bandwidth, SCKLS with variable bandwidth, CNLS, CWB in $p$-space and CWB in $y$-space, LL, and parametric Cobb--Douglas function estimated via ordinary least squares (OLS). Table~\ref{tab:A4.Exp1obs} and Table~\ref{tab:A5.exp1grd} show the RMSE of Experiment~\ref{exp:1} on observation points and evaluation points respectively.
		
	Table~\ref{tab:A6.exp1time} shows the computational time of Experiment~\ref{exp:1} for each estimator.
	
	% Table generated by Excel2LaTeX from sheet 'TableA4'
	\begin{table}[ht]
		\footnotesize
		\centering
		\caption{RMSE on observation points for Experiment~\ref{exp:1}}
		\begin{tabular}{llrrrrr}
			\toprule
			\multicolumn{2}{c}{} & \multicolumn{5}{c}{Average of RMSE on observation points} \\
			\multicolumn{2}{c}{Number of observations} & 100   & 200   & 300   & 400   & 500 \\
			\midrule
			\multirow{7}[0]{*}{2-input} & SCKLS fixed bandwidth & 0.193 & 0.171 & 0.141 & 0.132 & 0.118 \\
			& SCKLS variable bandwidth & \textbf{0.183} & 0.158 & \textbf{0.116} & \textbf{0.118} & \textbf{0.098} \\
			& CNLS  & 0.229 & 0.163 & 0.137 & 0.138 & 0.116 \\
			& CWB in $p$-space & 0.189 & 0.167 & 0.158 & 0.140 & 0.129 \\
			& CWB in $y$-space & 0.205 & \textbf{0.136} & 0.173 & 0.141 & 0.120 \\
			& LL & 0.212 & 0.166 & 0.149 & 0.152 & 0.140 \\
			\cmidrule{2-7}
			& Cobb--Douglas  & 0.078 & 0.075 & 0.048 & 0.039 & 0.043 \\
			\midrule
			\multirow{7}[0]{*}{3-input} & SCKLS fixed bandwidth & 0.230 & 0.187 & 0.183 & 0.152 & 0.165 \\
			& SCKLS variable bandwidth & 0.216 & \textbf{0.183} & \textbf{0.175} & \textbf{0.143} & \textbf{0.142} \\
			& CNLS  & 0.294 & 0.202 & 0.189 & 0.173 & 0.168 \\
			& CWB in $p$-space & 0.228 & 0.221 & 0.210 & 0.183 & 0.172 \\
			& CWB in $y$-space & \textbf{0.209} & 0.362 & 0.218 & 0.154 & 0.160 \\
			& LL & 0.250 & 0.230 & 0.235 & 0.203 & 0.181 \\
			\cmidrule{2-7}
			& Cobb--Douglas  & 0.104 & 0.089 & 0.070 & 0.047 & 0.041 \\
			\midrule
			\multirow{7}[0]{*}{4-input} & SCKLS fixed bandwidth & 0.225 & 0.248 & 0.228 & 0.203 & 0.198 \\
			& SCKLS variable bandwidth & \textbf{0.217} & \textbf{0.219} & \textbf{0.210} & \textbf{0.180} & \textbf{0.179} \\
			& CNLS  & 0.315 & 0.294 & 0.246 & 0.235 & 0.214 \\
			& CWB in $p$-space & 0.238 & 0.262 & 0.231 & 0.234 & 0.198 \\
			& CWB in $y$-space & 0.222 & 0.240 & 0.248 & 0.303 & 0.332 \\
			& LL & 0.256 & 0.297 & 0.252 & 0.240 & 0.226 \\
			\cmidrule{2-7}
			& Cobb--Douglas  & 0.120 & 0.073 & 0.091 & 0.067 & 0.063 \\
			\bottomrule
		\end{tabular}%
		\label{tab:A4.Exp1obs}%
	\end{table}%
	
	% Table generated by Excel2LaTeX from sheet 'TableA5'
	\begin{table}[ht]
		\footnotesize
		\centering
		\caption{RMSE on evaluation points for Experiment~\ref{exp:1}}
		\begin{tabular}{llrrrrr}
			\toprule
			\multicolumn{2}{c}{} & \multicolumn{5}{c}{Average of RMSE on evaluation points} \\
			\multicolumn{2}{c}{Number of observations} & 100   & 200   & 300   & 400   & 500 \\
			\midrule
			\multirow{7}[0]{*}{2-input} & SCKLS fixed bandwidth & 0.219 & 0.189 & 0.150 & 0.147 & 0.128 \\
			& SCKLS variable bandwidth & 0.212 & \textbf{0.176} & \textbf{0.125} & \textbf{0.132} & \textbf{0.103} \\
			& CNLS  & 0.350 & 0.299 & 0.260 & 0.284 & 0.265 \\
			& CWB in $p$-space & \textbf{0.206} & 0.186 & 0.174 & 0.154 & 0.143 \\
			& CWB in $y$-space & 0.259 & 0.228 & 0.228 & 0.172 & 0.167 \\
			& LL & 0.247 & 0.182 & 0.167 & 0.171 & 0.156 \\
			\cmidrule{2-7}
			& Cobb--Douglas  & 0.076 & 0.076 & 0.049 & 0.040 & 0.043 \\
			\midrule
			\multirow{7}[0]{*}{3-input} & SCKLS fixed bandwidth & \textbf{0.283} & \textbf{0.231} & 0.238 & 0.213 & 0.215 \\
			& SCKLS variable bandwidth & 0.292 & 0.237 & \textbf{0.235} & \textbf{0.196} & \textbf{0.187} \\
			& CNLS  & 0.529 & 0.587 & 0.540 & 0.589 & 0.598 \\
			& CWB in $p$-space & 0.291 & 0.289 & 0.269 & 0.252 & 0.233 \\
			& CWB in $y$-space & 0.314 & 0.474 & 0.265 & 0.346 & 0.261 \\
			& LL & 0.336 & 0.340 & 0.360 & 0.326 & 0.264 \\
			\cmidrule{2-7}
			& Cobb--Douglas  & 0.116 & 0.098 & 0.080 & 0.052 & 0.046 \\
			\midrule
			\multirow{7}[0]{*}{4-input} & SCKLS fixed bandwidth & \textbf{0.321} & 0.357 & \textbf{0.329} & \textbf{0.308} & \textbf{0.290} \\
			& SCKLS variable bandwidth & 0.378 & \textbf{0.348} & 0.363 & 0.320 & 0.301 \\
			& CNLS  & 0.845 & 0.873 & 0.901 & 0.827 & 0.792 \\
			& CWB in $p$-space & 0.360 & 0.385 & 0.358 & 0.361 & 0.325 \\
			& CWB in $y$-space & 0.355 & 0.470 & 0.338 & 0.410 & 0.602 \\
			& LL & 0.482 & 0.527 & 0.483 & 0.495 & 0.445 \\
			\cmidrule{2-7}
			& Cobb--Douglas  & 0.146 & 0.091 & 0.115 & 0.081 & 0.080 \\
			\bottomrule
		\end{tabular}%
		\label{tab:A5.exp1grd}%
	\end{table}%

	% Table generated by Excel2LaTeX from sheet 'TableA6'
	\begin{table}[ht]
		\footnotesize
		\centering
		\caption{Computational time for Experiment~\ref{exp:1}}
		\begin{tabular}{llrrrrr}
			\toprule
			\multicolumn{2}{c}{} & \multicolumn{5}{c}{\shortstack{Average of computational time in seconds;\\(percentage of Afriat constraints included\\in the final optimization problem)}} \\
			\multicolumn{2}{c}{Number of observations} & 100   & 200   & 300   & 400   & 500 \\
			\midrule
			\multirow{10}[0]{*}{2-input} & \multirow{2}[0]{*}{SCKLS fixed bandwidth} & 14.1  & 13.3  & 42.2  & 34.7  & 77.4 \\
			&       & \multicolumn{1}{c}{(6.14\%)} & \multicolumn{1}{c}{(5.28\%)} & \multicolumn{1}{c}{(8.86\%)} & \multicolumn{1}{c}{(7.80\%)} & \multicolumn{1}{c}{(8.31\%)} \\
			& \multirow{2}[0]{*}{SCKLS variable bandwidth} & 16.4  & 33.9  & 27.6  & 36.0  & \textbf{50.6} \\
			&       & \multicolumn{1}{c}{(3.47\%)} & \multicolumn{1}{c}{(3.44\%)} & \multicolumn{1}{c}{(3.34\%)} & \multicolumn{1}{c}{(3.22\%)} & \multicolumn{1}{c}{(3.53\%)} \\
			& \multirow{2}[0]{*}{CNLS} & \textbf{2.0} & \textbf{6.1} & \textbf{16.5} & \textbf{26.5} & 55.3 \\
			&       & \multicolumn{1}{c}{(100\%)} & \multicolumn{1}{c}{(100\%)} & \multicolumn{1}{c}{(100\%)} & \multicolumn{1}{c}{(100\%)} & \multicolumn{1}{c}{(100\%)} \\
			& \multirow{2}[0]{*}{CWB in $p$-space} & 24.1  & 33.2  & 76.6  & 82.3  & 130 \\
			&       & \multicolumn{1}{c}{(2.39\%)} & \multicolumn{1}{c}{(2.35\%)} & \multicolumn{1}{c}{(2.35\%)} & \multicolumn{1}{c}{(2.35\%)} & \multicolumn{1}{c}{(2.35\%)} \\
			& \multirow{2}[0]{*}{CWB in $y$-space} & 39.3  & 92.7  & 111   & 190   & 233 \\
			&       & \multicolumn{1}{c}{(2.35\%)} & \multicolumn{1}{c}{(2.35\%)} & \multicolumn{1}{c}{(2.35\%)} & \multicolumn{1}{c}{(2.35\%)} & \multicolumn{1}{c}{(2.36\%)} \\
			\midrule
			\multirow{10}[0]{*}{3-input} & \multirow{2}[0]{*}{SCKLS fixed bandwidth} & 26.9  & 40.4  & 45.5  & 67.3  & 136 \\
			&       & \multicolumn{1}{c}{(16.0\%)} & \multicolumn{1}{c}{(16.6\%)} & \multicolumn{1}{c}{(16.3\%)} & \multicolumn{1}{c}{(16.4\%)} & \multicolumn{1}{c}{(16.2\%)} \\
			& \multirow{2}[0]{*}{SCKLS variable bandwidth} & 20.0  & 42.0  & 37.4  & \textbf{47.1} & \textbf{58.2} \\
			&       & \multicolumn{1}{c}{(15.7\%)} & \multicolumn{1}{c}{(15.9\%)} & \multicolumn{1}{c}{(15.8\%)} & \multicolumn{1}{c}{(15.8\%)} & \multicolumn{1}{c}{(15.9\%)} \\
			& \multirow{2}[0]{*}{CNLS} & \textbf{3.8} & \textbf{16.4} & \textbf{37.0} & 82.9  & 161 \\
			&       & \multicolumn{1}{c}{(100\%)} & \multicolumn{1}{c}{(100\%)} & \multicolumn{1}{c}{(100\%)} & \multicolumn{1}{c}{(100\%)} & \multicolumn{1}{c}{(100\%)} \\
			& \multirow{2}[0]{*}{CWB in $p$-space} & 47.6  & 71.5  & 100   & 202   & 255 \\
			&       & \multicolumn{1}{c}{(15.5\%)} & \multicolumn{1}{c}{(15.5\%)} & \multicolumn{1}{c}{(15.5\%)} & \multicolumn{1}{c}{(15.5\%)} & \multicolumn{1}{c}{(15.5\%)} \\
			& \multirow{2}[0]{*}{CWB in $y$-space} & 120   & 357   & 443   & 529   & 424 \\
			&       & \multicolumn{1}{c}{(15.5\%)} & \multicolumn{1}{c}{(15.5\%)} & \multicolumn{1}{c}{(15.5\%)} & \multicolumn{1}{c}{(15.5\%)} & \multicolumn{1}{c}{(15.5\%)} \\
			\midrule
			\multirow{10}[0]{*}{4-input} & \multirow{2}[0]{*}{SCKLS fixed bandwidth} & 47.5  & 71.6  & 77.4  & 166   & 235 \\
			&       & \multicolumn{1}{c}{(40.1\%)} & \multicolumn{1}{c}{(39.9\%)} & \multicolumn{1}{c}{(39.9\%)} & \multicolumn{1}{c}{(40.0\%)} & \multicolumn{1}{c}{(39.8\%)} \\
			& \multirow{2}[0]{*}{SCKLS variable bandwidth} & 26.8  & 45.6  & \textbf{46.8} & \textbf{60.5} & \textbf{74.8} \\
			&       & \multicolumn{1}{c}{(39.9\%)} & \multicolumn{1}{c}{(40.0\%)} & \multicolumn{1}{c}{(39.8\%)} & \multicolumn{1}{c}{(39.9\%)} & \multicolumn{1}{c}{(39.8\%)} \\
			& \multirow{2}[0]{*}{CNLS} & \textbf{5.8} & \textbf{22.4} & 79.1  & 139.8 & 287.8 \\
			&       & \multicolumn{1}{c}{(100\%)} & \multicolumn{1}{c}{(100\%)} & \multicolumn{1}{c}{(100\%)} & \multicolumn{1}{c}{(100\%)} & \multicolumn{1}{c}{(100\%)} \\
			& \multirow{2}[0]{*}{CWB in $p$-space} & 68.8  & 136   & 196   & 327   & 442 \\
			&       & \multicolumn{1}{c}{(39.8\%)} & \multicolumn{1}{c}{(39.8\%)} & \multicolumn{1}{c}{(39.8\%)} & \multicolumn{1}{c}{(39.8\%)} & \multicolumn{1}{c}{(39.8\%)} \\
			& \multirow{2}[0]{*}{CWB in $y$-space} & 91.3  & 175   & 195   & 535   & 545 \\
			&       & \multicolumn{1}{c}{(39.8\%)} & \multicolumn{1}{c}{(39.8\%)} & \multicolumn{1}{c}{(39.8\%)} & \multicolumn{1}{c}{(39.8\%)} & \multicolumn{1}{c}{(39.8\%)} \\
			\bottomrule
		\end{tabular}%
		\label{tab:A6.exp1time}%
	\end{table}%
	
	We also conduct simulations with different bandwidths to analyze the sensitivity of each estimator to bandwidths. We estimate SCKLS with fixed bandwidth, CWB in $p$-space and local linear with bandwidth $h\in[0,10]$ with an increment by 0.01 for 1-input setting, and we use bandwidth $\bm{h}\in[0,5] \times [0,5]$ with an increment by 0.25 for 2-input setting. We perform 100 simulations for each bandwidth, and compute the optimal bandwidth with LOOCV for each simulation. Figure \ref{fig:A1.RMSE} displays the average RMSE of each estimator. The distribution of bandwidths selected by LOOCV are shown in the histogram. The instances when SCKLS, CWB-$p$, and local linear provide the lowest RMSE are shown in light gray, gray and dark gray respectively on the histogram. For one-input scenario, the SCKLS and CWB estimator perform similar for bandwidth between 0.25 - 2.25 as shown by the closeness of the light gray and gray curves in (a). In contrast, for two-input scenario, the SCKLS estimator performs better for most of the LOOCV values as shown by the majority of the histogram colored in light gray. This indicates that LOOCV calculate for unconstrained estimator provide bandwidths that work well for the SCKLS estimator. 
	
	\begin{figure}
		\centering
		\subfloat[One-input]{\includegraphics[width=0.5\textwidth]{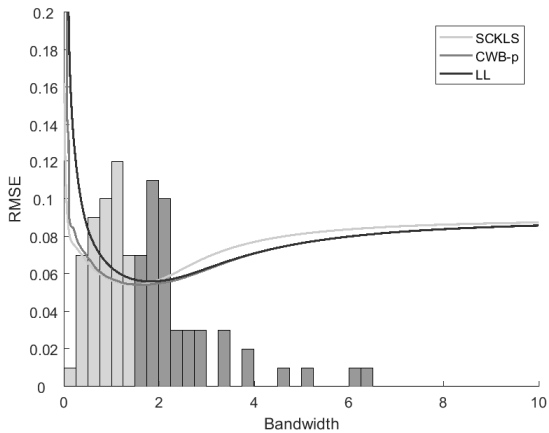}\label{fig:RMSE1wCWB}}
		\hfill
		\subfloat[Two-input]{\includegraphics[width=0.5\textwidth]{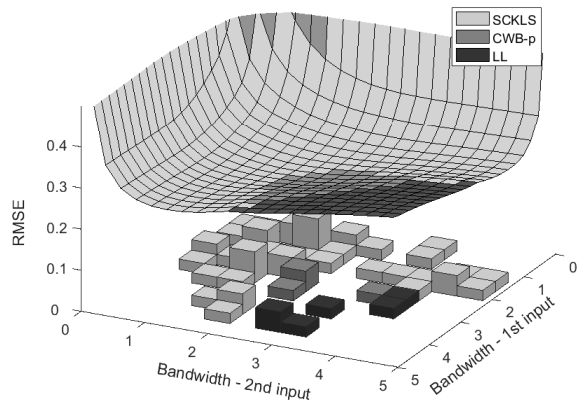}\label{fig:RMSE2wCWB}}
		\caption{The histogram shows the distribution of bandwidths selected by LOOCV. The curves show the relative performance of each estimator.}
		\label{fig:A1.RMSEwCWB}
	\end{figure}

	\clearpage
	\subsection{Uniform input -- low signal-to-noise ratio} \label{App:AppendixF}	
	
	%\label{sec:5.2.exp2}
	%\begin{experiment}
	%\label{exp:2}
		
	We consider a Cobb--Douglas production function with $d$-inputs and one-output, 
	\[
	g_0(x_1,\ldots,x_d)=\prod_{k=1}^{d}x_k^{\frac{0.8}{d}}. 
	\]
	For each pair $(\bm{X}_j,y_j)$, each component of the input, $\bm{X}_{jk}$, is randomly and independently drawn from uniform distribution $unif[1,10]$, and the additive noise, $\epsilon_j$, is randomly and independently sampled from a normal distribution, $N(0,1.3^2)$. We consider 15 different scenarios with different numbers of observations (100, 200, 300, 400 and 500) and input dimension (2, 3 and 4). The number of evaluation points is fixed at 400, and set as a uniform grid.
			This experiment has a higher noise level in the data generation process relative to Experiment~\ref{exp:1}.	
	%\end{experiment}
	
	We compare following seven estimators: SCKLS with fixed bandwidth, SCKLS with variable bandwidth, CNLS, CWB in $p$-space, CWB in $y$-space, LL, and parametric Cobb--Douglas function estimated via ordinary least squares (OLS). Table~\ref{tab:A7.exp2obs} and Table~\ref{tab:A8.exp2grd} show the RMSE of this experiment on observation points and evaluation points respectively.
	
	% Table generated by Excel2LaTeX from sheet 'TableA7'
	\begin{table}[ht]
		\footnotesize
		\centering
		\caption{RMSE on observation points for Experiment: uniform input with low signal-to-noise ratio}
		\begin{tabular}{llrrrrr}
			\toprule
			\multicolumn{2}{c}{} & \multicolumn{5}{c}{Average of RMSE on observation points} \\
			\multicolumn{2}{c}{Number of observations} & 100   & 200   & 300   & 400   & 500 \\
			\midrule
			\multirow{7}[0]{*}{2-input} & SCKLS fixed bandwidth & \textbf{0.239} & 0.203 & 0.203 & 0.155 & 0.140 \\
			& SCKLS variable bandwidth & 0.240 & \textbf{0.185} & \textbf{0.168} & \textbf{0.139} & \textbf{0.119} \\
			& CNLS  & 0.279 & 0.231 & 0.194 & 0.168 & 0.151 \\
			& CWB in $p$-space & 0.314 & 0.215 & 0.237 & 0.275 & 0.151 \\
			& CWB in $y$-space & 0.241 & 0.229 & 0.173 & 0.178 & 0.206 \\
			& LL & 0.287 & 0.244 & 0.230 & 0.214 & 0.161 \\
			\cmidrule{2-7}
			& Cobb--Douglas  & 0.109 & 0.108 & 0.081 & 0.042 & 0.048 \\
			\midrule
			\multirow{7}[0]{*}{3-input} & SCKLS fixed bandwidth & 0.292 & 0.263 & 0.221 & 0.204 & 0.184 \\
			& SCKLS variable bandwidth & \textbf{0.281} & \textbf{0.242} & \textbf{0.198} & \textbf{0.180} & \textbf{0.175} \\
			& CNLS  & 0.379 & 0.303 & 0.275 & 0.224 & 0.214 \\
			& CWB in $p$-space & 0.318 & 0.306 & 0.308 & 0.244 & 0.214 \\
			& CWB in $y$-space & 0.281 & 0.273 & 0.225 & 0.320 & 0.271 \\
			& LL & 0.333 & 0.306 & 0.288 & 0.259 & 0.214 \\
			\cmidrule{2-7}
			& Cobb--Douglas  & 0.176 & 0.118 & 0.101 & 0.084 & 0.072 \\
			\midrule
			\multirow{7}[0]{*}{4-input} & SCKLS fixed bandwidth & 0.317 & 0.291 & 0.249 & 0.241 & 0.254 \\
			& SCKLS variable bandwidth & \textbf{0.290} & \textbf{0.254} & \textbf{0.236} & \textbf{0.222} & \textbf{0.215} \\
			& CNLS  & 0.491 & 0.356 & 0.311 & 0.293 & 0.313 \\
			& CWB in $p$-space & 0.400 & 0.318 & 0.273 & 0.260 & 0.289 \\
			& CWB in $y$-space & 0.312 & 0.338 & 0.262 & 0.365 & 0.453 \\
			& LL & 0.335 & 0.342 & 0.257 & 0.274 & 0.283 \\
			\cmidrule{2-7}
			& Cobb--Douglas  & 0.157 & 0.150 & 0.112 & 0.075 & 0.077 \\
			\bottomrule
		\end{tabular}%
		\label{tab:A7.exp2obs}%
	\end{table}%
	
	% Table generated by Excel2LaTeX from sheet 'TableA8'
	\begin{table}[p]
		\footnotesize
		\centering
		\caption{RMSE on evaluation points for Experiment: uniform input with low signal-to-noise ratio}
		\begin{tabular}{llrrrrr}
			\toprule
			\multicolumn{2}{c}{} & \multicolumn{5}{c}{Average of RMSE on evaluation points} \\
			\multicolumn{2}{c}{Number of observations} & 100   & 200   & 300   & 400   & 500 \\
			\midrule
			\multirow{7}[0]{*}{2-input} & SCKLS fixed bandwidth & \textbf{0.253} & 0.225 & 0.222 & 0.172 & 0.160 \\
			& SCKLS variable bandwidth & 0.255 & \textbf{0.205} & \textbf{0.179} & \textbf{0.149} & \textbf{0.135} \\
			& CNLS  & 0.319 & 0.355 & 0.334 & 0.255 & 0.267 \\
			& CWB in $p$-space & 0.329 & 0.239 & 0.262 & 0.305 & 0.177 \\
			& CWB in $y$-space & 0.263 & 0.241 & 0.198 & 0.228 & 0.180 \\
			& LL & 0.330 & 0.272 & 0.257 & 0.239 & 0.194 \\
			\cmidrule{2-7}
			& Cobb--Douglas  & 0.112 & 0.112 & 0.083 & 0.044 & 0.049 \\
			\midrule
			\multirow{7}[0]{*}{3-input} & SCKLS fixed bandwidth & 0.367 & 0.339 & 0.302 & 0.268 & 0.231 \\
			& SCKLS variable bandwidth & \textbf{0.364} & \textbf{0.303} & \textbf{0.256} & \textbf{0.230} & \textbf{0.224} \\
			& CNLS  & 0.743 & 0.778 & 0.744 & 0.696 & 0.620 \\
			& CWB in $p$-space & 0.398 & 0.392 & 0.434 & 0.336 & 0.274 \\
			& CWB in $y$-space & 0.401 & 0.473 & 0.385 & 0.450 & 0.525 \\
			& LL & 0.452 & 0.444 & 0.438 & 0.398 & 0.302 \\
			\cmidrule{2-7}
			& Cobb--Douglas  & 0.202 & 0.130 & 0.110 & 0.093 & 0.079 \\
			\midrule
			\multirow{7}[0]{*}{4-input} & SCKLS fixed bandwidth & \textbf{0.405} & 0.460 & \textbf{0.349} & \textbf{0.350} & 0.347 \\
			& SCKLS variable bandwidth & 0.419 & \textbf{0.434} & 0.375 & 0.354 & \textbf{0.315} \\
			& CNLS  & 1.019 & 0.950 & 0.985 & 1.043 & 1.106 \\
			& CWB in $p$-space & 0.514 & 0.520 & 0.393 & 0.390 & 0.452 \\
			& CWB in $y$-space & 0.514 & 0.513 & 0.425 & 0.501 & 0.708 \\
			& LL & 0.524 & 0.626 & 0.451 & 0.491 & 0.550 \\
			\cmidrule{2-7}
			& Cobb--Douglas  & 0.187 & 0.194 & 0.134 & 0.092 & 0.091 \\
			\bottomrule
		\end{tabular}%
		\label{tab:A8.exp2grd}%
	\end{table}%
	
	\clearpage
    \subsection{Different numbers of evaluation points (Experiment~\ref{exp:4})} \label{App:AppendixH}	
	
	We compare following four estimators: SCKLS with fixed bandwidth, SCKLS with variable bandwidth, CWB in $p$-space and CWB in $y$-space. Table~\ref{tab:A11.exp4obs} and Table~\ref{tab:A12.exp4grd} show the RMSEs of Experiment~\ref{exp:4} on observation points and evaluation points respectively. In addition, Table~\ref{tab:A13.exp4time} shows the computational time of Experiment~\ref{exp:4} for each estimator.

	% Table generated by Excel2LaTeX from sheet 'TableA11'
	\begin{table}[ht]
		\footnotesize
		\centering
		\caption{RMSE on observation points for Experiment~\ref{exp:4}}
		\begin{tabular}{llP{.8in}P{.8in}P{.8in}}
			\toprule
			\multicolumn{2}{c}{} & \multicolumn{3}{c}{Average of RMSE on observation points} \\
			\multicolumn{2}{c}{Number of evaluation points} & 100   & 300   & 500 \\
			\midrule
			\multirow{4}[0]{*}{2-input} & SCKLS fixed bandwidth & 0.142 & 0.141 & 0.141 \\
			& SCKLS variable bandwidth & \textbf{0.113} & \textbf{0.112} & \textbf{0.112} \\
			& CWB in $p$-space & 0.149 & 0.151 & 0.156 \\
			& CWB in $y$-space & 0.225 & 0.122 & 0.129 \\
			\midrule
			\multirow{4}[0]{*}{3-input} & SCKLS fixed bandwidth & 0.198 & 0.203 & 0.197 \\
			& SCKLS variable bandwidth & \textbf{0.169} & \textbf{0.167} & \textbf{0.166} \\
			& CWB in $p$-space & 0.218 & 0.234 & 0.231 \\
			& CWB in $y$-space & 0.345 & 0.241 & 0.222 \\
			\midrule
			\multirow{4}[0]{*}{4-input} & SCKLS fixed bandwidth & 0.239 & 0.207 & 0.206 \\
			& SCKLS variable bandwidth & \textbf{0.195} & \textbf{0.192} & \textbf{0.191} \\
			& CWB in $p$-space & 0.219 & 0.227 & 0.296 \\
			& CWB in $y$-space & 0.466 & 0.290 & 0.292 \\
			\bottomrule
		\end{tabular}%
		\label{tab:A11.exp4obs}%
	\end{table}%
	
	% Table generated by Excel2LaTeX from sheet 'TableA12'
	\begin{table}[ht]
		\footnotesize
		\centering
		\caption{RMSE on evaluation points for Experiment~\ref{exp:4}}
		\begin{tabular}{llP{.8in}P{.8in}P{.8in}}
			\toprule
			\multicolumn{2}{c}{} & \multicolumn{3}{c}{Average of RMSE on evaluation points} \\
			\multicolumn{2}{c}{Number of evaluation points} & 100   & 300   & 500 \\
			\midrule
			\multirow{4}[0]{*}{2-input} & SCKLS fixed bandwidth & 0.181 & 0.164 & 0.158 \\
			& SCKLS variable bandwidth & \textbf{0.140} & \textbf{0.128} & \textbf{0.124} \\
			& CWB in $p$-space & 0.195 & 0.180 & 0.179 \\
			& CWB in $y$-space & 0.262 & 0.162 & 0.169 \\
			\midrule
			\multirow{4}[0]{*}{3-input} & SCKLS fixed bandwidth & 0.304 & 0.267 & 0.257 \\
			& SCKLS variable bandwidth & \textbf{0.242} & \textbf{0.213} & \textbf{0.205} \\
			& CWB in $p$-space & 0.332 & 0.329 & 0.302 \\
			& CWB in $y$-space & 0.792 & 0.582 & 0.559 \\
			\midrule
			\multirow{4}[0]{*}{4-input} & SCKLS fixed bandwidth & \textbf{0.383} & \textbf{0.296} & 0.270 \\
			& SCKLS variable bandwidth & 0.386 & 0.304 & \textbf{0.265} \\
			& CWB in $p$-space & 0.403 & 0.359 & 0.415 \\
			& CWB in $y$-space & 1.040 & 0.352 & 0.381 \\
			\bottomrule
		\end{tabular}%
		\label{tab:A12.exp4grd}%
	\end{table}%

	% Table generated by Excel2LaTeX from sheet 'TableA13'
	\begin{table}[ht]
		\footnotesize
		\centering
		\caption{Computational time for Experiment~\ref{exp:4}}
		\begin{tabular}{llP{.8in}P{.8in}P{.8in}}
			\toprule
			\multicolumn{2}{c}{} & \multicolumn{3}{c}{\shortstack{Average of computational time in seconds;\\(percentage of Afriat constraints included\\in the final optimization)}} \\
			\multicolumn{2}{c}{Number of evaluation points} & 100   & 300   & 500 \\
			\midrule
			\multirow{8}[0]{*}{2-input} & \multirow{2}[0]{*}{SCKLS fixed bandwidth} & \multicolumn{1}{c}{26.6} & \multicolumn{1}{c}{28.3} & \multicolumn{1}{c}{34} \\
			&       & \multicolumn{1}{c}{(11.7\%)} & \multicolumn{1}{c}{(6.6\%)} & \multicolumn{1}{c}{(5.4\%)} \\
			& \multirow{2}[0]{*}{SCKLS variable bandwidth} & \multicolumn{1}{c}{\textbf{21.3}} & \multicolumn{1}{c}{\textbf{21.6}} & \multicolumn{1}{c}{\textbf{24.9}} \\
			&       & \multicolumn{1}{c}{(9.9\%)} & \multicolumn{1}{c}{(4.4\%)} & \multicolumn{1}{c}{(3.2\%)} \\
			& \multirow{2}[0]{*}{CWB in $p$-space} & \multicolumn{1}{c}{41} & \multicolumn{1}{c}{56.5} & \multicolumn{1}{c}{74.2} \\
			&       & \multicolumn{1}{c}{(8.8\%)} & \multicolumn{1}{c}{(3.2\%)} & \multicolumn{1}{c}{(2.0\%)} \\
			& \multirow{2}[0]{*}{CWB in $y$-space} & \multicolumn{1}{c}{52.8} & \multicolumn{1}{c}{103} & \multicolumn{1}{c}{146} \\
			&       & \multicolumn{1}{c}{(8.8\%)} & \multicolumn{1}{c}{(3.2\%)} & \multicolumn{1}{c}{(2.0\%)} \\
			\midrule
			\multirow{8}[0]{*}{3-input} & \multirow{2}[0]{*}{SCKLS fixed bandwidth} & \multicolumn{1}{c}{84.8} & \multicolumn{1}{c}{112} & \multicolumn{1}{c}{134} \\
			&       & \multicolumn{1}{c}{(29.1\%)} & \multicolumn{1}{c}{(16.7\%)} & \multicolumn{1}{c}{(13.3\%)} \\
			& \multirow{2}[0]{*}{SCKLS variable bandwidth} & \multicolumn{1}{c}{\textbf{21.1}} & \multicolumn{1}{c}{\textbf{37.2}} & \multicolumn{1}{c}{\textbf{59.1}} \\
			&       & \multicolumn{1}{c}{(28.5\%)} & \multicolumn{1}{c}{(15.8\%)} & \multicolumn{1}{c}{(12.4\%)} \\
			& \multirow{2}[0]{*}{CWB in $p$-space} & \multicolumn{1}{c}{121} & \multicolumn{1}{c}{221} & \multicolumn{1}{c}{310} \\
			&       & \multicolumn{1}{c}{(28.2\%)} & \multicolumn{1}{c}{(15.5\%)} & \multicolumn{1}{c}{(12.2\%)} \\
			& \multirow{2}[0]{*}{CWB in $y$-space} & \multicolumn{1}{c}{181} & \multicolumn{1}{c}{625} & \multicolumn{1}{c}{948} \\
			&       & \multicolumn{1}{c}{(28.2\%)} & \multicolumn{1}{c}{(15.5\%)} & \multicolumn{1}{c}{(12.2\%)} \\
			\midrule
			\multirow{8}[0]{*}{4-input} & \multirow{2}[0]{*}{SCKLS fixed bandwidth} & \multicolumn{1}{c}{149} & \multicolumn{1}{c}{170} & \multicolumn{1}{c}{597} \\
			&       & \multicolumn{1}{c}{(62.3\%)} & \multicolumn{1}{c}{(40.0\%)} & \multicolumn{1}{c}{(27.7\%)} \\
			& \multirow{2}[0]{*}{SCKLS variable bandwidth} & \multicolumn{1}{c}{\textbf{24.6}} & \multicolumn{1}{c}{\textbf{52.7}} & \multicolumn{1}{c}{\textbf{468}} \\
			&       & \multicolumn{1}{c}{(62.1\%)} & \multicolumn{1}{c}{(39.9\%)} & \multicolumn{1}{c}{(27.5\%)} \\
			& \multirow{2}[0]{*}{CWB in $p$-space} & \multicolumn{1}{c}{175} & \multicolumn{1}{c}{275} & \multicolumn{1}{c}{729} \\
			&       & \multicolumn{1}{c}{(61.9\%)} & \multicolumn{1}{c}{(39.8\%)} & \multicolumn{1}{c}{(27.4\%)} \\
			& \multirow{2}[0]{*}{CWB in $y$-space} & \multicolumn{1}{c}{189} & \multicolumn{1}{c}{288} & \multicolumn{1}{c}{579} \\
			&       & \multicolumn{1}{c}{(61.9\%)} & \multicolumn{1}{c}{(39.8\%)} & \multicolumn{1}{c}{(27.4\%)} \\
			\bottomrule
		\end{tabular}%
		\label{tab:A13.exp4time}%
	\end{table}%

	\clearpage
	\subsection{Non-uniform input} \label{App:AppendixG}	

	\begin{experiment}
		\label{exp:3}
		We consider a Cobb--Douglas production function with $d$-inputs and one-output, 
		\[
		g_0(x_1,\ldots,x_d)=\prod_{k=1}^{d}x_k^{\frac{0.8}{d}}.
		\]
		For each pair $(\bm{X}_j,y_j)$, each component of the input, $\bm{X}_{jk}$, is randomly and independently drawn from a truncated exponential distribution with density function 
		\[
		f(x)= \frac{3}{e^{-3}-e^{-30}}e^{-3x} \mathbf{1}_{\{x \in [1,10]\}},
		\]
		and the additive noise, $\epsilon_j$, is randomly sampled from a normal distribution, $N(0,0.7^2)$.
		 We consider 15 different scenarios with different numbers of observations (100, 200, 300, 400 and 500) and input dimension (2, 3 and 4). The number of evaluation point is fixed at 400. Note that this experiment only differs from Experiment~\ref{exp:1} in that the distribution of inputs is skewed and thus non-uniform.
	\end{experiment}
	
	We compare following seven estimators: SCKLS with fixed bandwidth with uniform/non-uniform grid, SCKLS with variable bandwidth with uniform/non-uniform grid, CNLS, CWB in $p$-space with uniform/non-uniform grid. These extension of SCKLS were presented in detail in Appendix~\ref{App:Extensions}. Table~\ref{tab:A9.exp3obs} and Table~\ref{tab:A10.exp3grd} show the RMSEs of Experiment~\ref{exp:3} on observation points and evaluation points respectively. A uniform grid is used like in Experiment~\ref{exp:1}. As the dimension of input space and the number of observations increase, SCKLS with variable bandwidth performs better than the fixed bandwidth estimator. SCKLS with non-uniform grid performs better than SCKLS with uniform grid for almost all scenarios, largely due to the fact that the DGP has non-uniform input. Consequently, we conclude that variable bandwidth methods, such as $k$-NN approach, and non-uniform grid could be useful to handle skewed input data which is a common feature of census manufacturing data which is the type of data we considered in the application of the main manuscript.

	% Table generated by Excel2LaTeX from sheet 'TableA9'
	\begin{table}[ht]
		\footnotesize
		\centering
		\caption{RMSE on observation points for Experiment: non-uniform input}
		\begin{tabular}{llrrrrr}
			\toprule
			\multicolumn{2}{c}{} & \multicolumn{5}{c}{Average of RMSE on observation points} \\
			\multicolumn{2}{c}{Number of observations} & 100   & 200   & 300   & 400   & 500 \\
			\midrule
			\multirow{7}[0]{*}{2-input} & SCKLS fixed/uniform & 0.179 & 0.151 & 0.144 & 0.121 & 0.108 \\
			& SCKLS fixed/non-uniform & 0.185 & 0.153 & 0.159 & 0.123 & 0.107 \\
			& SCKLS variable/uniform & 0.183 & 0.156 & 0.142 & 0.125 & 0.104 \\
			& SCKLS variable/non-uniform & \textbf{0.176} & \textbf{0.144} & \textbf{0.132} & \textbf{0.114} & \textbf{0.093} \\
			& CNLS  & 0.193 & 0.160 & 0.140 & 0.130 & 0.117 \\
			& CWB $p$-space/uniform & 0.256 & 0.162 & 0.180 & 0.139 & 0.125 \\
			& CWB $p$-space/non-uniform & 0.243 & 0.160 & 0.174 & 0.135 & 0.125 \\
			\midrule
			\multirow{7}[0]{*}{3-input} & SCKLS fixed/uniform & \textbf{0.197} & 0.184 & 0.172 & 0.164 & 0.167 \\
			& SCKLS fixed/non-uniform & 0.200 & 0.181 & 0.173 & 0.161 & 0.172 \\
			& SCKLS variable/uniform & 0.212 & 0.187 & 0.170 & 0.175 & 0.170 \\
			& SCKLS variable/non-uniform & 0.210 & \textbf{0.180} & \textbf{0.162} & \textbf{0.160} & \textbf{0.155} \\
			& CNLS  & 0.303 & 0.246 & 0.201 & 0.185 & 0.166 \\
			& CWB $p$-space/uniform & 0.243 & 0.436 & 0.173 & 0.174 & 0.184 \\
			& CWB $p$-space/non-uniform & 0.233 & 0.194 & 0.176 & 0.165 & 0.173 \\
			\midrule
			\multirow{7}[0]{*}{4-input} & SCKLS fixed/uniform & 0.219 & 0.211 & 0.196 & 0.209 & 0.187 \\
			& SCKLS fixed/non-uniform & 0.210 & 0.206 & 0.181 & 0.197 & 0.180 \\
			& SCKLS variable/uniform & 0.208 & \textbf{0.193} & 0.167 & 0.171 & 0.170 \\
			& SCKLS variable/non-uniform & \textbf{0.206} & 0.193 & \textbf{0.164} & \textbf{0.169} & \textbf{0.168} \\
			& CNLS  & 0.347 & 0.292 & 0.250 & 0.228 & 0.218 \\
			& CWB $p$-space/uniform & 0.219 & 0.205 & 0.205 & 0.184 & 0.218 \\
			& CWB $p$-space/non-uniform & 0.221 & 0.205 & 0.182 & 0.170 & 0.170 \\
			\bottomrule
		\end{tabular}%
		\label{tab:A9.exp3obs}%
	\end{table}%

	% Table generated by Excel2LaTeX from sheet 'TableA10'
	\begin{table}[ht]
		\footnotesize
		\centering
		\caption{RMSE on evaluation points for Experiment: non-uniform input}
		\begin{tabular}{llrrrrr}
			\toprule
			\multicolumn{2}{c}{} & \multicolumn{5}{c}{Average of RMSE on evaluation points} \\
			\multicolumn{2}{c}{Number of observations} & 100   & 200   & 300   & 400   & 500 \\
			\midrule
			\multirow{7}[0]{*}{2-input} & SCKLS fixed/uniform & 0.262 & 0.220 & 0.244 & 0.157 & 0.196 \\
			& SCKLS fixed/non-uniform & 0.212 & 0.174 & 0.195 & 0.138 & 0.131 \\
			& SCKLS variable/uniform & 0.246 & 0.204 & 0.192 & 0.142 & 0.136 \\
			& SCKLS variable/non-uniform & \textbf{0.193} & \textbf{0.160} & \textbf{0.145} & \textbf{0.120} & \textbf{0.100} \\
			& CNLS  & 0.435 & 0.402 & 0.404 & 0.379 & 0.381 \\
			& CWB $p$-space/uniform & 0.422 & 0.287 & 0.376 & 0.246 & 0.264 \\
			& CWB $p$-space/non-uniform & 0.283 & 0.186 & 0.215 & 0.159 & 0.162 \\
			\midrule
			\multirow{7}[0]{*}{3-input} & SCKLS fixed/uniform & 0.323 & 0.308 & 0.311 & 0.286 & 0.293 \\
			& SCKLS fixed/non-uniform & \textbf{0.268} & 0.254 & 0.259 & 0.235 & 0.249 \\
			& SCKLS variable/uniform & 0.335 & 0.303 & 0.281 & 0.262 & 0.254 \\
			& SCKLS variable/non-uniform & 0.278 & \textbf{0.243} & \textbf{0.219} & \textbf{0.212} & \textbf{0.196} \\
			& CNLS  & 0.828 & 0.824 & 0.828 & 0.786 & 0.782 \\
			& CWB $p$-space/uniform & 0.438 & 0.684 & 0.357 & 0.363 & 0.350 \\
			& CWB $p$-space/non-uniform & 0.315 & 0.265 & 0.257 & 0.235 & 0.242 \\
			\midrule
			\multirow{7}[0]{*}{4-input} & SCKLS fixed/uniform & 0.406 & 0.398 & 0.397 & 0.404 & 0.400 \\
			& SCKLS fixed/non-uniform & \textbf{0.339} & \textbf{0.343} & 0.333 & 0.371 & 0.331 \\
			& SCKLS variable/uniform & 0.417 & 0.423 & 0.368 & 0.364 & 0.356 \\
			& SCKLS variable/non-uniform & 0.359 & 0.359 & 0.313 & 0.302 & 0.280 \\
			& CNLS  & 1.129 & 1.107 & 1.220 & 1.196 & 1.223 \\
			& CWB $p$-space/uniform & 0.421 & 0.442 & 0.435 & 0.418 & 0.487 \\
			& CWB $p$-space/non-uniform & 0.354 & 0.344 & \textbf{0.308} & \textbf{0.286} & \textbf{0.280} \\
			\bottomrule
		\end{tabular}%
		\label{tab:A10.exp3grd}%
	\end{table}%

	\clearpage

	\subsection{Estimation with a misspecified shape} \label{App:AppendixJ2}	
%	\begin{experiment}
%	\label{exp:8}
	We use the DGP proposed by \cite{olesen2014maintaining} that is consistent with the regular ultra passum law \citep{frisch1964theory}, which appears to have an ``S''-shape.
	\begin{equation*}
	g_0(x_1,x_2) = F(h(x_1,x_2))
	\end{equation*}	
	where the scaling function is: $F(w)=\frac{15}{1+e^{-5\log (w)}}$, and the linear homogeneous core function is
	\[
	h(x_1,x_2)=\left(\beta x_1^{\frac{\sigma-1}{\sigma}}+(1-\beta)x_2^{\frac{\sigma-1}{\sigma}}\right)^{\frac{\sigma}{\sigma-1}}
	\] 
	with $\beta =0.45$ and $\sigma =1.51$.
	For $j=1,\ldots,n$, input, $\bm{X}_j=(X_{j1},X_{j2})'$, is generated in polar coordinates with angles $\eta$ and modulus $\omega$ independently uniformly distributed on $[0.05,{\pi}/{2}-0.05]$ and $[0,2.5]$, respectively. The additive noise, $\epsilon_j$, is randomly sampled from $N(0,0.7^2)$.
%	\end{experiment}
	
	Note that this DGP is not concave. Here we run this experiment to assess the performance of each estimator in case of shape misspecification. %It also supports the consistency argument in Section 4 empirically. 
	Table~\ref{tab:S-shapeObs} and Table~\ref{tab:S-shapeGrd} show the RMSEs of this experiment on observation points and evaluation points. %The SCKLS estimator is competitive to the CNLS estimator on observation points where both estimators have certain convergence properties even for the misspecified functional form (See \cite{lim2012consistency} and our Theorem 4 for details). 
	%Consistency for the misspecified functional form indicates that the estimator converges to the closest shape constrained functional estimate in the violated area. In case of this experiment, the CNLS estimator converges to closest concave/monotonic function in the violated area (convex). Since we could both theoretically and empirically show that the SCKLS estimator also has consistency for the misspecified functional form, the SCKLS estimator has relatively robust performance even with data generated by the misspecified functional form. 
	Figure~\ref{fig:S-shape1d} shows the estimation results with 1-input S-shape function from a typical run of SCKLS. The figure shows that the SCKLS estimator results in a linear estimates for areas where concavity is violated. Here the CWB estimator performs slightly worse when the function is misspecified.% although it has consistency properties shown by \cite{du2013nonparametric}. 
	We speculate that the main reason for this is that the optimization problem becomes too complicated to solve since intuitively there are many binding constraints when the data is generated by the misspecified functional form, and thus, it becomes hard for the solver to find a feasible solution and an improving direction.
	
	% Table generated by Excel2LaTeX from sheet 'S-shapeObs'
	\begin{table}[!htbp]
		\footnotesize
		\centering
		\caption{RMSE on observation points for Experiment: misspecified shape}
		\begin{tabular}{lrrrrr}
			\toprule
			& \multicolumn{5}{c}{Average of RMSE on observation points} \\
			Number of observations & 100   & 200   & 300   & 400   & 500 \\
			\midrule
			SCKLS fixed bandwidth & 1.424 & 1.435 & 1.405 & 1.392 & 1.421 \\
			CNLS  & \textbf{1.326} & \textbf{1.346} & \textbf{1.337} & \textbf{1.316} & \textbf{1.353} \\
			CWB in $p$-space & 6.310 & 6.731 & 6.602 & 5.909 & 6.110 \\
			\bottomrule
		\end{tabular}%
		\label{tab:S-shapeObs}%
	\end{table}%
	
	% Table generated by Excel2LaTeX from sheet 'S-shapeGrd'
	\begin{table}[!htbp]
		\footnotesize
		\centering
		\caption{RMSE on evaluation points for Experiment: misspecified shape}
		\begin{tabular}{lrrrrr}
			\toprule
			& \multicolumn{5}{c}{Average of RMSE on evaluation points} \\
			Number of observations & 100   & 200   & 300   & 400   & 500 \\
			\midrule
			SCKLS fixed bandwidth & \textbf{1.337} & \textbf{1.162} & \textbf{1.149} & \textbf{1.140} & \textbf{1.123} \\
			CNLS  & 1.375 & 1.424 & 1.404 & 1.403 & 1.385 \\
			CWB in $p$-space & 9.100 & 9.483 & 9.599 & 8.435 & 8.719 \\
			\bottomrule
		\end{tabular}%
		\label{tab:S-shapeGrd}%
	\end{table}%
	\begin{figure}[!htbp]
		\begin{center}
			\includegraphics[width=3.5in]{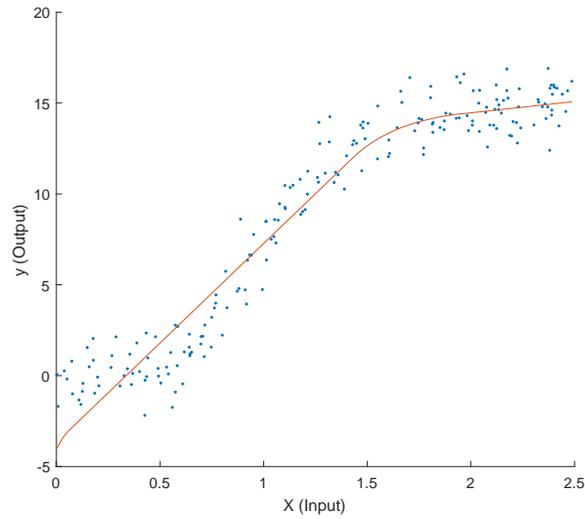}
		\end{center}
		\caption{A typical run of SCKLS when the truth is S-shaped. \label{fig:S-shape1d}}
	\end{figure}

	\clearpage
	\section{Semiparametric partially linear model} \label{App:AppendixI}	
	
	\subsection{The procedure}
	
	We develop a semiparametric partially linear model including the SCKLS estimator and a linear function of contextual variables. The partially linear model is often used in practice. The model estimated is represented as follows:
	\begin{equation*}
	y_j=\bm{Z}_j'\bm{\gamma}+g_0(\bm{X}_j)+\epsilon_j
	\end{equation*}
	where $\bm{Z}_j=(Z_{j1},Z_{j2},\ldots,Z_{jl})'$ denotes contextual variables and $\bm{\gamma}=(\gamma_1,\gamma_2,\ldots,\gamma_l)'$ is the coefficient of contextual variables, see \cite{johnson2011one,johnson2012one}. Then, we estimate the coefficient of contextual variable:
	\begin{equation*}
	\bm{\hat{\gamma}}=\left(\sum_{j=1}^{n}\tilde{\bm{Z}}_j\tilde{\bm{Z}}_j'\right)^{-1}\left(\sum_{j=1}^{n}\tilde{\bm{Z}}_j\tilde{y}_j\right)
	\end{equation*}
	where $\tilde{\bm{Z}}_j=\bm{Z}_j-\hat{E}[\bm{Z}_j|\bm{X}_j]$ and $\tilde{y}_j=y_j-\hat{E}[y_j|\bm{X}_j]$ respectively, and each conditional expectation is estimated by kernel estimation method such as local linear. Finally, we apply the SCKLS estimator to the data $\{\bm{X}_j,y_j-\bm{Z}_j'\hat{\bm{\gamma}}\}_{j=1}^n$. \cite{robinson1988root} proved that $\bm{\hat{\gamma}}$ is $n^{1/2}$-consistent for $\bm{\gamma}$ and asymptotically normal under regularity conditions. For details of the partially linear model, see \cite{li2007nonparametric}.

	\subsection{A simulation study}
%	We conduct an experiment to compare the performance of SCKLS and CNLS estimator with contextual variables. 
%	\begin{experiment}
%	\label{exp:9}	

    We show the effect of adding contextual variables $\bm{Z}_j$ to the estimation performance by comparing SCKLS with and without contextual variables. We use two different Cobb--Douglas production functions as the true DGP:
	\begin{equation}
	\label{eq:modelwithZ}
	g_0(\bm{x},z)=\prod_{k=1}^{d}x_k^{\frac{0.8}{d}}+z\gamma,
	\end{equation}
	
    \begin{equation}
	\label{eq:modelwithoutZ}
	g_0(\bm{x})=\prod_{k=1}^{d}x_k^{\frac{0.8}{d}},
	\end{equation}
	\noindent where for each $(\bm{X}_j, Z_j, y_j)$, the contextual variable $Z_j$ is a scalar value independent of $\bm{X}_j$ drawn randomly and independently from $uinf[0,1]$, the coefficient of the contextual variable $\gamma=5$, and other parameters follow DGP from Experiment \ref{exp:1}. We apply SCKLS with and without contextual variables to the data generated by the true production function (\ref{eq:modelwithZ}) and (\ref{eq:modelwithoutZ}), respectively.
%	\end{experiment}
	
	Table~\ref{tab:withZobs} and Table~\ref{tab:withZgrd} show the RMSEs of this experiment on observation points and evaluation points respectively. The RMSE is obtained by comparing estimates of production function and the true production function. We see that having extra contextual variables does not deteriorate the performance of SCKLS significantly, especially when the  input dimension is small and the number of observations is large. Our findings are consistent with the work of \cite{robinson1988root}. Since our application data in Section~\ref{sec:6.application} has only two-input, we expect that SCKLS with $Z$-variables tends not to deteriorate the estimator performance in our application.

	% Table generated by Excel2LaTeX from sheet 'Z2'
	\begin{table}[htbp]
		\centering
		\caption{RMSE on observation points for experiments with/without $Z$-variable}
		\begin{tabular}{rlrrrrr}
			\toprule
			\multicolumn{2}{c}{} & \multicolumn{5}{c}{Average of RMSE on observation points} \\
			\multicolumn{2}{c}{Number of observations} & 100   & 200   & 300   & 400   & 500 \\
			\midrule
			\multicolumn{1}{l}{2-input} & SCKLS-Z & 0.224 & 0.212 & 0.239 & 0.160 & 0.146 \\
			& SCKLS & 0.210 & 0.188 & 0.170 & 0.139 & 0.140 \\
			\midrule
			\multicolumn{1}{l}{3-input} & SCKLS-Z & 0.404 & 0.235 & 0.261 & 0.197 & 0.196 \\
			& SCKLS & {0.242} & {0.206} & {0.215} & {0.202} & {0.188} \\
			\midrule
			\multicolumn{1}{l}{4-input} & SCKLS-Z & 0.462 & 0.376 & 0.332 & 0.217 & 0.239 \\
			& SCKLS & {0.247} & {0.231} & {0.202} & {0.202} & {0.198} \\
			\bottomrule
		\end{tabular}%
		\label{tab:withZobs}%
	\end{table}%

	% Table generated by Excel2LaTeX from sheet 'Z2'
	\begin{table}[htbp]
		\centering
		\caption{RMSE on evaluation points for experiments with/without $Z$-variable}
		\begin{tabular}{rlrrrrr}
			\toprule
			\multicolumn{2}{c}{} & \multicolumn{5}{c}{Average of RMSE on evaluation points} \\
			\multicolumn{2}{c}{Number of observations} & 100   & 200   & 300   & 400   & 500 \\
			\midrule
			\multicolumn{1}{l}{2-input} & SCKLS-Z & 0.245 & 0.234 & 0.256 & 0.172 & 0.166 \\
			& SCKLS & {0.230} & {0.205} & {0.194} & {0.154} & {0.157} \\
			\midrule
			\multicolumn{1}{l}{3-input} & SCKLS-Z & 0.496 & 0.348 & 0.377 & 0.271 & 0.286 \\
			& SCKLS & {0.316} & {0.296} & {0.309} & {0.271} & {0.261} \\
			\midrule
			\multicolumn{1}{l}{4-input} & SCKLS-Z & 0.648 & 0.599 & 0.498 & 0.397 & 0.435 \\
			& SCKLS & {0.385} & {0.381} & {0.341} & {0.350} & {0.336} \\
			\bottomrule
		\end{tabular}%
		\label{tab:withZgrd}%
	\end{table}%

	%\clearpage
	%\section{Additional Experiments} \label{App:AppendixJ}
	% the \\ insures the section title is centered below the phrase: AppendixA
	
	\section{Details on  the application to the Chilean manufacturing data}
	\label{App:application}
    In section~\ref{sec:6.application}, we applied the SCKLS estimator to the Chilean manufacturing data to estimate a production function for plastic (2520) and wood (2010) industries. Here we provide the detailed specification of the SCKLS estimator applied to the real data. Since the application data is skewed as shown in Table \ref{tab:11.sumstat}, we use non-uniform grid of evaluation points and limit evaluation points to be inside the convex hull of $\{\bm{X}_j\}_{j=1}^n$. Figure~\ref{fig:eval_plastic} and Figure~\ref{fig:eval_wood} show how we set the evaluation points in our application. Originally we set the number of evaluation points is $m=400$, but after deleting ones which lie outside of the convex hull of $\{\bm{X}_j\}_{j=1}^n$, the number is $m\approx 270$ for both industries.

	\begin{figure}
		\centering
		\subfloat[Before deletion]{\includegraphics[width=0.5\textwidth]{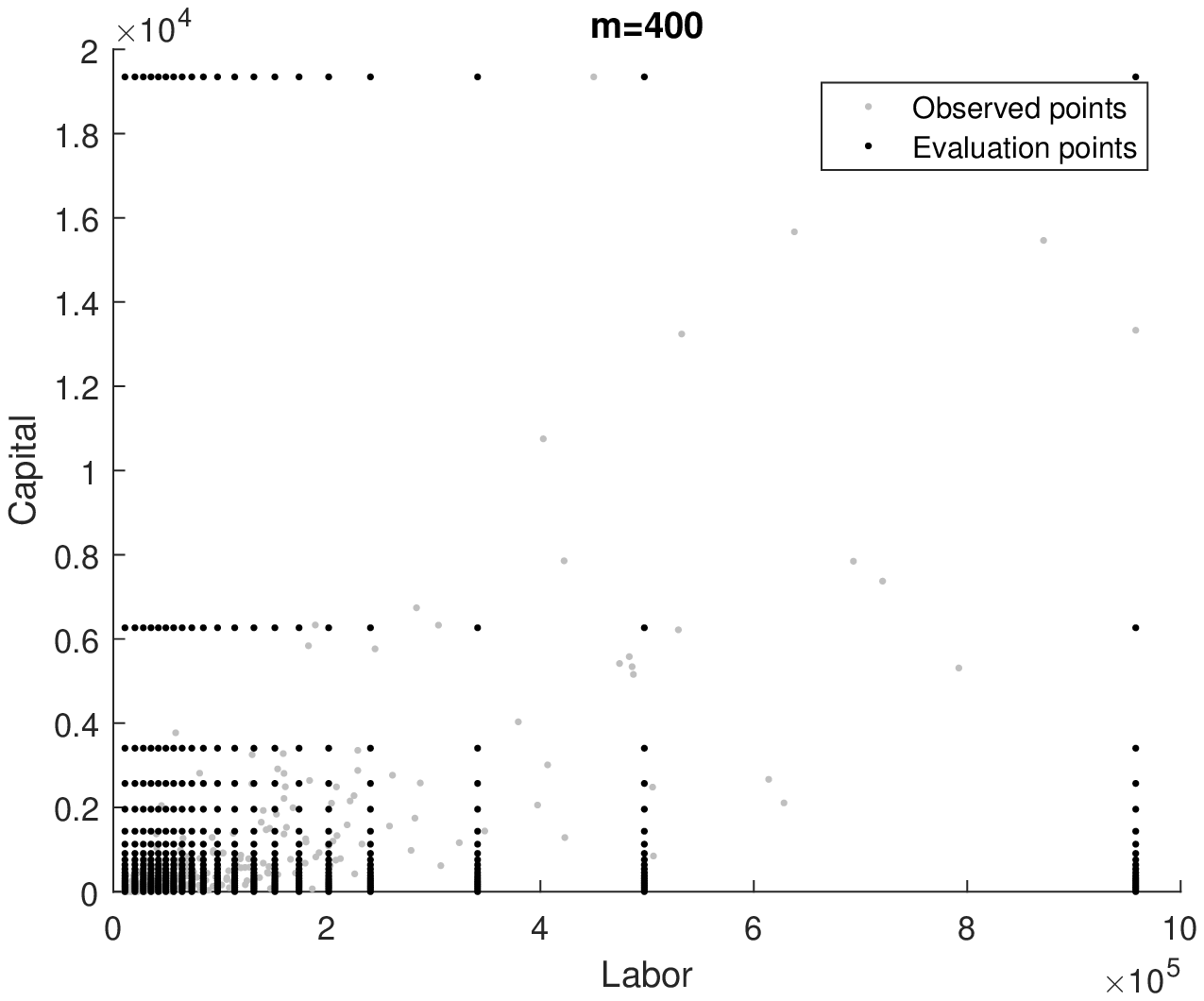}}
		\hfill
		\subfloat[After deletion]{\includegraphics[width=0.5\textwidth]{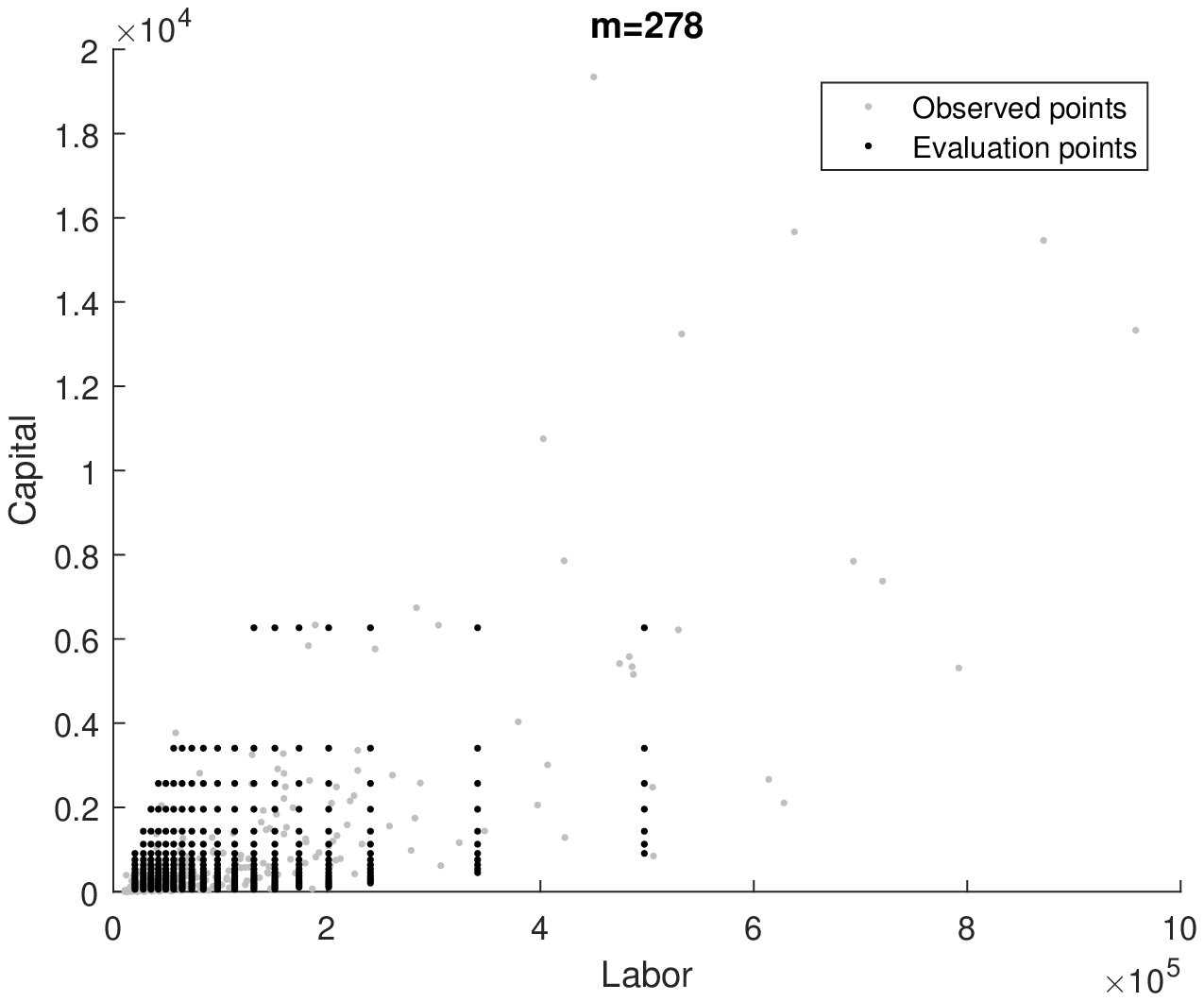}}
		\caption{Proposed evaluation points with Plastic industry (2520)}
		\label{fig:eval_plastic}
	\end{figure}
	\begin{figure}
		\centering
		\subfloat[Before deletion]{\includegraphics[width=0.5\textwidth]{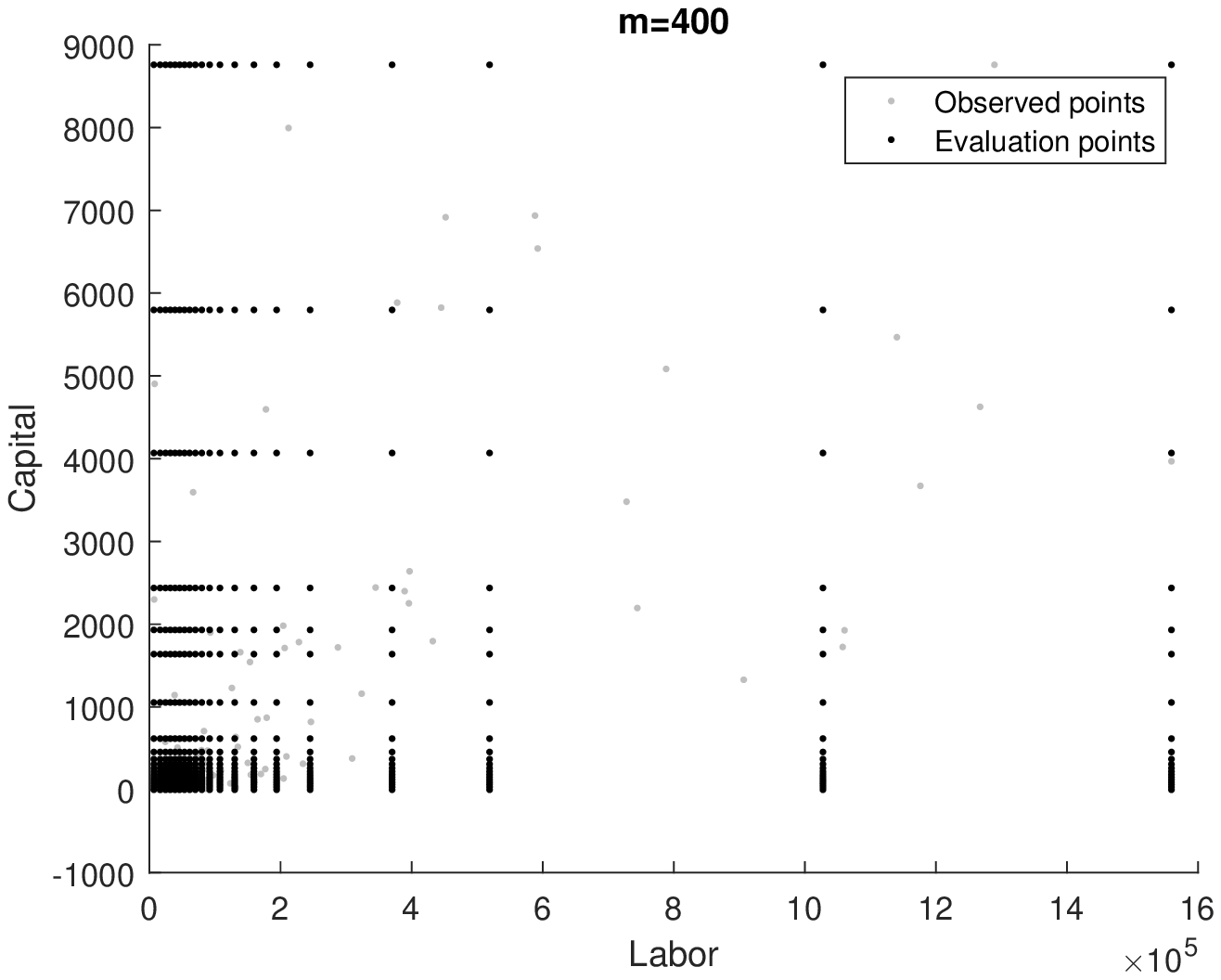}}
		\hfill
		\subfloat[After deletion]{\includegraphics[width=0.5\textwidth]{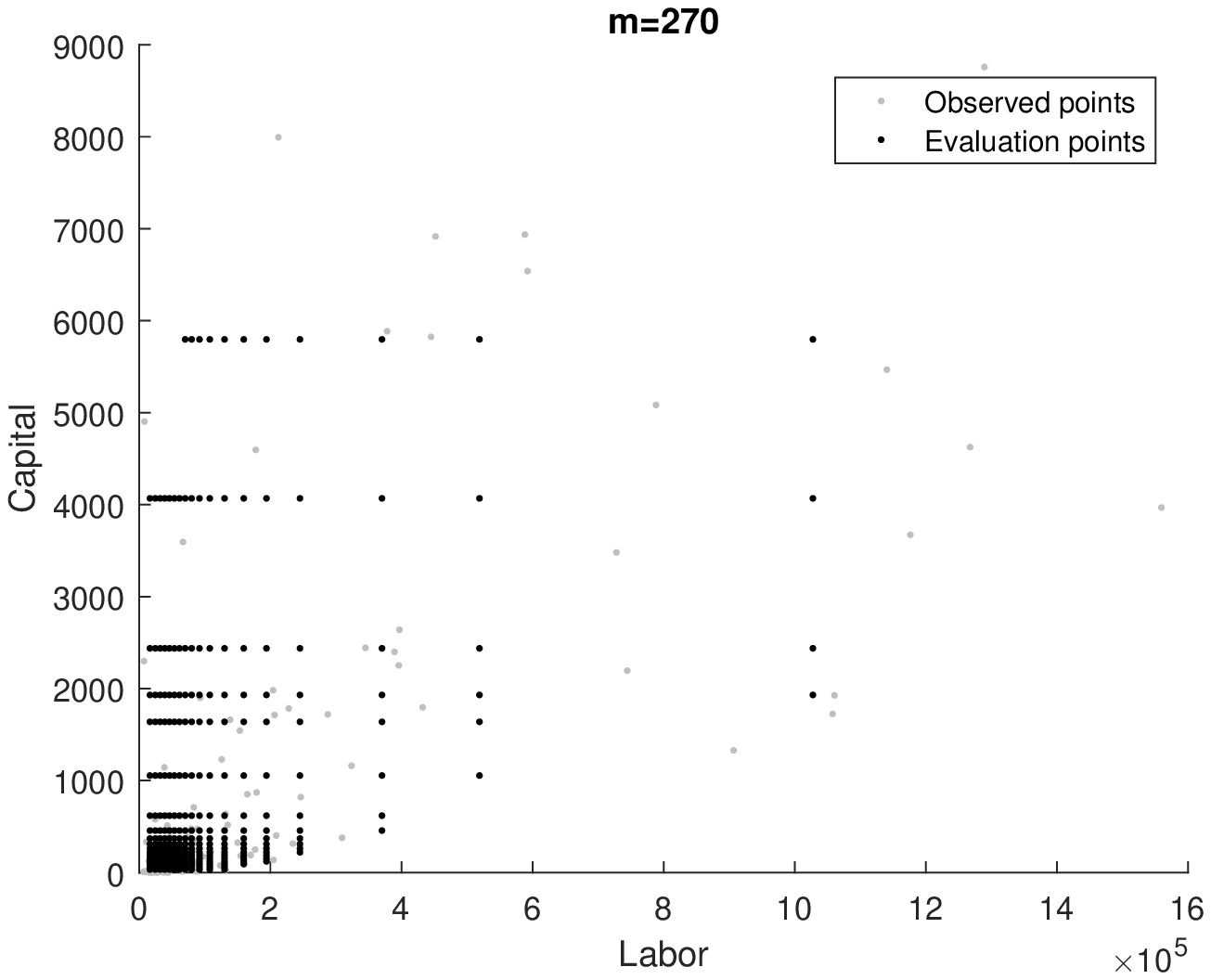}}
		\caption{Proposed evaluation points with Wood industry (2010)}
		\label{fig:eval_wood}
	\end{figure}

	\bibliographystyle{chicago}
	%\small
	\setlength{\bibsep}{0.0pt} 
	\bibliography{SCKLS_revision}
\end{document}